\documentclass[reprint,longbibliography,aps,prb,table]{revtex4-2}

\usepackage[english]{babel}
\usepackage{amsmath,amsthm,amssymb,bbm,amsfonts,graphicx,framed}
\usepackage{url}
\usepackage{tikz}
\usetikzlibrary{decorations.pathreplacing,decorations.pathmorphing}
\usepackage{enumerate}
\usepackage{mathrsfs}
\usepackage{relsize}
\usepackage[colorlinks=true, allcolors=blue]{hyperref}
\usepackage[utf8]{inputenc}
\usepackage{listings}
\usepackage{stmaryrd}
\usepackage[braket, qm]{qcircuit}
\usepackage{graphicx}
\usepackage{longtable}

\usepackage{bbold}

\lstdefinestyle{mystyle}{
  backgroundcolor=\color{backcolour},
  commentstyle=\color{codegreen},
  keywordstyle=\color{magenta},
  numberstyle=\tiny\color{codegray},
  stringstyle=\color{codepurple},
  basicstyle=\ttfamily\footnotesize,
  breakatwhitespace=false,
  breaklines=true,
  captionpos=b,
  keepspaces=true,
  numbers=left,
  numbersep=5pt,
  showspaces=false,
  showstringspaces=false,
  showtabs=false,
  tabsize=2
}
\newcommand{\lin}{\lstinline}

\def\bS{{\bf S}}

\newcommand{\Om}{\Omega}

\newcommand{\la}{\lambda}

\newcommand{\al}{\alpha}

\newcommand{\si}{{\sigma}}

\newcommand{\fSim}{{\mathrm{fSim}}}
\newcommand{\ii}{\mathrm{i}}
\newcommand{\ee}{\mathrm{e}}

\newcommand{\eq}[1]{Eq.~(\ref{eq:#1})}

\newcommand{\fig}[1]{Fig.~\ref{fig:#1}}
\renewcommand{\sec }[1]{Sec.~\ref{sec:#1}}

\newcommand{\tab}[1]{Table~\ref{tab:#1}}
\newcommand{\thm}[1]{Theorem~\ref{thm:#1}}

\newcommand{\SWAP}{\mathrm{SWAP}}
\newcommand{\CNOT}{\mathrm{CNOT}}

\newcommand{\mc}[1]{\mathcal{#1}}
\newcommand{\id}{\mathbb{1}}

\newcommand{\sfrac}[2]{{\textstyle\frac{#1}{#2}}}

\newcommand{\f}[2]{{\frac{#1}{#2}}}

\newcommand{\nn}{\nonumber}

\newcommand{\ketbra}[2]{\ket{#1}\!\!\bra{#2}}

\let\perptmp\perp
\renewcommand{\perp}{{\! \mathsmaller{\perptmp}}}

\newcommand{\e}{\mathlarger{e}}

\newcommand{\mrm}{\mathrm}

\newcommand{\nodagger}{{\phantom{\dagger}}}

\newtheorem{theorem}{Theorem}

\theoremstyle{definition}

\makeatletter
\def\@bibdataout@aps{%
  \immediate\write\@bibdataout{%
    @CONTROL{%
      apsrev41Control%
      \longbibliography@sw{%
        ,author="08",editor="1",pages="1",title="0",year="1"%
      }{%
        ,author="08",editor="1",pages="1",title="",year="1"%
      }%
    }%
  }%
  \if@filesw \immediate \write \@auxout {\string \citation {apsrev41Control}}\fi
}
\makeatother

\definecolor{blue}{rgb}{0.12156862745098039, 0.4666666666666667, 0.7058823529411765}
\definecolor{orange}{rgb}{1.0, 0.4980392156862745, 0.054901960784313725}
\definecolor{green}{rgb}{0.17254901960784313, 0.6274509803921569, 0.17254901960784313}
\definecolor{red}{rgb}{0.8392156862745098, 0.15294117647058825, 0.1568627450980392}
\definecolor{purple}{rgb}{0.5803921568627451, 0.403921568627451, 0.7411764705882353}
\definecolor{brown}{rgb}{0.5490196078431373, 0.33725490196078434, 0.29411764705882354}
\definecolor{pink}{rgb}{0.8901960784313725, 0.4666666666666667, 0.7607843137254902}
\definecolor{pygray}{rgb}{0.4980392156862745, 0.4980392156862745, 0.4980392156862745}
\definecolor{olive}{rgb}{0.7372549019607844, 0.7411764705882353, 0.13333333333333333}
\definecolor{9}{rgb}{0.09019607843137255, 0.7450980392156863, 0.8117647058823529}

\definecolor{lightblue}{rgb}{0.9,.95,1}

\newcommand{\hdimer}{
\!\!\raisebox{-.25em}{
  \begin{tikzpicture}[scale=.25,line width=0.75pt]
    \useasboundingbox (-0.2,-0.2) rectangle (1.2,1.2);
    \coordinate (0) at (0,0);
    \coordinate (1) at (1,0);
    \coordinate (2) at (1,1);
    \coordinate (3) at (0,1);

    \fill (0) circle (4pt);
    \fill (1) circle (4pt);
    \fill (2) circle (4pt);
    \fill (3) circle (4pt);

    \draw (0) -- (1);
    \draw (2) -- (3);
\end{tikzpicture}}
}
\newcommand{\vdimer}{
\!\!\raisebox{-.25em}{
  \begin{tikzpicture}[scale=.25,line width=0.75pt]
    \useasboundingbox (-0.2,-0.2) rectangle (1.2,1.2);

    \coordinate (0) at (0,0);
    \coordinate (1) at (1,0);
    \coordinate (2) at (1,1);
    \coordinate (3) at (0,1);

    \fill (0) circle (4pt);
    \fill (1) circle (4pt);
    \fill (2) circle (4pt);
    \fill (3) circle (4pt);

    \draw (0) -- (3);
    \draw (1) -- (2);
\end{tikzpicture}}
}

\newcommand{\singledimer}{
\!\!\raisebox{-.25em}{
  \begin{tikzpicture}[scale=.25,line width=0.75pt]
    \useasboundingbox (-0.2,-0.2) rectangle (1.2,1.2);

    \coordinate (0) at (0,0.5);
    \coordinate (1) at (1,0.5);

    \fill (0) circle (4pt);
    \fill (1) circle (4pt);

    \draw (0) -- (1);
\end{tikzpicture}}
}

\newcommand{\xlink}{
  \makebox[0pt]{
  \begin{tikzpicture}[scale=.25,line width=0.75pt]
    \coordinate (A) at (0,0);
    \coordinate (B) at (30:1);
    \fill (A) circle (4pt);
    \fill (B) circle (4pt);
    \draw (A) -- (B);
  \end{tikzpicture}
}
}

\newcommand{\ylink}{
  \makebox[0pt]{
  \begin{tikzpicture}[scale=.25,line width=0.75pt]
    \coordinate (A) at (0,0);
    \coordinate (B) at (120:1);
    \fill (A) circle (4pt);
    \fill (B) circle (4pt);
    \draw (A) -- (B);
\end{tikzpicture}}
}

\newcommand{\zlink}{
  \makebox[0pt]{
  \begin{tikzpicture}[scale=.25,line width=0.75pt]
    \coordinate (A) at (0,0);
    \coordinate (B) at (-90:1);
    \fill (A) circle (4pt);
    \fill (B) circle (4pt);
    \draw (A) -- (B);
\end{tikzpicture}}
}

\newcommand{\diag}{\mathrm{diag}}

\usepackage{calc}

\newcommand{\placefig}[2]{
  \parbox{
    \widthof{
      \includegraphics[scale=.15]{#1}
    }
  }{
    \includegraphics[scale=.15]{#1} \\ \centering #2
  }
}

\usepackage{pgffor}

\begin{document}

\title{Efficient quantum simulation for translationally invariant systems}

\author{Joris Kattem\"olle}
\author{Guido Burkard}
\affiliation{Department of Physics, University of Konstanz, D-78457 Konstanz, Germany}

\begin{abstract}
Discrete translational symmetry plays a fundamental role in condensed matter physics and lattice gauge theories, enabling the analysis of systems that would otherwise be intractable. Despite this, many open problems remain. Quantum simulation promises to offer new insights, but progress is often limited by device connectivity constraints, which lead to prohibitively long computation times. We extend the use of spatial symmetry from the systems to be simulated to the quantum circuits simulating them. One application is that it becomes possible to efficiently and optimally alleviate device connectivity constraints algorithmically. This leads to reductions in quantum computational time by several orders of magnitude even for moderate system sizes, making such simulations feasible, with even greater relative gains for larger systems. This substantially enhances the capabilities of quantum computers in the simulation of condensed matter systems and lattice gauge theories, even before hardware improvements. Our work forms the basis for using spatial symmetry of quantum circuits in other areas of quantum computation, such as in the design and implementation of quantum error correcting codes.
\end{abstract}

\maketitle
Quantum simulation~\cite{feynman1982simulating,lloyd1996universal} remains one of the most promising applications of quantum computers. Recent experiments have demonstrated encouraging results in simulating quantum magnetism~\cite{diepen2021quantum, frey2022realization,youngseok2023evidence, rosenberg2024dynamics, mi2024stable, burkard2025recipes}. Other systems that can be simulated in principle include the Fermi-Hubbard model, which provides insights into high-temperature superconductivity~\cite{hensgens2017quantum,dagotto1994correlated,cade2020strategies}, and lattice gauge theories such as the Rokhsar-Kivelson model~\cite{rokhsar1988superconductivity} and the Kogut-Susskind model of 2D quantum electrodynamics (QED)~\cite{kogut1975hamiltonian,haase2021resource,paulison2021simulating,meth2023simulating}, where the phase diagram and real-time dynamics remain only partially understood~\cite{pichler2016realtime, felser2020twodimensional}. These systems possess discrete translational symmetry, and this symmetry is inherited by circuits for the quantum simulation of those systems. In this work, we call such circuits \emph{tileable} since they consist of a circuit tile or motif that is repeated spatially and optionally temporally.

A quantum simulation circuit typically comprises a sequence of one- and two-qubit unitaries acting on spin-1/2 degrees of freedom (qubits). In many architectures, such as superconducting or quantum dot devices, two-qubit gates can only be applied to specific pairs determined by the device’s coupling graph~\cite{qiskit2024quantum,cirq2024cirq} (also called the coupling map~\cite{qiskit2024quantum} or device graph~\cite{cirq2024cirq}). Examples include the square grid~\cite{acharya2025quantum} and the heavy-hexagonal lattice~\cite{youngseok2023evidence}.

\begin{figure*}[t]
\includegraphics[width=\linewidth]{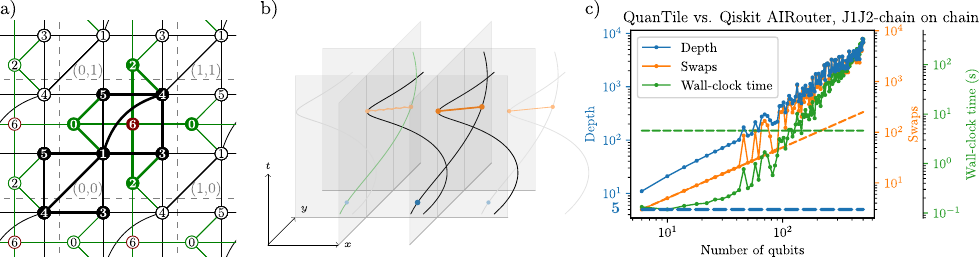}
\caption{\label{fig:figs}
    \textbf{(a)} Basis circuit (bold edges) for the quantum simulation of the Kogut-Susskind model of 2D QED. Dashed lines delineate the circuit cells, with coordinates in gray, and seed numbers within the vertices. The Jordan-Wigner transformation breaks 2D translational symmetry, which we circumvent by using the compact encoding~\cite{derby2021compact, clinton2024towards}, leading to the red auxiliary qubits.
    \textbf{(b)} Qudit mobility zone. The logical qudits of the basis circuit (black worldlines) are free to move by the action of SWAPs (omitted) along the edges in a patch of $2\delta+1$ by $2\delta+1$ hardware cells (here, $\delta=1$ and edges are likewise omitted). The logical qudits are acted upon by single-qudit (dark blue dot) and two-qudit (dark orange edge) physical gates. In this example, the physical circuit (black lines, dark blue dots, and dark orange edges) is repeated to form a $3\times1$ physical circuit patch. These physical basis circuits interact; a two-qudit gate from a translated copy of the physical basis circuit (wavy orange line) acts between a logical qudit from the original basis circuit and a logical qudit from a translated copy (green line).
    \textbf{(c)} Benchmarking results comparing our method's implementation (``QuanTile", dashed lines) \cite{kattemolle2024quantile} with Qiskit's AIRouter \cite{qiskit2024quantum} (solid lines).
  }
\end{figure*}

To overcome restricted qubit connectivity, SWAP gates (SWAPs) are commonly inserted~\cite{childs2019circuit,cowtan2019qubit}. A SWAP gate exchanges the states of two qubits, effectively swapping their positions. To apply a two-qubit gate to a pair of remote qubits, first these qubits must be routed along the edges of the coupling graph using SWAPs until they become adjacent. The \emph{routing problem}~\cite{cowtan2019qubit}, also called quantum layout synthesis~\cite{tan2020optimal,ping2025assessing}, is the task of determining the initial positions of the qubits and the placement of the SWAPs. An optimal solution minimizes routing overhead, commonly measured by the routed circuit's depth or the number of inserted SWAPs.

Optimal qubit routing is crucial. First, on today’s quantum computers, routing overhead increases the effects of noise, currently limiting simulations to models of quantum magnetism that can be directly embedded into the hardware coupling graph~\cite{frey2022realization, youngseok2023evidence, rosenberg2024dynamics, mi2024stable}. Second, on early fault-tolerant quantum computers,   minimizing routing overhead maximizes the system sizes and simulation times that can be achieved. Finally, while scalable fault-tolerant quantum computers can reliably execute circuits of arbitrary depth, the routing overhead directly increases the monetary cost of running a circuit.

However, solving the routing problem optimally is NP hard~\cite{childs2019circuit, maslov2008quantum, cowtan2019qubit}, making optimal methods~\cite{tan2020optimal, nannicini2022optimal} slow and limited in the circuit sizes they can handle. Using highly optimized implementations~\cite{lin2023scalable}, the performance of these methods can be improved, but given the problem's computational complexity, it is unlikely that optimal routing can be achieved for circuits acting on more than 100 qubits---roughly the number of qubits already available on today’s quantum computers and the scale anticipated to be necessary for a quantum advantage~\cite{childs2018toward, zimboras2025myths}. Consequently, efficient but suboptimal heuristic methods have been developed~\cite{li2019tackling,cowtan2019qubit,kremer2024practical}, with SWAP counts sometimes exceeding the optimal by several orders of magnitude~\cite{ping2025assessing}.

In this work, we introduce a framework for analyzing  tileable circuits. It enables the modification of existing routing methods so that they only need to solve the routing problem for a single circuit tile,
while ensuring that the routed circuit remains tileable. We provide an implementation of our method~\cite{kattemolle2024quantile}, in which we modify the optimal satisfiability modulo theories (SMT) based routing method from Ref.~\cite{tan2020optimal}, with extensive verification of the correctness and performance of our method and implementation.  We use it to find routing solutions with negligible overhead for circuits derived from quantum cellular automata~\cite{farrelly2020review,gopalakrishnan2018facilitated, klobas2021exact}, condensed matter simulations, and lattice gauge theory simulations. This unlocks the potential of today's quantum computers to study frustrated magnetism and reduces computational cost on fault-tolerant devices by several orders of magnitude. 

Although finding an (optimal) solution to the routing problem for a single tile is still NP hard, our method allows it to be repeated spatiotemporally at essentially no computational cost. The complexity of solving the routing problem now only depends on local properties of the simulated system and becomes essentially independent of the system size and the simulation time. Benchmarking our specialized method's implementation~\cite{kattemolle2024quantile} demonstrates that also in practice it greatly outperforms leading general-purpose routing software. For instance, on a 500-qubit circuit, our method reduces the depth overhead by more than 3 orders of magnitude while also requiring less computational time.

\emph{Fundamentals---}We define a circuit tile formally as a \emph{basis circuit} (\fig{figs}a). It consists of a finite set of gates $g$. Since the routing problem is independent of whether the degrees of freedom are qubits or qudits with more than two basis states, we use the term ``qudits'' for generality. A single-qudit (two-qudit) gate $g$ acts on qudit(s) $g.\tilde q$ ($g.\tilde q$, $g.\tilde q'$).  We assign each gate a circuit layer $g.\tilde t\in \mathbb N^0$. Each qudit $q$ has cell coordinates $q.\tilde x, q.\tilde y\in\mathbb{Z}$ and a seed number $q.\tilde s\in \mathbb N^0$, which distinguishes qudits within a cell. This generalizes to higher dimensions by adding spatial coordinates.

A spatiotemporal \emph{circuit patch} $P_{n,m,l}(C)$ is formed by merging $n\times m\times l$ translated copies of a depth-$d$ basis circuit $C$,
\begin{equation}\label{eq:patch}
P_{n,m,l}(C) = \bigcup_{\substack{(\Delta x,\Delta y, \Delta t)\\ \in\,\mathbb{Z}_n\times \mathbb{Z}_m \times d\,\mathbb{Z}_l}} T_{\Delta x, \Delta y, \Delta t}(C).
\end{equation}
Here, $T_{\Delta x,\Delta y,\Delta t}(C)$ translates each gate's time coordinate by $\Delta t$ and each qudit's spatial coordinates by $\Delta x$ and $\Delta y$, and  $d\,\mathbb{Z}_l=\{0,d,2d,\dots,(l-1)d\}$. Replacing $\mathbb{Z}_n$ and $\mathbb{Z}_m$ with $\mathbb{Z}$ yields spatially infinite, or \emph{lattice}, circuits. For formal reasons, the union is taken as a multiset union (allowing duplicates). We call circuits in which a basis graph can be identified \emph{tileable circuits}.

A \emph{gate collision} occurs when two gates simultaneously act on the same qudit. A  basis circuit $C$ is valid if the circuit patch $P_{n,m,l}(C)$ (or equivalently, its induced lattice circuit) is collision-free for all $n,m,l\in\mathbb N^+$, which can be checked straightforwardly due to the following theorem. 

\textbf{Theorem 1.} Let $C$ be a basis circuit. The circuit patch $P_{n,m,l}(C)$ is collision-free for all $n,m,l\in\mathbb{N}^+$ if and only if, for each time $t$ and seed number $s$, at most one qudit $q$ with $q.\tilde{s}=s$ is acted on by a gate in layer $t$. A proof is provided in the Supplemental Material (SM)~\footnote{The Supplemental Material is appended. It includes the proof of Theorem 1, the explicit SMT formula for tileable routing, a detailed description of our implementation \cite{kattemolle2024quantile}, and additional results. All basis graphs and tileable quantum circuits used in this work are given explicitly.}.

A \emph{basis graph} \cite{kattemolle2024edge} is a basis circuit with edges instead of gates. Edges lack a time coordinate, and two edges may meet at the same vertex. Merging infinitely many translated copies of a basis graph produces a \emph{lattice graph}, formalizing the concept of an ``infinite lattice'' with edges. Examples include the square, honeycomb, and kagome lattices, with edges connecting nearest neighbors.

Tileable circuits naturally arise in quantum simulations of lattice systems via Trotterization. Consider a two-local Hamiltonian $H=\sum_i H_i+\sum_{(i,j)}H_{ij}$,
acting on $\delta$-dimensional local degrees of freedom (i.e., $\delta$-level qudits), where $H_i$ acts on qudit $i$ and $H_{(i,j)}$ on qudits $i$ and $j$ along the edges of a lattice graph $G$. By Trotterization~\cite{lloyd1996universal,childs2021theory, suzuki1991general}, $\ee^{-\ii \tau H}\approx [\tilde U_p(\tau/r)]^r$,
with $\tilde U_p(\tau/r)$ a $p$th-order Trotter step and $r$ the number of Trotter steps. A first-order step has the form $\tilde U_1(\tau)=\ee^{-\ii \tau H^{(1)}}\cdots \ee^{-\ii \tau H^{(\Gamma)}}$, for some sequence $H^{(1)},\ldots,H^{(\Gamma)}$ of the terms in $H$. Here, $H^{(i)}$ is the $i$th term in the sequence, and each $\ee^{-\ii \tau H^{(i)}}$ is a two-qudit gate. A second-order Trotter step is constructed by applying the first-order step twice, reversing the order in the second application,
$\tilde U_2(\tau)=\tilde U_1^\dagger(-\tau/2)\,\tilde U_1(\tau/2)$. Using the bounds from~\cite{childs2021theory}, the error is $O\bigl[N\tau\,(\tau/r)^p\bigr]$, where $N$ is the number of qudits.

Importantly, the Trotter step $\tilde U_p$ is tileable for any $p$. For the routing problem, only the structure of $G$ matters, so we consider arbitrary two-local (ATL) Hamiltonians, including models such as the Ising, (an)isotropic Heisenberg, Kitaev, and Bose-Hubbard models~\cite{Note1}. Trotterization does not prescribe a specific order $H^{(i)}$. Finding the order that minimizes circuit depth is equivalent to solving an NP-hard edge-coloring problem \cite{kattemolle2024edge}. Tileable circuits also arise naturally in the task of finding ground states of lattice systems, for example, because the implementation of $\ee^{-\ii \tau H}$ is a core subroutine \cite{kitaev1995quantum}, or because $\tilde U_1$  determines the structure of a parameterized ansatz circuit \cite{wecker2015progress,kattemolle2022variational}.

\emph{Routing---}It is common to call the circuit to be routed the \emph{logical circuit}, acting with \emph{logical gates} on \emph{logical qudits}. The routed circuit is the \emph{physical circuit}, acting with \emph{physical gates} on \emph{physical qudits}, while respecting hardware connectivity constraints. Although qudit routing is also needed for fault-tolerant quantum computing, we note the terms \emph{logical} and \emph{physical} here differ from their common usage in error correction. To physically execute a logical circuit, each logical qudit $q$ is assigned to a physical qudit $Q$ via a qudit map $Q_q$. We say that the logical qudit $q$ resides at the physical qudit $Q_q$. If $q$ starts at $Q_q$ and a SWAP acts on $Q_q$ and $Q'$, then afterward $q$ resides at $Q'$. The qudit map has to be updated accordingly.

The routing problem for standard circuits was formulated as an SMT formula in~\cite{tan2020optimal}. An SMT formula consists of variables and constraints. Here, Boolean and integer variables encode the qudit map, the coordinates of the physical qudits of gates, and the time coordinates of physical gates. Constraints on those variables ensure that the corresponding physical circuit is collision-free, respects the hardware connectivity, implements the same unitary as the logical circuit (up to a final reordering of the qudits), and has a preset depth $D>0$. The smallest $D$ for which all constraints become satisfiable yields the routed circuit with minimal depth overhead.

The second-order Trotter formula (and higher-order versions~\cite{Note1}) offers an advantage in qudit routing. Let $U(\tau)$ be the routed version of $\tilde U_1(\tau)$. At the end of $U(\tau)$, logical qudits may not return to their initial positions, but after $U^\dagger(-\tau)U(\tau)$, they do. Thus, $\ee^{-i\tau H} \approx [\tilde U_2(\tau/r)]^r = [U^\dagger(-\tau/(2r))U(-\tau/(2r))]^r$, where the right-hand side is a fully routed circuit. Thus, it suffices to route a single first-order step without enforcing logical qudits to return to their initial position. This temporal symmetry  can be exploited alongside spatial symmetry. Alternatively, one can enforce qudits to return to their initial positions at the end of $U(\tau)$, though this may increase circuit depth.

When simulating Hamiltonians with two-qubit interactions $H_{ij}(\Delta) \sim X_iX_j+Y_iY_j+\Delta Z_iZ_j$ on hardware where the CNOT is native, SWAPs can be absorbed into directly preceding or following two-qubit simulation unitaries $\ee^{-\ii \tau H^{(i)}}$ at no increase of the infidelity of the subcircuits implementing those unitaries~\cite{vidal2004universal,Note1}. The same applies when the two-qubit fSim gate~\cite{foxen2020demonstrating,arute2019quantum,Note1} is native.

\begin{table}
\addtolength{\tabcolsep}{-.5em}
\rowcolors{5}{}{gray!20}
\begin{longtable}{cccc}
\hline
Simulated                   & Hardware          & Depth       & SWAP    \\
system              & coupling graph    & overhead      & overhead  \\
\hline
ATL ladder          	& chain          	& 0 (0 \%)  	& 0       \\ 
ATL J1J2-ladder     	& chain          	& 0 (0 \%)  	& 0       \\ 
ATL J1J2-chain       	& chain          	& 1 (25 \%) 	& 0       \\ 
\hline
Rule 54               	& ladder        	& 0 (0 \%)  	& 1 (4)   \\ 
ATL J1J2-chain     	& ladder        	& 0 (0 \%)  	& 0       \\ 
\hline
ATL J1J2-square     	& square grid       	& 0 (0 \%)  	& 0       \\ 
ATL triangular     	& square grid         	& 0 (0 \%)  	& 0       \\ 
ATL kagome         	& square grid       	& 1 (25 \%) 	& 0       \\ 
ATL shuriken       	& square grid       	& 1 (25 \%) 	& 0       \\ 
ATL snub-square    	& square grid       	& 1 (20 \%) 	& 0       \\ 

Rokhsar-Kivelson      	& square grid       	& 2 (11 \%) 	& 2 (4)   \\ 
Fermi-Hubbard        	& square grid       	& 0 (0 \%)  	& 0       \\ 
Kogut-Susskind       	& square grid         & 9 (4 \%)  	& 253 (6) \\ 
 \end{longtable}
 \caption{\label{tab:results}
   Excerpt of the results obtained by our method \cite{kattemolle2024quantile} in routing a single Trotter step for various lattice systems (first column) to various hardware coupling graphs (second column). ATL J1J2-$\ell$ denotes an arbitrary two-local model (ATL) on a lattice $\ell$ with edges between nearest and next-nearest neighbors. The reported depth overhead (third column) is given per Trotter step and remains valid for any number of Trotter steps and logical circuit patch size. The SWAP overhead (last column) is per Trotter step and per the number of qudits indicated in parentheses. In all listed solutions, the qudit overhead is zero, except for Kogut-Susskind, where it is one qudit per seven logical qudits.}
 \end{table}

\emph{Routing tileable circuits---}To leverage spatial symmetry, we map tileable circuits to hardware whose connectivity graph is described by (a patch of) a lattice graph. The unitary implemented by any $n\times m\times l$ physical circuit patch must equal the unitary implemented by the corresponding $n\times m\times l$ logical circuit patch (up to a final reordering of qudits). Since this holds for any $n$ and $m$, we may even consider infinite physical circuit patches, or \emph{physical lattice circuits}.

In the mathematical formulation of the qudit routing problem for tileable circuits, several concepts and assumptions beyond those required for standard routing are essential. First, when placing any gate into the physical basis circuit, any physical circuit patch induced by the physical basis circuit [\eq{patch}] must be collision-free. This is done straightforwardly by using Theorem~1, demonstrating its necessity and effectiveness. Second, due to the discrete spatial translational symmetry of a physical lattice circuit, the qudit map must share the same symmetry. We assume that the initial qudit map conserves cell coordinates. Translational symmetry then requires that for each seed number $s$, there is a unique seed number $S$ such that all logical qudits with seed number $s$ are mapped to physical qudits with seed number $S$. Logical qudits can move along the edges of the hardware lattice graph via SWAPs, but we assume they cannot move outside the \emph{mobility zone} (\fig{figs}b), which is a patch of $2\delta+1$ by $2\delta+1$ hardware cells centered around the central cell. The mobility zone can be made arbitrarily large at the cost of increased computational resources. 

Finally, consider inserting a SWAP gate $g$ acting on physical qudits $(Q,Q')$ into the physical basis circuit $C$. By \eq{patch}, the physical lattice circuit induced by $C$ will include SWAPs along $T_{\Delta x,\Delta y,0}[g]$ for every $\Delta x,\Delta y$. Crucially, these translated SWAPs, although they are not in $C$ itself, may act on physical qudits of $C$ holding logical qudits. We account for this effect by temporarily inserting all relevant possible translated versions of each SWAP gate that is added to the physical basis circuit. This renders the basis circuit technically invalid by Theorem 1. In postprocessing, however, we retain only the untranslated SWAP gate, ensuring that the physical lattice circuit is collision-free. The translated SWAPs acting on $C$ automatically reemerge in the final physical lattice circuit. For finite physical circuit patches, boundary effects occur, which are dealt with in detail in the SM~\cite{Note1}.

%
%


\emph{Implementation---}
Until now, our discussion has been conceptual and not related to the implementation. For demonstration, we reformulate the above concepts as an SMT formula, as it was done in \cite{tan2020optimal} for general, structureless circuits, and implement this SMT formulation in code. In the End Matter, we give an example of a constraint unique to the routing of tileable circuits. The full details of our SMT formulation and implementation can be found in the SM~\cite{Note1}. Our framework is modular and the SMT-based method can be replaced by other routing methods, including heuristic methods. 

The implementation provides the following options that, like tileability, go beyond the capabilities of standard routing methods: (i) optimize the order of two-qudit gates in the Trotter step; (ii) allow SWAPs to merge with directly preceding or following two-qudit gates; (iii) enforce that logical qudits return to their initial positions at the end of the circuit (``cyclic routing''); and (iv) slice the logical circuit into subcircuits to route sequentially, reducing the computational complexity to linear in the logical circuit depth.

\emph{Results---}We applied our method to basis circuits for one Trotter step in the quantum simulation of: ATL Hamiltonians on 24 different 1D and 2D lattices, the Fermi-Hubbard model on the square lattice~\cite{dagotto1994correlated}, the Rokhsar-Kivelson model~\cite{rokhsar1988superconductivity}, and the Kogut-Susskind model of 2D QED~\cite{kogut1975hamiltonian}. Finally, we applied our method to the basis circuit of the Rule 54 quantum cellular automaton~\cite{gopalakrishnan2018facilitated, klobas2021exact}. Using \eq{patch}, the solutions can be tiled spatially and temporally, creating arbitrarily deep and wide routed circuits. The routing problems were solved for various combinations of the routing options (i--iv). 

All circuits and the 1D and 2D lattices are provided explicitly in the SM, along with extensive results~\cite{Note1}. An excerpt of the results is given in \tab{results}, with the exact routing options used given in the End Matter. For the solutions we found, the routing overhead is remarkably low. In some cases there is no overhead at all, possible under option (i). For the zero-overhead cases, we note that increasing $\delta$ cannot decrease the depth further. For the remaining low-overhead cases, we expect limited to no further improvements from increasing $\delta$ because of the locality of the input circuits. 

Our method offers a scaling advantage over general-purpose methods because its running time is essentially independent of the logical circuit size. Nevertheless, the question remains if this advantage is already significant for today’s quantum chips, with on the order of 100 to 1000 qubits. We therefore benchmarked our method against multiple established routing methods across various routing problems~\cite{Note1}. Here, we show the results of comparing our method to Qiskit's leading AIRouter~\cite{kremer2024practical} on the problem of routing a single Trotter step for the simulation of an ATL J1J2 model on a chain to hardware with chain connectivity (\fig{figs}c). Since the AIRouter does not optimize gate order (i), we allow our implementation to perform this optimization and then use the resulting gate order as a fixed input for the AIRouter. Because the AIRouter does not support SWAP merging (ii) or cyclic routing (iii), we disable these options in our implementation. As both methods route a single Trotter step, their solutions can be repeated temporally to construct arbitrarily deep second-order Trotter circuits, but only our solution can also be repeated spatially. Our method's implementation becomes faster for system sizes above approximately 100 qudits, while also producing solutions with significantly lower depth and SWAP overhead. At around 500 qudits, the depth overhead decreases from approximately $10^4$ to just 5.

\emph{Conclusion---}We have demonstrated that for circuits with discrete spatiotemporal translational symmetry, naturally arising in the quantum simulation of condensed matter systems and lattice gauge theories, inherently scalable qubit routing solutions can be achieved with negligible overhead. Beyond a reduction of costs on fault-tolerant devices by several orders of magnitude, this enables the simulation of geometrically frustrated magnetism on current devices. One possibility is the observation of disorder-free localization and many-body quantum scars in a Heisenberg model on the kagome lattice \cite{mcclarty2020disorder-free,lee2020exact}. The circuits require $O(1)$-depth state preparation, followed by the Trotterized simulation circuit, which for this system is possible on square-grid hardware using five fSim gate layers per Trotter step (\tab{results}). With 100 qubits and a depth-100 circuit, which could be within reach of pre-fault-tolerant devices \cite{zimboras2025myths}, it becomes possible to simulate around 11 second-order (or 17 first-order \cite{Note1}) Trotter steps on 100 qubits. 

Using the techniques from~\cite{lin2023scalable}, an optimization of our implementation is expected to decrease its own (classical) running time by several orders of magnitude. Also, while our focus has been on an optimal method, within our framework, heuristic methods can also be adapted so that their solutions become tileable. This could be valuable when the basis circuits become too large for optimal methods, e.g., when routing circuits for future modular and tileable hardware, where each module consists of hundreds of qubits~\cite{bravyi2022future}.

We have laid the foundation for addressing compilation tasks other than qubit routing in the tileable setting. Examples include leveraging gate identities to reduce logical circuit depth~\cite{maslov2008quantum_circuit_simplification}, compilation to error-correction native gates \cite{ryan-anderson2024high-fidelity}, automated and optimized construction of logical quantum simulation circuits~\cite{li2022paulihedral}, and routing by shuttling~\cite{bluvstein2024logical,tan2024compiling}. Similar improvements over leading methods are expected in these areas as well.

\vspace{1em}

\emph{Data availability---}Our method's implementation (``QuanTile") and all data are available open-source at Ref.~\cite{kattemolle2024quantile}.

\vspace{1em}

\emph{Acknowledgments---}We acknowledge the support from the German Ministry for Education and Research, under the QSolid project, Grant No.~13N16167, and from the State of Baden-W\"urttemberg within the Competence Center Quantum Computing, projects QORA~II and KQCBW24.

\nocite{sriluckshmy2023optimal,yordanov2020efficient,loss1998quantum,burkard2023semiconductor,moura2008z3,aho1983data,zou2024lightsabre,kattemolle2023line}

\bibliography{bib.bib}

@article{kattemolle2022variational,
  title = {{Variational quantum eigensolver for the Heisenberg antiferromagnet on the kagome lattice}},
  author = {Kattem\"olle, Joris and van Wezel, Jasper},
  journal = {Phys. Rev. B},
  volume = {106},
  issue = {21},
  pages = {214429},
  numpages = {17},
  year = {2022},
  month = {Dec},
  publisher = {American Physical Society},
  doi = {10.1103/PhysRevB.106.214429},
  url = {https://link.aps.org/doi/10.1103/PhysRevB.106.214429},
  clean={yes},
  tags={~other}
}

@inproceedings{childs2019circuit,
  author =	{Andrew M. Childs and Eddie Schoute and Cem M. Unsal},
  title =	{{Circuit Transformations for Quantum Architectures}},
  booktitle =	{14th Conference on the Theory of Quantum Computation, Communication and Cryptography (TQC 2019)},
  pages =	{3:1--3:24},
  series =	{Leibniz International Proceedings in Informatics (LIPIcs)},
  ISBN =	{978-3-95977-112-2},
  ISSN =	{1868-8969},
  year =	{2019},
  volume =	{135},
  editor =	{Wim van Dam and Laura Mancinska},
  publisher =	{Schloss Dagstuhl--Leibniz-Zentrum fuer Informatik},
  address =	{Dagstuhl, Germany},
  URL =		{http://drops.dagstuhl.de/opus/volltexte/2019/10395},
  URN =		{urn:nbn:de:0030-drops-103958},
  doi =		{10.4230/LIPIcs.TQC.2019.3},
  tags={~routing}
}

@inproceedings{li2019tackling,
author = {Li, Gushu and Ding, Yufei and Xie, Yuan},
title = {{Tackling the Qubit Mapping Problem for NISQ-Era Quantum Devices}},
year = {2019},
isbn = {9781450362405},
publisher = {Association for Computing Machinery},
address = {New York, NY, USA},
url = {https://doi.org/10.1145/3297858.3304023},
doi = {10.1145/3297858.3304023},
booktitle = {Proceedings of the Twenty-Fourth International Conference on Architectural Support for Programming Languages and Operating Systems},
pages = {1001–1014},
numpages = {14},
location = {Providence, RI, USA},
series = {ASPLOS '19},
tags={~routing}
}

@InProceedings{cowtan2019qubit,
  author =	{Alexander Cowtan and Silas Dilkes and Ross Duncan and Alexandre Krajenbrink and Will Simmons and Seyon Sivarajah},
  title =	{{On the Qubit Routing Problem}},
  booktitle =	{14th Conference on the Theory of Quantum Computation, Communication and Cryptography (TQC 2019)},
  pages =	{5:1--5:32},
  series =	{Leibniz International Proceedings in Informatics (LIPIcs)},
  ISBN =	{978-3-95977-112-2},
  ISSN =	{1868-8969},
  year =	{2019},
  volume =	{135},
  editor =	{Wim van Dam and Laura Mancinska},
  publisher =	{Schloss Dagstuhl--Leibniz-Zentrum fuer Informatik},
  address =	{Dagstuhl, Germany},
  URL =		{http://drops.dagstuhl.de/opus/volltexte/2019/10397},
  URN =		{urn:nbn:de:0030-drops-103972},
  doi =		{10.4230/LIPIcs.TQC.2019.5},
  tags={~routing}
}

@article{maslov2008quantum,
  title={Quantum circuit placement},
  author={Maslov, Dmitri and Falconer, Sean M and Mosca, Michele},
  journal={IEEE Transactions on Computer-Aided Design of Integrated Circuits and Systems},
  volume={27},
  number={4},
  pages={752--763},
  year={2008},
  publisher={IEEE},
  doi={10.1109/TCAD.2008.917562},
  tags={~routing}
}

@article{feynman1982simulating,
	url = {https://doi.org/10.1007%2Fbf02650179},
	year = 1982,
	month = {jun},
	publisher = {Springer Science and Business Media {LLC}},
	volume = {21},
	number = {6-7},
	pages = {467--488},
	author = {Richard P. Feynman},
	title = {Simulating physics with computers},
	journal = {International Journal of Theoretical Physics},
  tags={~simulation}
}

@article{lloyd1996universal,
	url = {https://doi.org/10.1126%2Fscience.273.5278.1073},
	year = 1996,
	month = {aug},
	publisher = {American Association for the Advancement of Science ({AAAS})},
	volume = {273},
	number = {5278},
	pages = {1073--1078},
	author = {Seth Lloyd},
	title = {Universal Quantum Simulators},
	journal = {Science},
  tags={~simulation}
}

@inproceedings{moura2008z3,
author = {De Moura, Leonardo and Bj\o{}rner, Nikolaj},
title = {{Z3: An Efficient SMT Solver}},
year = {2008},
isbn = {3540787992},
publisher = {Springer-Verlag},
address = {Berlin, Heidelberg},
booktitle = {Proceedings of the Theory and Practice of Software, 14th International Conference on Tools and Algorithms for the Construction and Analysis of Systems},
pages = {337–340},
numpages = {4},
location = {Budapest, Hungary},
series = {TACAS'08/ETAPS'08},
doi={10.1007/978-3-540-78800-3\_24}
}

@Article{burkard2025recipes,
	title={{Recipes for the digital quantum simulation of lattice spin systems}},
	author={Guido Burkard},
	journal={SciPost Phys. Core},
	volume={8},
	pages={030},
	year={2025},
	publisher={SciPost},
	doi={10.21468/SciPostPhysCore.8.1.030},
	url={https://scipost.org/10.21468/SciPostPhysCore.8.1.030},
}

@article{childs2021theory,
  title = {{Theory of Trotter Error with Commutator Scaling}},
  author = {Childs, Andrew M. and Su, Yuan and Tran, Minh C. and Wiebe, Nathan and Zhu, Shuchen},
  journal = {Phys. Rev. X},
  volume = {11},
  issue = {1},
  pages = {011020},
  numpages = {49},
  year = {2021},
  month = {Feb},
  publisher = {American Physical Society},
  doi = {10.1103/PhysRevX.11.011020},
  url = {https://link.aps.org/doi/10.1103/PhysRevX.11.011020}
}

@article{youngseok2023evidence,
	author = {Kim, Youngseok and Eddins, Andrew and Anand, Sajant and Wei, Ken Xuan and van den Berg, Ewout and Rosenblatt, Sami and Nayfeh, Hasan and Wu, Yantao and Zaletel, Michael and Temme, Kristan and Kandala, Abhinav},
	doi = {10.1038/s41586-023-06096-3},
	isbn = {1476-4687},
	journal = {Nature},
	number = {7965},
	pages = {500--505},
	title = {Evidence for the utility of quantum computing before fault tolerance},
	volume = {618},
	year = {2023}
  }

@article{kattemolle2023line,
author = {Kattem\"{o}lle, Joris and Hariharan, Seenivasan},
title = {Line-Graph Qubit Routing},
year = {2025},
issue_date = {September 2025},
publisher = {Association for Computing Machinery},
address = {New York, NY, USA},
volume = {6},
number = {3},
url = {https://doi.org/10.1145/3733842},
doi = {10.1145/3733842},
journal = {ACM Transactions on Quantum Computing},
month = jun,
articleno = {22},
numpages = {18}
}

@article{arute2019quantum,
  author={Arute, Frank and Arya, Kunal and Babbush, Ryan and Bacon, Dave and Bardin, Joseph C and Barends, Rami and Biswas, Rupak and Boixo, Sergio and Brandao, Fernando G. S. L. and Buell, David A and others},
	Isbn = {1476-4687},
	Journal = {Nature},
	Number = {7779},
	Pages = {505--510},
	Title = {Quantum supremacy using a programmable superconducting processor},
	Volume = {574},
	Year = {2019},
  doi={10.1038/s41586-019-1666-5},
  url={https://doi.org/10.1038/s41586-019-1666-5}
}

@article{yordanov2020efficient,
  title = {{Efficient quantum circuits for quantum computational chemistry}},
  author = {Yordanov, Yordan S. and Arvidsson-Shukur, David R. M. and Barnes, Crispin H. W.},
  journal = {Phys. Rev. A},
  volume = {102},
  issue = {6},
  pages = {062612},
  numpages = {7},
  year = {2020},
  month = {Dec},
  publisher = {American Physical Society},
  doi = {10.1103/PhysRevA.102.062612},
  url = {https://link.aps.org/doi/10.1103/PhysRevA.102.062612}
}

@article{clinton2024towards,
  title={{Towards near-term quantum simulation of materials}},
  author={Clinton, Laura and Cubitt, Toby and Flynn, Brian and Gambetta, Filippo Maria and Klassen, Joel and Montanaro, Ashley and Piddock, Stephen and Santos, Raul A. and Sheridan, Evan},
  year={2024},
  month=jan,
  volume={15},
  ISSN={2041-1723},
  url={http://dx.doi.org/10.1038/s41467-023-43479-6},
  DOI={10.1038/s41467-023-43479-6},
  number={1},
  journal={{Nature Communications}},
  publisher={Springer Science and Business Media LLC}
}

@article{derby2021compact,
  title = {{Compact fermion to qubit mappings}},
  author = {Derby, Charles and Klassen, Joel and Bausch, Johannes and Cubitt, Toby},
  journal = {Phys. Rev. B},
  volume = {104},
  issue = {3},
  pages = {035118},
  numpages = {12},
  year = {2021},
  month = {Jul},
  publisher = {American Physical Society},
  doi = {10.1103/PhysRevB.104.035118},
  url = {https://link.aps.org/doi/10.1103/PhysRevB.104.035118}
}

@article{sriluckshmy2023optimal,
  title={{Optimal, hardware native decomposition of parameterized multi-qubit Pauli gates}},
  author={Sriluckshmy, P V and Pina-Canelles, Vicente and Ponce, Mario and Algaba, Manuel G and {\v{S}}imkovic IV, Fedor and Leib, Martin},
  journal={Quantum Science and Technology},
  volume={8},
  pages={045029},
  year={2023},
  ISSN={2058-9565},
  url={http://dx.doi.org/10.1088/2058-9565/acfa20},
  DOI={10.1088/2058-9565/acfa20},
  number={4},
  publisher={IOP Publishing},
  month=sep,
}

@article{kogut1975hamiltonian,
  title = {{Hamiltonian formulation of Wilson's lattice gauge theories}},
  author = {Kogut, John and Susskind, Leonard},
  journal = {Phys. Rev. D},
  volume = {11},
  issue = {2},
  pages = {395--408},
  numpages = {0},
  year = {1975},
  month = {Jan},
  publisher = {American Physical Society},
  doi = {10.1103/PhysRevD.11.395},
  url = {https://link.aps.org/doi/10.1103/PhysRevD.11.395}
}

@article{paulison2021simulating,
  title = {{Simulating 2D Effects in Lattice Gauge Theories on a Quantum Computer}},
  author = {Paulson, Danny and Dellantonio, Luca and Haase, Jan F. and Celi, Alessio and Kan, Angus and Jena, Andrew and Kokail, Christian and van Bijnen, Rick and Jansen, Karl and Zoller, Peter and Muschik, Christine A.},
  journal = {PRX Quantum},
  volume = {2},
  issue = {3},
  pages = {030334},
  numpages = {26},
  year = {2021},
  month = {Aug},
  publisher = {American Physical Society},
  doi = {10.1103/PRXQuantum.2.030334},
  url = {https://link.aps.org/doi/10.1103/PRXQuantum.2.030334}
}

@article{haase2021resource,
  doi = {10.22331/q-2021-02-04-393},
  url = {https://doi.org/10.22331/q-2021-02-04-393},
  title = {{A resource efficient approach for quantum and classical simulations of gauge theories in particle physics}},
  author = {Haase, Jan F. and Dellantonio, Luca and Celi, Alessio and Paulson, Danny and Kan, Angus and Jansen, Karl and Muschik, Christine A.},
  journal = {{Quantum}},
  issn = {2521-327X},
  publisher = {{Verein zur F{\"{o}}rderung des Open Access Publizierens in den Quantenwissenschaften}},
  volume = {5},
  pages = {393},
  month = feb,
  year = {2021}
}

@ARTICLE{meth2023simulating,
       author = {{Meth}, Michael and {Haase}, Jan F. and {Zhang}, Jinglei and {Edmunds}, Claire and {Postler}, Lukas and {Steiner}, Alex and {Jena}, Andrew J. and {Dellantonio}, Luca and {Blatt}, Rainer and {Zoller}, Peter and {Monz}, Thomas and {Schindler}, Philipp and {Muschik}, Christine and {Ringbauer}, Martin},
        title = "{Simulating 2D lattice gauge theories on a qudit quantum computer}",
      journal = {arXiv e-prints},
         year = 2023,
        month = oct,
          doi = {10.48550/arXiv.2310.12110},
archivePrefix = {arXiv},
       eprint = {2310.12110},
 primaryClass = {quant-ph}
}

@article{gopalakrishnan2018facilitated,
doi = {10.1088/2058-9565/aad759},
url = {https://dx.doi.org/10.1088/2058-9565/aad759},
year = {2018},
month = {aug},
publisher = {IOP Publishing},
volume = {3},
number = {4},
pages = {044004},
author = {Sarang Gopalakrishnan and Bahti Zakirov},
title = {Facilitated quantum cellular automata as simple models with non-thermal eigenstates and dynamics},
journal = {{Quantum Science and Technology}}
}

@article{klobas2021exact,
  title = {{Exact Thermalization Dynamics in the ``Rule 54'' Quantum Cellular Automaton}},
  author = {Klobas, Katja and Bertini, Bruno and Piroli, Lorenzo},
  journal = {Phys. Rev. Lett.},
  volume = {126},
  issue = {16},
  pages = {160602},
  numpages = {7},
  year = {2021},
  month = {Apr},
  publisher = {American Physical Society},
  doi = {10.1103/PhysRevLett.126.160602},
  url = {https://link.aps.org/doi/10.1103/PhysRevLett.126.160602}
}

@article{vidal2004universal,
  title = {{Universal quantum circuit for two-qubit transformations with three controlled-NOT gates}},
  author = {Vidal, G. and Dawson, C. M.},
  journal = {Phys. Rev. A},
  volume = {69},
  issue = {1},
  pages = {010301},
  numpages = {4},
  year = {2004},
  month = {Jan},
  publisher = {American Physical Society},
  doi = {10.1103/PhysRevA.69.010301},
  url = {https://link.aps.org/doi/10.1103/PhysRevA.69.010301}
}

@article{foxen2020demonstrating,
  title = {{Demonstrating a Continuous Set of Two-Qubit Gates for Near-Term Quantum Algorithms}},
  author = {Foxen, B. and Neill, C. and Dunsworth, A. and Roushan, P. and Chiaro, B. and Megrant, A. and Kelly, J. and Chen, Zijun and Satzinger, K. and Barends, R. and others},
  collaboration = {Google AI Quantum},
  journal = {Phys. Rev. Lett.},
  volume = {125},
  issue = {12},
  pages = {120504},
  numpages = {6},
  year = {2020},
  month = {Sep},
  publisher = {American Physical Society},
  doi = {10.1103/PhysRevLett.125.120504},
  url = {https://link.aps.org/doi/10.1103/PhysRevLett.125.120504}
}

@inproceedings{tan2020optimal,
author = {Tan, Bochen and Cong, Jason},
title = {Optimal layout synthesis for quantum computing},
year = {2020},
isbn = {9781450380263},
publisher = {Association for Computing Machinery},
address = {New York, NY, USA},
url = {https://doi.org/10.1145/3400302.3415620},
doi = {10.1145/3400302.3415620},
booktitle = {Proceedings of the 39th International Conference on Computer-Aided Design},
articleno = {137},
numpages = {9},
location = {Virtual Event, USA},
series = {ICCAD '20}
}

@article{suzuki1991general,
    author = {Suzuki, Masuo},
    title = "{General theory of fractal path integrals with applications to many‐body theories and statistical physics}",
    journal = {Journal of Mathematical Physics},
    volume = {32},
    number = {2},
    pages = {400-407},
    year = {1991},
    month = {02},
    issn = {0022-2488},
    doi = {10.1063/1.529425},
    url = {https://doi.org/10.1063/1.529425}
}

@software{kattemolle2024quantile,
  author       = {Kattem\"olle, Joris},
  title        = {{QuanTile}},
  month        = oct,
  year         = 2024,
  publisher    = {Zenodo},
  doi          = {10.5281/zenodo.14336692},
  url          = {https://doi.org/10.5281/zenodo.14336692},
  howpublished = "\url{https://doi.org/10.5281/zenodo.14336692}"
}

@article{kattemolle2024edge,
    author = {Kattemölle, Joris},
    title = {Edge coloring lattice graphs},
    journal = {Journal of Mathematical Physics},
    volume = {66},
    number = {5},
    pages = {051901},
    year = {2025},
    month = {05},
    issn = {0022-2488},
    doi = {10.1063/5.0243007},
    url = {https://doi.org/10.1063/5.0243007}
}

@INPROCEEDINGS{lin2023scalable,
  author={Lin, Wan-Hsuan and Kimko, Jason and Tan, Bochen and Bjørner, Nikolaj and Cong, Jason},
  booktitle={{2023 60th ACM/IEEE Design Automation Conference (DAC)}},
  title={{Scalable Optimal Layout Synthesis for NISQ Quantum Processors}},
  year={2023},
  volume={},
  number={},
  pages={1-6},
  doi={10.1109/DAC56929.2023.10247760}}

@misc{zou2024lightsabre,
       author = {{Zou}, Henry and {Treinish}, Matthew and {Hartman}, Kevin and {Ivrii}, Alexander and {Lishman}, Jake},
        title = "{LightSABRE: A Lightweight and Enhanced SABRE Algorithm}",
         year = 2024,
        month = sep,
          eid = {arXiv:2409.08368},
        pages = {arXiv:2409.08368},
          doi = {10.48550/arXiv.2409.08368},
archivePrefix = {arXiv},
       eprint = {2409.08368},
 primaryClass = {quant-ph}
}

@misc{kremer2024practical,
       author = {{Kremer}, David and {Villar}, Victor and {Paik}, Hanhee and {Duran}, Ivan and {Faro}, Ismael and {Cruz-Benito}, Juan},
        title = "{Practical and efficient quantum circuit synthesis and transpiling with Reinforcement Learning}",
         year = 2024,
        month = may,
          eid = {arXiv:2405.13196},
        pages = {arXiv:2405.13196},
          doi = {10.48550/arXiv.2405.13196},
archivePrefix = {arXiv},
       eprint = {2405.13196},
 primaryClass = {quant-ph}
}

@misc{cirq2024cirq,
title={Cirq},
url={https://zenodo.org/doi/10.5281/zenodo.4062499},
howpublished={\url{https://zenodo.org/doi/10.5281/zenodo.4062499}},
DOI={10.5281/ZENODO.4062499},
publisher={Zenodo},
author={{Cirq Developers}},
year=2024,
month=may
}

@misc{qiskit2024quantum,
      title={Quantum computing with {Q}iskit},
      author={Javadi-Abhari, Ali and Treinish, Matthew and Krsulich, Kevin and Wood, Christopher J. and Lishman, Jake and Gacon, Julien and Martiel, Simon and Nation, Paul D. and Bishop, Lev S. and Cross, Andrew W. and others},
      year={2024},
      doi={10.48550/arXiv.2405.08810},
      eprint={2405.08810},
      archivePrefix={arXiv},
      primaryClass={quant-ph}
}

@book{aho1983data,
  location = {Reading, Massachusetts},
  author = {Aho, Alfred V. and Hopcroft, John E. and Ullman, Jeffrey D.},
  publisher = {Addison-Wesley},
  series = {Computer Science and Information Processing},
  title = {{Data Structures and Algorithms}},
  year = 1983
}

@article{loss1998quantum,
  title = {Quantum computation with quantum dots},
  author = {Loss, Daniel and DiVincenzo, David P.},
  journal = {Phys. Rev. A},
  volume = {57},
  issue = {1},
  pages = {120--126},
  numpages = {0},
  year = {1998},
  month = {Jan},
  publisher = {American Physical Society},
  doi = {10.1103/PhysRevA.57.120},
  url = {https://link.aps.org/doi/10.1103/PhysRevA.57.120}
}

@article{burkard2023semiconductor,
  title = {{Semiconductor spin qubits}},
  author = {Burkard, Guido and Ladd, Thaddeus D. and Pan, Andrew and Nichol, John M. and Petta, Jason R.},
  journal = {Rev. Mod. Phys.},
  volume = {95},
  issue = {2},
  pages = {025003},
  numpages = {58},
  year = {2023},
  month = {Jun},
  publisher = {American Physical Society},
  doi = {10.1103/RevModPhys.95.025003},
  url = {https://link.aps.org/doi/10.1103/RevModPhys.95.025003}
}

@article{mi2024stable,
author = {X. Mi and A. A. Michailidis  and S. Shabani  and K. C. Miao  and P. V. Klimov  and J. Lloyd  and E. Rosenberg  and R. Acharya  and I. Aleiner  and T. I. Andersen  and others},
title = {{Stable quantum-correlated many-body states through engineered dissipation}},
journal = {Science},
volume = {383},
number = {6689},
pages = {1332-1337},
year = {2024},
doi = {10.1126/science.adh9932},
URL = {https://www.science.org/doi/abs/10.1126/science.adh9932}
}

@article{
frey2022realization,
author = {Frey, Philipp and Rachel, Stephan},
title = {{Realization of a discrete time crystal on 57 qubits of a quantum computer}},
journal = {Science Advances},
volume = {8},
number = {9},
pages = {eabm7652},
year = {2022},
doi = {10.1126/sciadv.abm7652},
URL = {https://www.science.org/doi/abs/10.1126/sciadv.abm7652}
}

@article{
rosenberg2024dynamics,
author = {Rosenberg, Eliott and Andersen, T. I. and Samajdar, Rhine and Petukhov, Andre and Hoke, J. C. and Abanin, Dmitry and Bengtsson, Andreas and Drozdov, I. K. and Erickson, Catherine and Klimov, P. V. and others},
title = {{Dynamics of magnetization at infinite temperature in a Heisenberg spin chain}},
journal = {Science},
volume = {384},
number = {6691},
pages = {48-53},
year = {2024},
doi = {10.1126/science.adi7877},
URL = {https://www.science.org/doi/abs/10.1126/science.adi7877}
}

@article{pichler2016realtime,
  title = {{Real-Time Dynamics in U(1) Lattice Gauge Theories with Tensor Networks}},
  author = {Pichler, T. and Dalmonte, M. and Rico, E. and Zoller, P. and Montangero, S.},
  journal = {Phys. Rev. X},
  volume = {6},
  issue = {1},
  pages = {011023},
  numpages = {17},
  year = {2016},
  month = {Mar},
  publisher = {American Physical Society},
  doi = {10.1103/PhysRevX.6.011023},
  url = {https://link.aps.org/doi/10.1103/PhysRevX.6.011023}
}

@article{rokhsar1988superconductivity,
  title = {{Superconductivity and the Quantum Hard-Core Dimer Gas}},
  author = {Rokhsar, Daniel S. and Kivelson, Steven A.},
  journal = {Phys. Rev. Lett.},
  volume = {61},
  issue = {20},
  pages = {2376--2379},
  numpages = {0},
  year = {1988},
  month = {Nov},
  publisher = {American Physical Society},
  doi = {10.1103/PhysRevLett.61.2376},
  url = {https://link.aps.org/doi/10.1103/PhysRevLett.61.2376}
}

@article{dagotto1994correlated,
  title = {{Correlated electrons in high-temperature superconductors}},
  author = {Dagotto, Elbio},
  journal = {Rev. Mod. Phys.},
  volume = {66},
  issue = {3},
  pages = {763--840},
  numpages = {0},
  year = {1994},
  month = {Jul},
  publisher = {American Physical Society},
  doi = {10.1103/RevModPhys.66.763},
  url = {https://link.aps.org/doi/10.1103/RevModPhys.66.763}
}

@article{cade2020strategies,
  title = {{Strategies for solving the Fermi-Hubbard model on near-term quantum computers}},
  author = {Cade, Chris and Mineh, Lana and Montanaro, Ashley and Stanisic, Stasja},
  journal = {Phys. Rev. B},
  volume = {102},
  issue = {23},
  pages = {235122},
  numpages = {25},
  year = {2020},
  month = {Dec},
  publisher = {American Physical Society},
  doi = {10.1103/PhysRevB.102.235122},
  url = {https://link.aps.org/doi/10.1103/PhysRevB.102.235122}
}

@article{felser2020twodimensional,
  title = {{Two-Dimensional Quantum-Link Lattice Quantum Electrodynamics at Finite Density}},
  author = {Felser, Timo and Silvi, Pietro and Collura, Mario and Montangero, Simone},
  journal = {Phys. Rev. X},
  volume = {10},
  issue = {4},
  pages = {041040},
  numpages = {25},
  year = {2020},
  month = {Nov},
  publisher = {American Physical Society},
  doi = {10.1103/PhysRevX.10.041040},
  url = {https://link.aps.org/doi/10.1103/PhysRevX.10.041040}
}

@article{nannicini2022optimal,
  author = {Nannicini, Giacomo and Bishop, Lev S. and G\"{u}nl\"{u}k, Oktay and Jurcevic, Petar},
  title = {{Optimal Qubit Assignment and Routing via Integer Programming}},
  year = {2022},
  issue_date = {March 2023},
  publisher = {Association for Computing Machinery},
  address = {New York, NY, USA},
  volume = {4},
  number = {1},
  url = {https://doi.org/10.1145/3544563},
  doi = {10.1145/3544563},
  journal = {ACM Transactions on Quantum Computing},
  month = oct,
  articleno = {7},
  numpages = {31}
}

@misc{ping2025assessing,
       author = {{Ping}, Shuohao and {Lin}, Wan-Hsuan and {Bochen Tan}, Daniel and {Cong}, Jason},
        title = "{Assessing Quantum Layout Synthesis Tools via Known Optimal-SWAP Cost Benchmarks}",
         year = 2025,
        month = feb,
          doi = {10.48550/arXiv.2502.08839},
archivePrefix = {arXiv},
       eprint = {2502.08839},
 primaryClass = {quant-ph}
}

@misc{zimboras2025myths,
       author = {{Zimbor{\'a}s}, Zolt{\'a}n and {Koczor}, B{\'a}lint and {Holmes}, Zo{\"e} and {Borrelli}, Elsi-Mari and {Gily{\'e}n}, Andr{\'a}s and {Huang}, Hsin-Yuan and {Cai}, Zhenyu and {Ac{\'\i}n}, Antonio and {Aolita}, Leandro and {Banchi}, Leonardo and {Brand{\~a}o}, Fernando G.~S.~L. and {Cavalcanti}, Daniel and {Cubitt}, Toby and {Filippov}, Sergey N. and {Garc{\'\i}a-P{\'e}rez}, Guillermo and {Goold}, John and {K{\'a}lm{\'a}n}, Orsolya and {Kyoseva}, Elica and {Rossi}, Matteo A.~C. and {Sokolov}, Boris and {Tavernelli}, Ivano and {Maniscalco}, Sabrina},
        title = "{Myths around quantum computation before full fault tolerance: What no-go theorems rule out and what they don't}",
         year = 2025,
        month = jan,
          doi = {10.48550/arXiv.2501.05694},
archivePrefix = {arXiv},
       eprint = {2501.05694},
 primaryClass = {quant-ph}
}

@article{childs2018toward,
  author = {Andrew M. Childs  and Dmitri Maslov  and Yunseong Nam  and Neil J. Ross  and Yuan Su},
  title = {{Toward the first quantum simulation with quantum speedup}},
  journal = {Proceedings of the National Academy of Sciences},
  volume = {115},
  number = {38},
  pages = {9456-9461},
  year = {2018},
  doi = {10.1073/pnas.1801723115},
  URL = {https://www.pnas.org/doi/abs/10.1073/pnas.1801723115},
  eprint = {https://www.pnas.org/doi/pdf/10.1073/pnas.1801723115}
}

@article{acharya2025quantum,
	author = {Acharya, Rajeev and Abanin, Dmitry A. and Aghababaie-Beni, Laleh and Aleiner, Igor and Andersen, Trond I. and Ansmann, Markus and Arute, Frank and Arya, Kunal and Asfaw, Abraham and Astrakhantsev, Nikita and others},
	doi = {10.1038/s41586-024-08449-y},
	journal = {Nature},
	number = {8052},
	pages = {920--926},
	title = {{Quantum error correction below the surface code threshold}},
	url = {https://doi.org/10.1038/s41586-024-08449-y},
	volume = {638},
	year = {2025}
}

@article{bravyi2022future,
    author = {Bravyi, Sergey and Dial, Oliver and Gambetta, Jay M. and Gil, Darío and Nazario, Zaira},
    title = {The future of quantum computing with superconducting qubits},
    journal = {Journal of Applied Physics},
    volume = {132},
    number = {16},
    pages = {160902},
    year = {2022},
    month = {10},
    issn = {0021-8979},
    doi = {10.1063/5.0082975},
    url = {https://doi.org/10.1063/5.0082975}
}

@inproceedings{li2022paulihedral,
author = {Li, Gushu and Wu, Anbang and Shi, Yunong and Javadi-Abhari, Ali and Ding, Yufei and Xie, Yuan},
title = {Paulihedral: a generalized block-wise compiler optimization framework for Quantum simulation kernels},
year = {2022},
isbn = {9781450392051},
publisher = {Association for Computing Machinery},
address = {New York, NY, USA},
url = {https://doi.org/10.1145/3503222.3507715},
doi = {10.1145/3503222.3507715},
booktitle = {Proceedings of the 27th ACM International Conference on Architectural Support for Programming Languages and Operating Systems},
pages = {554–569},
numpages = {16},
location = {Lausanne, Switzerland},
series = {ASPLOS '22}
}

@ARTICLE{maslov2008quantum_circuit_simplification,
  author={Maslov, Dmitri and Dueck, Gerhard W. and Miller, D. Michael and Negrevergne, Camille},
  journal={IEEE Transactions on Computer-Aided Design of Integrated Circuits and Systems},
  title={{Quantum Circuit Simplification and Level Compaction}},
  year={2008},
  volume={27},
  number={3},
  pages={436-444},
  doi={10.1109/TCAD.2007.911334}
 }

@article{diepen2021quantum,
  title = {{Quantum Simulation of Antiferromagnetic Heisenberg Chain with Gate-Defined Quantum Dots}},
  author = {van Diepen, C. J. and Hsiao, T.-K. and Mukhopadhyay, U. and Reichl, C. and Wegscheider, W. and Vandersypen, L. M. K.},
  journal = {Phys. Rev. X},
  volume = {11},
  issue = {4},
  pages = {041025},
  numpages = {15},
  year = {2021},
  month = {Nov},
  publisher = {American Physical Society},
  doi = {10.1103/PhysRevX.11.041025},
  url = {https://link.aps.org/doi/10.1103/PhysRevX.11.041025}
}

@article{hensgens2017quantum,
	author = {Hensgens, T. and Fujita, T. and Janssen, L. and Li, Xiao and Van Diepen, C. J. and Reichl, C. and Wegscheider, W. and Das Sarma, S. and Vandersypen, L. M. K.},
	doi = {10.1038/nature23022},
	journal = {Nature},
	number = {7665},
	pages = {70--73},
	title = {{Quantum simulation of a Fermi-Hubbard model using a semiconductor quantum dot array}},
	url = {https://doi.org/10.1038/nature23022},
	volume = {548},
	year = {2017}
}

@article{farrelly2020review,
  doi = {10.22331/q-2020-11-30-368},
  url = {https://doi.org/10.22331/q-2020-11-30-368},
  title = {A review of {Q}uantum {C}ellular {A}utomata},
  author = {Farrelly, Terry},
  journal = {{Quantum}},
  issn = {2521-327X},
  publisher = {{Verein zur F{\"{o}}rderung des Open Access Publizierens in den Quantenwissenschaften}},
  volume = {4},
  pages = {368},
  month = nov,
  year = {2020}
}

@article{lee2020exact,
  title = {Exact three-colored quantum scars from geometric frustration},
  author = {Lee, Kyungmin and Melendrez, Ronald and Pal, Arijeet and Changlani, Hitesh J.},
  journal = {Phys. Rev. B},
  volume = {101},
  issue = {24},
  pages = {241111},
  numpages = {6},
  year = {2020},
  month = {Jun},
  publisher = {American Physical Society},
  doi = {10.1103/PhysRevB.101.241111},
  url = {https://link.aps.org/doi/10.1103/PhysRevB.101.241111}
}

@article{mcclarty2020disorder-free,
  title = {{Disorder-free localization and many-body quantum scars from magnetic frustration}},
  author = {McClarty, Paul A. and Haque, Masudul and Sen, Arnab and Richter, Johannes},
  journal = {Phys. Rev. B},
  volume = {102},
  issue = {22},
  pages = {224303},
  numpages = {15},
  year = {2020},
  month = {Dec},
  publisher = {American Physical Society},
  doi = {10.1103/PhysRevB.102.224303},
  url = {https://link.aps.org/doi/10.1103/PhysRevB.102.224303}
}

@misc{kitaev1995quantum,
       author = {{Kitaev}, A. Yu.},
        title = "{Quantum measurements and the Abelian Stabilizer Problem}",
         year = 1995,
        month = nov,
          doi = {10.48550/arXiv.quant-ph/9511026},
archivePrefix = {arXiv},
       eprint = {quant-ph/9511026},
 primaryClass = {quant-ph}
}

@article{wecker2015progress,
	title={Progress towards practical quantum variational algorithms},
	author={Wecker, Dave and Hastings, Matthew B and Troyer, Matthias},
	journal={Physical Review A},
	volume={92},
	number={4},
	pages={042303},
	year={2015},
	publisher={APS},
	doi = {10.1103/physreva.92.042303}
}

@article{
ryan-anderson2024high-fidelity,
author = {Ryan-Anderson, C and Brown, N.C. and Baldwin, C.H. and Dreiling, J.M. and Foltz, C and Gaebler, J.P. and Gatterman, T.M. and Hewitt, N and Holliman, C and Horst, C.V. and others},
title = {High-fidelity teleportation of a logical qubit using transversal gates and lattice surgery},
journal = {Science},
volume = {385},
number = {6715},
pages = {1327-1331},
year = {2024},
doi = {10.1126/science.adp6016},
URL = {https://www.science.org/doi/abs/10.1126/science.adp6016},
eprint = {https://www.science.org/doi/pdf/10.1126/science.adp6016}}

@article{bluvstein2024logical,
	author = {Bluvstein, Dolev and Evered, Simon J. and Geim, Alexandra A. and Li, Sophie H. and Zhou, Hengyun and Manovitz, Tom and Ebadi, Sepehr and Cain, Madelyn and Kalinowski, Marcin and Hangleiter, Dominik and Bonilla Ataides, J. Pablo and Maskara, Nishad and Cong, Iris and Gao, Xun and Sales Rodriguez, Pedro and Karolyshyn, Thomas and Semeghini, Giulia and Gullans, Michael J. and Greiner, Markus and Vuleti{\'c}, Vladan and Lukin, Mikhail D.},
	da = {2024/02/01},
	doi = {10.1038/s41586-023-06927-3},
	id = {Bluvstein2024},
	isbn = {1476-4687},
	journal = {Nature},
	number = {7997},
	pages = {58--65},
	title = {Logical quantum processor based on reconfigurable atom arrays},
	ty = {JOUR},
	url = {https://doi.org/10.1038/s41586-023-06927-3},
	volume = {626},
	year = {2024}}

@article{tan2024compiling,
  doi = {10.22331/q-2024-03-14-1281},
  url = {https://doi.org/10.22331/q-2024-03-14-1281},
  title = {Compiling {Q}uantum {C}ircuits for {D}ynamically {F}ield-{P}rogrammable {N}eutral {A}toms {A}rray {P}rocessors},
  author = {Tan, Daniel Bochen and Bluvstein, Dolev and Lukin, Mikhail D. and Cong, Jason},
  journal = {{Quantum}},
  issn = {2521-327X},
  publisher = {{Verein zur F{\"{o}}rderung des Open Access Publizierens in den Quantenwissenschaften}},
  volume = {8},
  pages = {1281},
  month = mar,
  year = {2024}
}

\phantom{\ }

\section*{End Matter}
\emph{Tileable SMT formulation---}Denote by $G_{g}.T$ the integer variable describing the time coordinate of the physical gate $G$ implementing the logical gate $g$, and denote by $G_{g}.Q.S$ ($G_{g}.Q'.S$) the integer variable describing the seed number of qubit 1 (2) of $G_g$. Using Theorem 1, tileability of the physical basis circuit is asserted by adding the constraint \begin{equation}
G_{g}.T=G_{g'}.T\Rightarrow G_{g}.A.S\neq G_{g'}.B.S
\end{equation}
for all pairs $\{g,g'\}$ of indistinct logical gates and all $A,B\in\{Q,Q'\}$ (momentarily assuming only two-qudit gates for simplicity). 

\emph{Routing options---}For the results displayed in \tab{results}, we assumed a mobility zone of 3 by 3 cells ($\delta=1$). For the ATL models only, we optimized the gate order (i). We allowed SWAP-merging (ii), except for the Fermi-Hubbard circuit. We enforced cyclic routing (iii), except for the ATL circuits. Logical circuits were not sliced (iv), except for the Kogut-Susskind circuit, which had a slice depth of 20.


\clearpage
\onecolumngrid
{\huge{\centering Supplemental material to:\\
\vspace{.5em}``Efficient quantum simulation for translationally invariant systems''}}
\vspace{4em}
\twocolumngrid
\tableofcontents
\onecolumngrid
\onecolumngrid
\clearpage

\section{Fundamental concepts}

Here, we formally introduce concepts fundamental for defining and routing tileable circuits, as well as for constructing tileable hardware connectivity graphs. These correspond to the definitions in \lin{quantile/base.py} in the implementation~\cite{kattemolle2024quantile}.

\subsection{Basis graphs}\label{sec:basis_graphs}

A \emph{basis graph} $B$ is a graph $B = (V(B), E(B))$ with vertex set $V(B)$ and edge set $E(B)$. Each vertex $v \in V(B)$ has the form $v = (v_x, v_y, v_s)$, where $v_x, v_y \in \mathbb{Z}$ are \emph{cell coordinates}, and $v_s \in\mathbb N^0$ is an identifier called the \emph{seed number} of the vertex. For clarity, we assume two-dimensional basis graphs; higher-dimensional basis graphs can be obtained by extending the list of cell coordinates accordingly.

Basis graphs (along with lattice graphs) were formally introduced in Ref.~\cite{kattemolle2024edge}. However, we present their definitions here in a slightly different but equivalent form for completeness and consistency with the definitions of basis circuits. For a visual example of basis and lattice graphs, we refer the reader to Fig.~1 of Ref.~\cite{kattemolle2024edge}.

We define the translation operator on vertices as $T_{\Delta x,\Delta y}(v) = (v_x + \Delta x, v_y + \Delta y, v_s)$. This naturally extends to a translation operator on an edge $e = (e_0, e_1)$ by applying $T$ to each of its vertices, $T_{\Delta x,\Delta y}(e) = (T_{\Delta x,\Delta y}(e_0), T_{\Delta x,\Delta y}(e_1))$. For simplicity of notation, we use $T$ to denote the translation operator on both vertices and edges.
Extending this definition further, $T$ induces a translation operator on basis graphs by applying $T$ to all vertices and edges in the graph, $T(B) = (T[V(B)], T[E(B)])$, where $T[E(B)]$ is the set obtained by applying $T$ to each edge in $E(B)$ separately, and similarly for $T[V(B)]$.

A \emph{lattice graph} $\mathcal{G}$ is obtained by translating and merging infinitely many copies of a basis graph $B$,
\begin{equation}
    \mathcal{G}(B) = \bigcup_{\Delta x, \Delta y \in \mathbb{Z}} T_{\Delta x, \Delta y}(B).
\end{equation}
The union of two graphs $A$ and $B$ is defined as $A \cup B = \big(V(A) \cup V(B), E(A) \cup E(B)\big)$. A \emph{patch} of a lattice graph is a finite subgraph obtained by translating and merging a bounded number of copies of $B$,
\begin{equation}
    P_{n,m}(B) = \bigcup_{\Delta x \in \mathbb{Z}_n, \Delta y \in \mathbb{Z}_m} T_{\Delta x, \Delta y}(B).
\end{equation}

A \emph{cell} at coordinates $(x,y)$ consists of all vertices in $V(\mathcal{G})$ that share the same cell coordinates $(x,y)$. We refer to the cell at $(x,y) = (0,0)$ as the \emph{central cell}.

\subsection{Basis circuits}\label{sec:basis_circuits}
Similar to a basis graph, a \emph{basis circuit} $C=(Q(C),G(C))$ with qudits $Q(C)$ and gates $G(C)$ is a circuit in which each qudit $q$ has a cell $x$-coordinate $q.\tilde x$, a cell $y$-coordinate $q.\tilde y\in{\mathbb Z}$ and a seed number $q.\tilde s\in \mathbb N^0$. We denote the qudit a single-qudit gate $g\in G(C)$ acts on by $g.{\tilde q}$. For two-qudit gates, we denote the qudits the gate acts on by $g.{\tilde q}$ and $g.{\tilde q'}$. To ensure an unambiguous gate order when using a basis circuit to generate a circuit patch, we explicitly include the circuit layer (or time step) at which each gate $g$ operates, denoted by $g.\tilde{t}\in \mathbb N^0$. In \sec{scheduling}, we demonstrate that if gate times are unspecified, they can always be assigned while preserving tileability and while utilizing gate parallelism.

Similar to vertices, we define the translation operator acting on a qudit $q$ as $T_{\Delta x,\Delta y}(q) = q'$, with $q'.\tilde{x} = q.\tilde{x} + \Delta x$, $q'.\tilde{y} = q.\tilde{y} + \Delta y$, and $q'.\tilde{s} = q.\tilde{s}$. This naturally induces a translation operator on gates by translating each qudit that a gate $g$ acts on by the same amount. That is, for the translated gate $T_{\Delta x,\Delta y, \Delta t}(g)=g'$, we have $g'.\tilde{q} = T_{\Delta x,\Delta y}(g.\tilde{q})$. For two-qudit gates, we additionally have $g'.\tilde{q}' = T_{\Delta x,\Delta y}(g.\tilde{q}')$. Furthermore, the time coordinate transforms as $g'.\tilde{t} = g.\tilde{t} + \Delta t$. The translation operator on gates induces a translation operator on basis circuits by translating each gate and qudit in the circuit; $T_{\Delta x,\Delta y, \Delta t}(C) = (T_{\Delta x,\Delta y}(Q(C)), T_{\Delta x,\Delta y, \Delta t}(G(C)))$, where $T_{\Delta x,\Delta y}(Q(C))$ translates each qudit in $Q(C)$, and $T_{\Delta x,\Delta y, \Delta t}(G(C))$ translates each gate in $G(C)$.

Continuing the analogy to graphs, a (finite-depth) \emph{lattice circuit} $\mc C$ is obtained by translating and merging infinitely many copies of a basis circuit $C$. That is,
\begin{equation}
  \mc C(C) = \bigcup_{\Delta x,\Delta y\in\mathbb{Z}} T_{\Delta x,\Delta y}(C).
\end{equation}
The union of two circuits $A$ and $B$ is defined as $A \cup B = (Q(A) \cup Q(B), G(A) \cup G(B))$. A \emph{patch} of a lattice circuit is obtained by translating and merging a finite number of copies of a basis circuit $C$,
\begin{equation}
  P_{n,m}(C) = \bigcup_{\Delta x\in\mathbb{Z}_n, \Delta y\in\mathbb{Z}_m} T_{\Delta x, \Delta y,0}(C).
\end{equation}
A \emph{circuit cell} at $(x,y)$ consists of all qudits in $Q(C)$ that have cell coordinates $(x,y)$. The \emph{circuit layer} $t$ of a basis circuit, circuit patch, or lattice circuit consists of all gates $g$ in that circuit for which $g.\tilde{t} = t$. The depth $d$ of a basis circuit, circuit patch, or lattice circuit is the number of layers in the circuit.

Quantum circuits for quantum simulation via Trotterization (\sec{trotterization}) consist of a circuit cycle that is repeated temporally. For this application, it is natural to consider circuits that are periodic in time and define \emph{spacetime patches} as
\[
P_{n,m,l}(C) = \bigcup_{\Delta x\in\mathbb{Z}_n, \Delta y\in\mathbb{Z}_m, \Delta t\in d\,\mathbb{Z}_l} T_{\Delta x, \Delta y, \Delta t}(C),
\]
where $d\,\mathbb{Z}_l$ denotes $\{0,d, 2d, \ldots, d(l-1)\}$, with $d$ the depth of $C$.

Unlike edges, two gates may not collide. A \emph{gate collision} between two gates $g, g' \in G(C)$ occurs if $g.\tilde{t} = g'.\tilde{t}$ and there exist qudits $\tilde{a}$ and $\tilde{b}$ such that $g.\tilde{a} = g'.\tilde{b}$, where $\tilde{a} \in \{\tilde{q}\}$ for single-qudit gates and $\tilde{a} \in \{\tilde{q}, \tilde{q}'\}$ for two-qudit gates, with $\tilde{b}$ satisfying a similar condition. Gate collisions cannot occur in $C$ if $C$ is a valid quantum circuit. Moreover, gate collisions should not arise when generating a lattice circuit or any circuit patch. Therefore, for $C$ to qualify as a basis circuit, we require that no gate collisions occur in $P_{n,m,l}(C)$ for any $n, m,l \in \mathbb{N^+}$. A necessary and sufficient condition for this property is that in each layer of the circuit, at most one qudit with seed number $s$ is acted upon by a gate in that layer.

\begin{theorem}\label{thm:unique_seeds}
Let $C$ be a basis circuit. There are no gate collisions in the circuit patch $P_{n,m,l}(C)$ for arbitrary $n,m,l \in \mathbb{N}^+$ if and only if, for every time $t$ and seed number $s$, there is at most one qudit $q \in Q(C)$ acted on by layer $t$ of $C$ so that $q.\tilde{s} = s$.
\end{theorem}

\begin{proof}By construction of $P_{n,m,l}$, there cannot be gate collisions due to temporal translation and hence in the following we may take $l=1$. Suppose there exist $n,m,l \in \mathbb{N}^+$ for which a gate collision occurs in $P_{n,m,l}(C)$. Then there exists a layer $t$ and a qudit $q$ such that two distinct gates $g, g' \in G(C)$, both acting in layer $t$, operate on $q$. By the construction of $P_{n,m,l}(C)$, there exist translations $(\Delta x, \Delta y)$ and $(\Delta x', \Delta y')$ such that $g = T_{\Delta x,\Delta y}(\mathfrak{g})$ and $g' = T_{\Delta x',\Delta y'}(\mathfrak{g}')$ for some $\mathfrak{g},\mathfrak{g}' \in G(C)$. Note that $(\Delta x,\Delta y) \neq (\Delta x',\Delta y')$, since if they were equal, a gate collision would already occur in $C$, contradicting its validity. Applying the corresponding translations to $q$, we obtain $q = T_{\Delta x,\Delta y}(\mathfrak{q})$ and $q = T_{\Delta x',\Delta y'}(\mathfrak{q}')$ with $\mathfrak{q},\mathfrak{q}' \in Q(C)$ and $\mathfrak{q} \neq \mathfrak{q}'$. Since these translations yield the same qudit $q$, it follows that $\mathfrak{q}.\tilde{s} = \mathfrak{q}'.\tilde{s}$.

Conversely, assume that there exist a time $t$, a seed number $s$, and distinct qudits $q,q' \in Q(C)$, both acted on by layer $t$ of $C$, such that $q.\tilde{s} = q'.\tilde{s} = s$. Then there exist $\Delta x,\Delta y \in \mathbb{Z}$ such that $T_{\Delta x,\Delta y}(q') = q$. Consider the circuit layer $t$ of the translated circuit $T_{\Delta x,\Delta y}(C)$. Since $q' \in Q(C)$, this layer acts on the qudit $T_{\Delta x,\Delta y}(q') = q$. Hence, both $C$ and $T_{\Delta x,\Delta y}(C)$ act on qudit $q$ in layer $t$, which implies that there exist integers $n,m$ for which the circuit patch $P_{n,m,l}(C)$ contains a gate collision.
\end{proof}

\subsection{Tileable circuits}

In this work, a \emph{tileable circuit} refers to the informal notion of a circuit intended for spatial and/or temporal repetition. When all qudits of a circuit patch or basis circuit are mapped to integers, the resulting circuit is formally no longer a circuit patch or basis circuit.
Additionally, circuits that lead to full gate collisions may still be considered tileable circuits. In such cases, it is implicitly understood that, depending on the context, these full gate collisions are resolved either by retaining only one of the colliding gates upon spatial repetition or by merging the two gates into a single operation. A gate collision between two gates $g$ and $g'$ is considered \emph{full} if the sets of qudits acted upon by $g$ and $g'$ are identical. Moreover, a circuit may be described as tileable or composed of circuit tiles even before explicitly identifying the specific basis circuit that will be repeated.

\subsection{Reseeding}\label{sec:reseeding}
In the remainder of this work, we will always route basis circuits to hardware with a connectivity graph generated by a basis graph. However, the number of seed numbers of the basis graph (usually equal to the number of vertices in the central cell of the basis graph) may significantly exceed the number of seed numbers in the basis circuit, or vice versa. Moreover, we anticipate that the quality of the routing solution generally improves when large circuit patches are mapped onto correspondingly large hardware patches. Consequently, we seek a general approach for routing circuit patches $P_{n,m,l}(C)$ to hardware generated by patches $P_{n',m'}(B)$ while ensuring that the solution can be tiled.

However, in formulating the routing problem for tileable circuits as a satisfiability modulo theories (SMT) problem (\sec{transpiler}) and in \thm{unique_seeds}, it is significantly more convenient to focus on basis graphs and basis circuits exclusively. Additionally, we have not yet established a formal method for generating lattice graphs from patches or lattice circuits from circuit patches. Fortunately, any lattice graph patch $P_{n',m'}(B)$ can be transformed into a basis graph via a process known as \emph{reseeding}~\cite{kattemolle2024edge}. A similar approach applies to circuit patches.

The key idea behind \emph{reseeding} a hardware patch $P_{n',m'}(B)$ generated by a basis graph $B$ is to assign a unique seed number $s'(v)$ to each vertex $v$ in $P_{n',m'}(B)$ that satisfies $0 \leq v_x < n'$ and $0 \leq v_y < m'$, and then map these vertices to $v' = (0,0,s'(v))$. This effectively places the vertices in the central cell of a new basis graph. Vertices $v$ that fall outside the range $0 \leq v_x < n'$ and $0 \leq v_y < m'$ must be mapped accordingly. The explicit details of this procedure are provided in Ref.~\cite{kattemolle2024edge}. Reseeding a circuit patch is essentially the same as reseeding a lattice graph patch; all qudits have to be updated according to essentially the same procedure.

\subsection{Gate scheduling}\label{sec:scheduling}

Quantum circuits are typically represented as an ordered list of gates, without explicitly specifying the time steps at which the gates should be executed. From this list, a valid time assignment can be obtained by scheduling the $i$th gate, $g^{(i)}$, at time step $i$. However, this results in unnecessarily deep circuits, as no gates are performed in parallel.  A method that leads to shallower circuits is to place each gate as early as possible.  Intuitively, this can be visualized as follows. Consider the circuit diagram and iteratively move each gate to the left until it encounters another gate that prevents further movement. This process continues until no gate can be shifted leftward. The resulting diagram is then partitioned into slices, each representing a layer of depth 1, with all gates in slice $i$ assigned to time step $i$.

However, directly applying this approach to basis circuits leads to invalidly scheduled basis circuits, as collisions may arise when the basis circuit is used to generate a circuit patch. To address this, we first impose periodic boundary conditions on the basis circuit. This is accomplished by setting $g.\tilde{q}.\tilde{x} = g.\tilde{q}.\tilde{y} = 0$ (and similarly for $\tilde{q}'$) for all gates $g$ in the circuit. After imposing the periodic boundary conditions, we then schedule the gates by assigning each gate $g$ to the earliest possible time $g.t$. Finally, we revert all qudits to their original values. This process is depicted in \fig{scheduling}.

In defining a specific basis circuit, it is often less confusing to explicitly specify the times at which gates in the basis circuit should be executed, rather than providing a list of gates without any time assignments. However, the resulting basis circuit may not always fully exploit gate parallelizability. In such cases, the process outlined above can be applied to reduce the circuit depth, and in this context, the process is more accurately referred to as gate \emph{re}scheduling.

\begin{figure}
\input{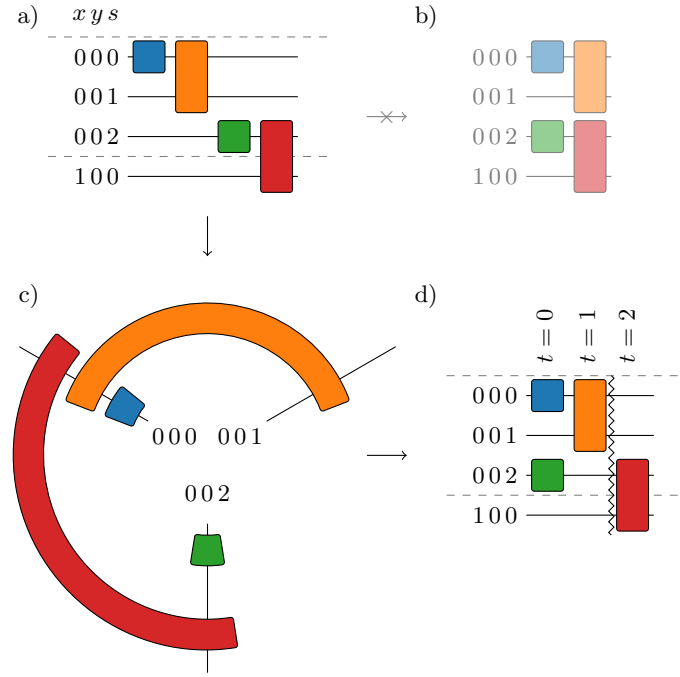}
\caption{\label{fig:scheduling} Gate scheduling. (a) A basis circuit is specified as a list of gates. Dashed lines delineate the central cell of the basis circuit. In this example, the basis circuit is defined on a 1D chain. (b) Treating the basis circuit as a standard circuit results in incorrect gate scheduling; the depicted circuit leads to gate collisions when the circuit is used to generate a circuit patch. (c) The correct scheduling method is to first impose periodic boundary conditions on the circuit. In this picture, time runs radially outward in integer steps. Gates are assigned times accordingly. (d) Periodic boundary conditions are lifted, yielding the correctly scheduled basis circuit.}
\end{figure}

\subsection{Gate dependencies}\label{sec:gate_dependencies}
A gate dependency $(A, B)$ indicates that gate $A$ must be executed before gate $B$. When routing a circuit, the gate dependencies of the input circuit need to be respected in the output circuit as well (unless they explicitly do not need to be respected, as is the case when gates commute or in Trotterized quantum simulation circuits). In tileable circuits, identifying dependencies requires special attention. For example, considering the circuit in \fig{scheduling}a as a standard circuit, there is no dependency between the two-qudit gates $g^{(1)}$ (the first gate) and $g^{(3)}$ (the third gate). However, when treating it as a basis circuit, we have the dependency $(g^{(1)}, g^{(3)})$. This dependency becomes apparent only in the wrapped circuit [\fig{scheduling}c].

To extract gate dependencies from the wrapped basis circuit, we use its directed acyclic graph (DAG) representation (\fig{DAG}). In this representation, gates correspond to vertices. Each vertex has incoming edges labeled with the qudits the gate acts on and outgoing edges labeled with the qudits it has acted on. The edges labeled with qudit $q$ are sourced from a vertex $\mrm{source}_q$ and terminate at a sink vertex $\mrm{sink}_q$. The DAG representation can be viewed as a formalized version of standard quantum circuit diagrams, where the absolute positioning of objects carries no intrinsic meaning.

\begin{figure}
    \definecolor{blue}{rgb}{0.12156862745098039, 0.4666666666666667, 0.7058823529411765}
\definecolor{orange}{rgb}{1.0, 0.4980392156862745, 0.054901960784313725}
\definecolor{green}{rgb}{0.17254901960784313, 0.6274509803921569, 0.17254901960784313}
\definecolor{red}{rgb}{0.8392156862745098, 0.15294117647058825, 0.1568627450980392}
\definecolor{purple}{rgb}{0.5803921568627451, 0.403921568627451, 0.7411764705882353}
\definecolor{brown}{rgb}{0.5490196078431373, 0.33725490196078434, 0.29411764705882354}
\definecolor{pink}{rgb}{0.8901960784313725, 0.4666666666666667, 0.7607843137254902}
\definecolor{pygray}{rgb}{0.4980392156862745, 0.4980392156862745, 0.4980392156862745}
\definecolor{olive}{rgb}{0.7372549019607844, 0.7411764705882353, 0.13333333333333333}
\definecolor{9}{rgb}{0.09019607843137255, 0.7450980392156863, 0.8117647058823529}
\definecolor{lightblue}{rgb}{0.9,.95,1}


\begin{tikzpicture}[font=\small]
    \usetikzlibrary{arrows.meta, positioning,shapes.geometric}
    \node[draw, shape=ellipse, minimum width=1.5cm, minimum height=0.8cm] (n1) at (2, -2) {001};
    \node[draw, shape=ellipse, fill=orange, minimum width=1.5cm, minimum height=0.8cm] (n2) at (4, -4) {{\color{white}$g^{(1)}[000,001]$}};
    \draw[->] (n1) -- (n2) node[midway, sloped, above] {001};

    \node[draw, fill=gray,shape=ellipse, minimum width=1.5cm, minimum height=0.8cm] (n5) at (6, -2) {{\color{white}001}};
    \draw[->] (n2) -- (n5) node[midway, sloped, above] {001};

    \node[draw, shape=ellipse, minimum width=1.5cm, minimum height=0.8cm] (n3) at (0, -8) {000};
    \node[draw, shape=ellipse, fill=blue, minimum width=1.5cm, minimum height=0.8cm] (n4) at (2, -6) {{\color{white}$g^{(0)}$}};
    \draw[->] (n3) -- (n4) node[midway, sloped, above] {000};

    \draw[->] (n4) -- (n2) node[midway, sloped, above] {000};

    \node[draw, shape=ellipse, fill=red, minimum width=1.5cm, minimum height=0.8cm] (n6) at (6, -6) {{\color{white}$g^{(3)}[000,002]$}};
    \draw[->] (n2) -- (n6) node[midway, sloped, above] {000};

    \node[draw, shape=ellipse, minimum width=1.5cm, minimum height=0.8cm, fill=gray] (n7) at (8, -4) {{\color{white}000}};
    \draw[->] (n6) -- (n7) node[midway, sloped, above] {000};

    \node[draw, shape=ellipse, minimum width=1.5cm, minimum height=0.8cm, fill=gray] (n8) at (8, -8) {{\color{white}002}};
    \draw[->] (n6) -- (n8) node[midway, sloped, above] {002};

    \node[draw, shape=ellipse, minimum width=1.5cm, minimum height=0.8cm] (n9) at (2, -10) {002};
    \node[draw, shape=ellipse, fill=green, minimum width=1.5cm, minimum height=0.8cm] (n10) at (4, -8) {{\color{white}$g^{(2)}[000,001]$}};
    \draw[->] (n9) -- (n10) node[midway, sloped, above] {002};

    \draw[->] (n10) -- (n6) node[midway, sloped, above] {002};

\end{tikzpicture}

    \caption{\label{fig:DAG} DAG representation of the circuit in \fig{scheduling}c. Source nodes are white and sink nodes are gray. Each gate node also displays the qudits it acts on as an ordered list. This information is necessary; otherwise, it would be impossible to determine, for example, in the case of a CNOT gate, which of the two incoming edges represents the control qudit and which represents the target qudit.}
\end{figure}
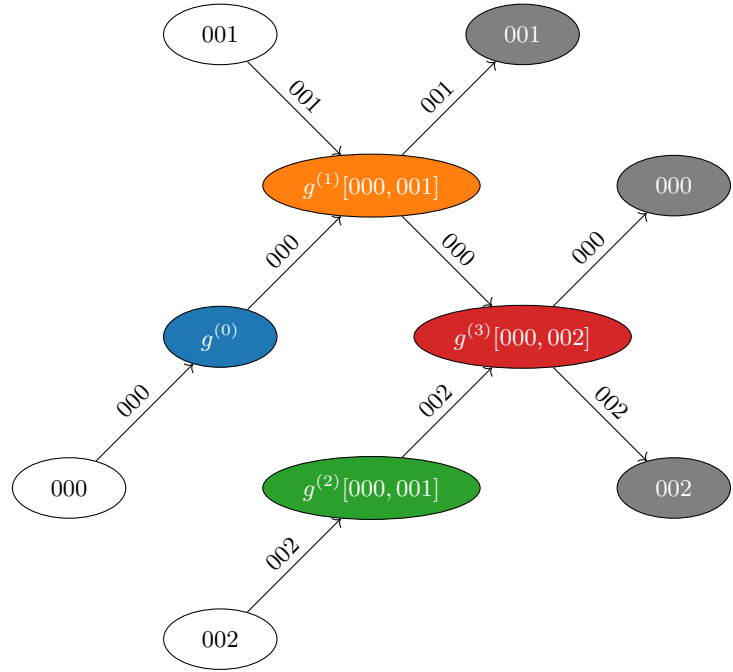

Denote the $j$th gate in the basis circuit by $g^{(j)}$ and its counterpart after imposing periodic boundary conditions by $g^{(j)'}$. From the DAG of the wrapped basis circuit, we iterate through all gates $g^{(j)'}$. For each $g^{(j)'}$, we add to a dependency list all pairs $(g^{(i)}, g^{(j)})$ such that $g^{(i)'}$ is a parent of $g^{(j)'}$ in the DAG. The dependency list thus contains only the `direct' dependencies. However, since gate dependency is transitive, ensuring that direct dependencies are respected in the routed circuit guarantees that all gate dependencies are respected.

\section{Tileable circuits from Trotterization}\label{sec:trotterization}

Consider the Hamiltonian
\begin{equation}\label{eq:klocal}
    H=\sum_{i_1,\ldots,i_k} H_{i_1,\ldots,i_k},
\end{equation}
where each term $H_{i_1,\ldots,i_k}$ acts on $k$ degrees of freedom $(i_1,\ldots, i_k)$ defined on an (infinite) Bravais lattice, with $k$ fixed throughout the lattice. Assume that $H$ is geometrically local, meaning that for any nonzero term $H_{i_1,\ldots,i_k}$, the indices $i_1,\ldots,i_k$ lie within a neighborhood where the maximum distance between any two degrees of freedom $i_j$ and $i_l$ is independent of system size. We assume this neighborhood corresponds to a supercell of $3 \times 3$ unit cells in two dimensions or $3^D$ unit cells in general dimension $D$. This assumption is without loss of generality because, given the geometric locality of $H$, the unit cell size can always be enlarged until the neighborhood fully contains all degrees of freedom within a fixed radius.

In quantum simulation by Trotterization, the system size is fixed to $n$ degrees of freedom, leading to the Hamiltonian
\begin{equation}\label{eq:patch}
    H=\sum_{\gamma=1}^\Gamma H_\gamma
\end{equation}
for some ordering $(H_1,\ldots,H_\Gamma)$ of $\Gamma$ of the terms in \eq{klocal}, where $H_\gamma$ denotes the $\gamma$th term in this ordering. An approximation $[\tilde U_p(t/r)]^r$ of the time-evolution unitary $U(t)=\ee^{-itH}$ is applied to an initial state. (Throughout, we use units where $\hbar=1$.) Here, $p$ determines the order of the approximation, and $r$ is the number of Trotter steps.

The first-order Lie-Trotter approximation is given by
\begin{equation}\label{eq:lie-trotter}
    \tilde U_1(t)=\ee^{-\ii t H_1} \ldots \ee^{-\ii t H_\Gamma}.
\end{equation}
Higher-order Suzuki formulas are defined recursively as
\begin{equation}\label{eq:suzuki}
    \begin{aligned}
        \tilde U_2(t) &= \ee^{-\ii \frac{t}{2} H_\Gamma} \ldots \ee^{-\ii \frac{t}{2} H_1}
                        \ee^{-\ii \frac{t}{2} H_1} \ldots \ee^{-\ii \frac{t}{2} H_\Gamma}, \\
        \tilde U_{2\ell}(t) &= [\tilde U_{2\ell-2}(u_\ell t)]^2
                           \tilde U_{2\ell-2}[(1 - 4u_\ell) t]
                           [\tilde U_{2\ell-2}(u_\ell t)]^2,
    \end{aligned}
\end{equation}
with $\tilde U_{2\ell-1}(t) = \tilde U_{2\ell}(t)$ and $u_\ell = 1/[4 - 4^{1/(2\ell-1)}]$~\cite{suzuki1991general}. We refer to $\tilde U_p(t)$ (for any $p=1,2,4,\ldots$) as a \emph{Trotter step}. Notably, $[\tilde U_{p}(\sfrac{t}{r})]^r$ consists of $k$-body gates, each of which can be compiled into a number of two-body gates independent of $n$. Given $H$ as in \eq{klocal} the circuit $[\tilde U_p(t)]^r$ is tileable (and can hence be described as a patch of a lattice circuit).

In Ref.~\cite{childs2021theory}, it was shown that
\begin{equation} \lVert\tilde U_p(t)-\ee^{-\ii tH}\rVert=O(\lvert\hspace{-.1em}\lvert\hspace{-.1em}\lvert H\rvert\hspace{-.1em}\rvert\hspace{-.1em}\rvert^p_1,\lVert H \rVert_1 t^{p+1}), \end{equation}
where
\begin{equation} \lvert\hspace{-.1em}\lvert\hspace{-.1em}\lvert H\rvert\hspace{-.1em}\rvert\hspace{-.1em}\rvert_1=\max_l\max_{i_l}\sum_{i_1,\ldots,i_{l-1},i_{l+1},\ldots,i_k}\lVert H_{i_1,\ldots,i_k}\rVert, \end{equation}
with $\Vert\cdot\rVert$ denoting the operator norm, and
\begin{equation} \lVert H \rVert_1=\sum_{\gamma}\lVert H_\gamma \rVert. \end{equation}
The sum in $\lvert\hspace{-.1em}\lvert\hspace{-.1em}\lvert H\rvert\hspace{-.1em}\rvert\hspace{-.1em}\rvert_1$ implicitly only includes terms also appearing in \eq{patch}.

We now consider what this implies for lattice circuits. Note that, due to locality, $\lvert\hspace{-.1em}\lvert\hspace{-.1em}\lvert H\rvert\hspace{-.1em}\rvert\hspace{-.1em}\rvert_1=O(1)$ and $\lVert H \rVert_1=O(n)$. Using the inequality
\begin{equation} \lVert [\tilde U_p(t/r)]^r-\ee^{-\ii t H}\rVert\leq r \lVert\tilde U_p(t/r)-\ee^{-\ii \frac{t}{r} H}\rVert, \end{equation}
the error from applying the $p$th-order Suzuki formula becomes
\begin{equation} \lVert [\tilde U_p(t/r)]^r-\ee^{-\ii t H}\rVert=O\left[nt\left(\sfrac{t}{r}\right)^p\right]. \end{equation}
Thus, the error is expected to be reduced with increasing $p$, but note that this also increases the number of gates per Trotter step.

In qudit routing, the temporal periodicity, or cyclicity, of the circuit $[\tilde U_p(t/r)]^r$ can be directly exploited. One first routes the circuit $\tilde U_p(t/r)$ and then repeats the routed version $r$ times. The only requirement is that the logical qudits $q$ return to their initial positions after the final layer of the routed circuit $\tilde U_p(t/r)$. (In our implementation~\cite{kattemolle2024quantile}, this can be achieved by setting the transpiler option \lin{cyclic = True}.)

Higher-order Suzuki formulas provide an advantage in routing. The second-order Suzuki formula consists of two applications of the first-order formula, where the second application has the reversed gate order. We may write this as
\begin{equation}
  \tilde U_2(t) = \tilde U_1^\dagger(-t/2) \tilde U_1(t/2).
\end{equation}
If $C(t)$ is the routed version of $\tilde U_1(t/2)$ (where the logical qudits need not return to their initial positions), then in the circuit $C^\dagger(-t/r)C(t/r)$
the logical qudits do return to their initial positions. Therefore, when routing $[\tilde U_2(t/r)]^r$, it suffices to route $\tilde U_1(t/r)$ without enforcing that the logical qudits return to their initial positions. That is,
\begin{equation}
  [\tilde U_2(t/r)]^r = [C^\dagger(-t/r)C(t/r)]^r.
\end{equation}
More generally, the same alternating structure of $\tilde U_1$ and its reversed circuit persists for any Suzuki order $p$. Thus, to route $r$ repetitions of $p$th order Trotter step, it suffices to route $\tilde U_1$, independent of $r$ and $p$, and without the need to impose that the logical qudits return to their initial positions.

\section{Examples}\label{sec:examples}

Here we treat some example circuits to which QuanTile is later applied. All circuits can also be found explicitly under \lin{circuits/} in the implementation \cite{kattemolle2024quantile}.

\subsection{Arbitrary two-local Hamiltonians}\label{sec:ATL}
Consider the Hamiltonian
\begin{equation}\label{eq:two-local}
 H=\sum_{i\in V(G)} H_{i} + \sum_{(i,j)\in E(G)}H_{ij},
\end{equation}
where $i$ ranges over all vertices $V(G)$ of some (possibly directed) graph $G$, $H_i$ is a Hermitian operator acting on qudit $i$ only, $(i,j)$ ranges over all edges of $G$ ($E(G)$), and where $H_{ij}$ is a Hermitian operator acting on qudits $i$ and $j$ only. For this general two-local Hamiltonian, first-order Trotterization (\sec{trotterization}) amounts to
\begin{align}
  \ee^{-\ii t H}\approx [\tilde U_1(t/r)]^r, && \tilde U_1(t/r)=\left(\prod_{(i,j)\in E(G)}U_{ij}\right)\left(\prod_{i\in V(G)}U_i\right),&& U_{i}\left(\sfrac{t}{r}\right)=\ee^{-\ii \frac{t}{r} H_{i}},&&U_{ij}\left(\sfrac{t}{r}\right)=\ee^{-\ii \frac{t}{r} H_{ij}}\label{eq:general_gate},
\end{align}
and similarly for the higher-order formulas. The single-qudit unitaries $U_i$ are irrelevant for the routing problem as they do not require any qudit connectivity. They can also be absorbed into preceding or subsequent two-qudit gates (except when $G$ contains isolated vertices, which are trivial to simulate).

The unitaries $U_{ij}$ are two-qudit operators. For the routing problem, only the structure of the circuit $\prod_{(i,j)} U_{ij}$ is relevant, treating each $U_{ij}$ as a monolithic gate. This structure depends solely on the graph $G$ and not on the interactions defined by $H_i, H_j, H_{ij}$. Consequently, when solving the routing problem for quantum simulation via Trotterization, it suffices to specify the graph $G$. The physical details of the two-local interactions—beyond the edges on which they act—can be incorporated later, even in postprocessing, by explicitly defining $U_{ij}$ in terms of $H_i, H_j,$ and $H_{ij}$, if needed.

The \emph{interactions} that can be considered include the following.
\begin{itemize}
\item[-] \emph{Ising.}
For the disordered transverse-field Ising model, the Hamiltonian terms are given by
\begin{align}
H_{i} = \la_i X_i, && H_{ij} = \la_{ij} Z_i Z_j,
\end{align}
where $\la_i$ and $\la_{ij}$ are site-dependent coupling parameters, and $X_i,Z_i$ are Pauli operators acting on qubit $i$. The model includes the standard homogeneous Ising model $\la_i=0, \la_{ij}=\mrm{const.}$ as a special case.
\item[-] \emph{XXZ.} For the (disordered) XXZ model,
\begin{align}\label{eq:XXZ}
 H_i=\la_i Z_i,  && H_{ij}=\la_{ij}[(X_iX_j+Y_iY_j)+\Delta_{ij} Z_iZ_j],
\end{align}
with anisotropy parameters $\Delta_{ij}$. It reduces to the Heisenberg model (sometimes called the XXX in this context) at $\la_i=0, \la_{ij}=\mrm{const.}, \Delta_{ij}=1$.

\item[-] \emph{Bose-Hubbard.}
The Bose-Hubbard model is given by
\begin{equation}\label{eq:BH}
  H = - T \sum_{(i,j)} (b_{i}^\dagger b_{j}^\nodagger+b_{j}^\dagger b_{i}^\nodagger)+\frac{U}{2}\sum_in_{i}(n_{i}-1)-\mu
  \sum_in_i,
\end{equation}
        with $T$ the tunneling constant, $b_i^\dagger$ ($b_i$) bosonic creation (annihilation) operators, $U$ the strength of the on-site repulsion, $n_i=b_i^\dagger b_i^\nodagger$ the number operator, and $\mu$ the chemical potential. It is essentially the Fermi-Hubbard model [\eq{FH}], but with bosonic operators $b_i$ instead of fermionic operators, and an added chemical potential, which is commonly omitted for the Fermi-Hubbard model. For all $i$, let us truncate the space that $b_i$ and $b_i^\dagger$ act on to $l$ levels. Then, $H$ becomes a Hamiltonian defined on qudits, where the space of qudit $i$ is spanned by $\{\ket m_i\}_{m=0}^{l-1}$. In this basis, the bosonic operators become matrices whose elements are given by $b_i\ket{m}_i=\sqrt{m}\ket{m-1}_i$, $b_i^\dagger=\sqrt{m+1}\ket{m+1}_i$ and the Hamiltonian becomes a two-local Hamiltonian of the form of \eq{two-local}. For qubits specifically, $b_i\mapsto \si^-_i:=(X_i+\ii Y_i)/2$, $b_i^\dagger\mapsto \si^+_i:=(X_i-\ii Y_i)/2$. Since $(\si_i^+)^2=0$, the resulting model describes hard-core bosons, where, similar to fermions, each bosonic mode cannot be occupied by more than one boson. For qubits, the above mapping gives
\begin{equation}
  H=-\frac{T}{2}\sum_{(i,j)}(X_iX_j+Y_iY_j)-\frac{\mu}{2}\sum_i(\id-Z_i),
\end{equation}
which is equivalent to a XXZ model.
\item[-]\emph{Kitaev.} In the Kitaev honeycomb model, spin-1/2 particles are placed on the vertices of the honeycomb lattice and interact via
\begin{equation}
H=\sum_{\xlink}X_iX_j+\sum_{\ylink}Y_iY_j+\sum_{\zlink}Z_iZ_j,
\end{equation}
where the sums are over all edges in the honeycomb lattice in the indicated direction.

\end{itemize}

The (infinite) \emph{graphs} we consider are the Archimedean lattices, the Laves lattices, and various other lattices, all depicted in \sec{basis_graphs_database}. Worth mentioning here are the J1J2-$G$ models, obtained from graphs $G$ by adding edges to geometrically nearest and second-nearest neighbors. J1J2J3-$G$ models additionally add edges to third-nearest neighbors. The notation J1 refers to the uniform strength of the nearest-neighbor interaction that is commonly assumed, with J2 and J3 following the same convention for the second- and third-nearest neighbors.

\subsection{Rule 54}
The Toffoli-gate model~\cite{gopalakrishnan2018facilitated}, or Rule 54 quantum cellular automaton~\cite{klobas2021exact}, is defined by the unitary gate
\begin{equation}\label{eq:rule54}
  \begin{tikzpicture}[scale=1.000000,x=1pt,y=1pt]
\filldraw[color=white] (0.000000, -7.500000) rectangle (221.000000, 37.500000);
\draw[color=black] (0.000000,30.000000) -- (221.000000,30.000000);
\draw[color=black] (0.000000,15.000000) -- (221.000000,15.000000);
\draw[color=black] (0.000000,0.000000) -- (221.000000,0.000000);
\draw (8.000000,30.000000) -- (8.000000,0.000000);
\begin{scope}
\draw[fill=blue] (8.000000, 15.000000) +(-45.000000:8.485281pt and 29.698485pt) -- +(45.000000:8.485281pt and 29.698485pt) -- +(135.000000:8.485281pt and 29.698485pt) -- +(225.000000:8.485281pt and 29.698485pt) -- cycle;
\clip (8.000000, 15.000000) +(-45.000000:8.485281pt and 29.698485pt) -- +(45.000000:8.485281pt and 29.698485pt) -- +(135.000000:8.485281pt and 29.698485pt) -- +(225.000000:8.485281pt and 29.698485pt) -- cycle;
\draw (8.000000, 15.000000) node {$\ $};
\end{scope}
\draw[fill=white,color=white] (18.000000, -6.000000) rectangle (33.000000, 36.000000);
\draw (25.500000, 15.000000) node {$:=$};
\draw (40.000000,30.000000) -- (40.000000,15.000000);
\filldraw (40.000000, 30.000000) circle(1.500000pt);
\begin{scope}
\draw[fill=white] (40.000000, 15.000000) circle(3.000000pt);
\clip (40.000000, 15.000000) circle(3.000000pt);
\draw (37.000000, 15.000000) -- (43.000000, 15.000000);
\draw (40.000000, 12.000000) -- (40.000000, 18.000000);
\end{scope}
\draw (50.000000,30.000000) -- (50.000000,0.000000);
\filldraw (50.000000, 30.000000) circle(1.500000pt);
\begin{scope}
\draw[fill=white] (50.000000, 15.000000) circle(3.000000pt);
\clip (50.000000, 15.000000) circle(3.000000pt);
\draw (47.000000, 15.000000) -- (53.000000, 15.000000);
\draw (50.000000, 12.000000) -- (50.000000, 18.000000);
\end{scope}
\filldraw (50.000000, 0.000000) circle(1.500000pt);
\draw (60.000000,15.000000) -- (60.000000,0.000000);
\filldraw (60.000000, 0.000000) circle(1.500000pt);
\begin{scope}
\draw[fill=white] (60.000000, 15.000000) circle(3.000000pt);
\clip (60.000000, 15.000000) circle(3.000000pt);
\draw (57.000000, 15.000000) -- (63.000000, 15.000000);
\draw (60.000000, 12.000000) -- (60.000000, 18.000000);
\end{scope}
\draw[fill=white,color=white] (67.000000, -6.000000) rectangle (82.000000, 36.000000);
\draw (74.500000, 15.000000) node {$=$};
\draw (103.500000, 37.500000) node[text width=144pt,above,text centered] {0};
\draw (103.500000,30.000000) -- (103.500000,15.000000);
\begin{scope}
\draw[fill=white] (103.500000, 15.000000) +(-45.000000:24.748737pt and 12.020815pt) -- +(45.000000:24.748737pt and 12.020815pt) -- +(135.000000:24.748737pt and 12.020815pt) -- +(225.000000:24.748737pt and 12.020815pt) -- cycle;
\clip (103.500000, 15.000000) +(-45.000000:24.748737pt and 12.020815pt) -- +(45.000000:24.748737pt and 12.020815pt) -- +(135.000000:24.748737pt and 12.020815pt) -- +(225.000000:24.748737pt and 12.020815pt) -- cycle;
\draw (103.500000, 15.000000) node {\raisebox{-10pt}{$\tilde{\mathrm{RX}}_{3\pi/2}$}};
\end{scope}
\filldraw (103.500000, 30.000000) circle(1.500000pt);
\draw (128.000000, 37.500000) node[text width=144pt,above,text centered] {1};
\draw (128.000000,30.000000) -- (128.000000,0.000000);
\filldraw (128.000000, 0.000000) circle(1.500000pt);
\begin{scope}
\draw[fill=white] (128.000000, 30.000000) circle(3.000000pt);
\clip (128.000000, 30.000000) circle(3.000000pt);
\draw (125.000000, 30.000000) -- (131.000000, 30.000000);
\draw (128.000000, 27.000000) -- (128.000000, 33.000000);
\end{scope}
\draw (152.500000, 37.500000) node[text width=144pt,above,text centered] {2};
\draw (152.500000,30.000000) -- (152.500000,15.000000);
\begin{scope}
\draw[fill=white] (152.500000, 15.000000) +(-45.000000:24.748737pt and 12.020815pt) -- +(45.000000:24.748737pt and 12.020815pt) -- +(135.000000:24.748737pt and 12.020815pt) -- +(225.000000:24.748737pt and 12.020815pt) -- cycle;
\clip (152.500000, 15.000000) +(-45.000000:24.748737pt and 12.020815pt) -- +(45.000000:24.748737pt and 12.020815pt) -- +(135.000000:24.748737pt and 12.020815pt) -- +(225.000000:24.748737pt and 12.020815pt) -- cycle;
\draw (152.500000, 15.000000) node {\raisebox{-10pt}{$\tilde{\mathrm{RX}}_{-\pi/2}$}};
\end{scope}
\filldraw (152.500000, 30.000000) circle(1.500000pt);
\draw (177.000000, 37.500000) node[text width=144pt,above,text centered] {3};
\draw (177.000000,30.000000) -- (177.000000,0.000000);
\filldraw (177.000000, 0.000000) circle(1.500000pt);
\begin{scope}
\draw[fill=white] (177.000000, 30.000000) circle(3.000000pt);
\clip (177.000000, 30.000000) circle(3.000000pt);
\draw (174.000000, 30.000000) -- (180.000000, 30.000000);
\draw (177.000000, 27.000000) -- (177.000000, 33.000000);
\end{scope}
\draw (201.500000, 37.500000) node[text width=144pt,above,text centered] {4};
\draw (201.500000,15.000000) -- (201.500000,0.000000);
\begin{scope}
\draw[fill=white] (201.500000, 15.000000) +(-45.000000:24.748737pt and 12.020815pt) -- +(45.000000:24.748737pt and 12.020815pt) -- +(135.000000:24.748737pt and 12.020815pt) -- +(225.000000:24.748737pt and 12.020815pt) -- cycle;
\clip (201.500000, 15.000000) +(-45.000000:24.748737pt and 12.020815pt) -- +(45.000000:24.748737pt and 12.020815pt) -- +(135.000000:24.748737pt and 12.020815pt) -- +(225.000000:24.748737pt and 12.020815pt) -- cycle;
\draw (201.500000, 15.000000) node {\raisebox{-10pt}{$\tilde{\mathrm{RX}}_{3\pi/2}$}};
\end{scope}
\filldraw (201.500000, 0.000000) circle(1.500000pt);
\end{tikzpicture}.
\end{equation}
To decompose the Toffoli gate, we used the standard decomposition of doubly-controlled gates in \fig{circuitids}c, with $U=X=\tilde{\mathrm{RX}}_{\pi}$ and $V=\sqrt{X}=\tilde{\mathrm{RX}}_{\pi/2}$ (with $U$ and $V$ as defined in  \fig{circuitids}c), where $\tilde{\mathrm{RX}}_\al=\ee^{\ii\al/2}\ee^{-\ii X/2}$. Additionally, in \eq{rule54}, we have merged the leftmost and rightmost CNOTs with the $C[\tilde{\mrm{RX}}_{\pi/2}]$ gates arising in the Toffoli decomposition, considering it a single two-qubit unitary. In the Rule 54 quantum cellular automaton, the gate on the left-hand side of \eq{rule54} constitutes the circuit

\begin{equation}
  \input{figures/rule54circ.tikz}\hspace{2em}.
\end{equation}

Here, the basis circuit that is repeated in space and time is depicted with darker blue gates. In the implementation~\cite{kattemolle2024quantile}, the two-qubit gates in this basis circuit, obtained after decomposition according to \eq{rule54}, are labeled sequentially by simply continuing the labeling in \eq{rule54}.

\subsection{2D Fermi-Hubbard}
The Fermi-Hubbard Hamiltonian is given by
\begin{equation}\label{eq:FH}
  H = - T \sum_{(i,j),\si} (a_{i\si}^\dagger a_{j\si}^\nodagger+a_{j\si}^\dagger a_{i\si}^\nodagger)+U\sum_in_{i\uparrow}n_{i\downarrow},
\end{equation}
where $a_{i\si}^\dagger$ ($a_{i\si}$) creates (annihilates) an electron at site $i$ and where $n_{i\si}=a_{i\si}^\dagger a_{i\si}^\nodagger$ is the number operator of the orbital with spin $\si$ at site $i$. The parameter $T$ sets the tunneling energy and $U$ sets the strength of the on-site repulsion. The first sum runs over all edges $(i,j)$ of the square lattice and the spin labels $\si\in\{\uparrow,\downarrow\}$, the second sum runs over all vertices of the square lattice.

The fermionic operators $a_{i\si}^\dagger$ and $a_{i\si}^\nodagger$ satisfy the fermionic anticommutation relations. The Fermi-Hubbard Hamiltonian can be mapped to a Hamiltonian defined on qubits by the Jordan-Wigner transformation, which, however, destroys the 2D periodic structure of $H$. Therefore, we instead use the hybrid encoding from Ref.~\cite{clinton2024towards} which is derived from the compact encoding~\cite{derby2021compact}.

In the hybrid encoding, one introduces Majorana operators $\gamma_{i\si}=a_{i\si}^\nodagger+a_{i\si}^\dagger$ and $\bar \gamma_{i\si}=(a_{i\si}^\nodagger-a_{i\si}^\dagger)/\ii$, which define edge and vertex operators
\begin{equation}
  E_{i\si,i'\si'}=-\ii\gamma_{i\si}\gamma_{i'\si'}, \quad V_{i\si}=-\ii\gamma_{i\si}\bar\gamma_{i\si}.
\end{equation}
The edge and vertex operators are Hermitian, they anticommute when they share one (and only one) index $i\si$, and commute otherwise. In terms of these operators, the Hamiltonian reads
\begin{equation}\label{eq:FHev}
  H\hat =\frac{\ii T}{2}\sum_{(i,j),\si}(E_{{i\si},j\si}V_{j\si}+V_{i\si}E_{{i\si},{j\si}})+\frac{U}{4}\sum_i(V_{i\uparrow}V_{i\downarrow}-V_{i\downarrow}-V_{i\uparrow}),
\end{equation}
where the hat indicates equality up to a term proportional to the identity. A final map maps the edge and vertex operators to Pauli words (tensor products of Pauli operators) so that these words satisfy the same anticommutation relations as the edge and vertex operators. This leads to the encoding in \fig{hybrid}.

To simulate the time evolution generated by the Fermi-Hubbard Hamiltonian using Trotterization, the remaining task is to construct a circuit cycle that implements $\ee^{- \ii P \frac{t}{r}}$ for some real number $\frac{t}{r}$ for all tensor Pauli words $P$ arising from the encoding (i.e., the Pauli words depicted in \fig{hybrid}). In the cases where $P$ is a single- or two-qubit Pauli word, $\ee^{- \ii P \frac{t}{r}}$ does not need to be compiled further to single- and two-qubit operations and they can hence be added to the circuit cycle directly. In the cases where $P$ is a three-qubit Pauli word, $\ee^{- \ii P \frac{t}{r}}$ can be compiled into two-qubit operations by the XYZ decomposition~\cite{sriluckshmy2023optimal}.

This XYZ decomposition reads
\begin{equation}\label{eq:XYZ}
  \ee^{\ii \al A}=\ee^{\ii \frac{\pi}{4} B}\ee^{\ii \al C}\ee^{-\ii \frac{\pi}{4} B}
\end{equation}
for unitary operators $A,B,C$ satisfying $A^2=\id$ and $[C,B]=2\ii A$~\cite{sriluckshmy2023optimal}.  Consider a three-qubit operator $A=A_0A_1A_2$, where $A_i$ is a Pauli operator acting on qubit $i$. Let $B=A_0B_1$ for some Pauli operator $B_1$ (possibly carrying an extra factor of $\pm$) and similarly $C=C_1A_2$.  Then, the condition $[C,B]=2\ii A$ becomes
\begin{equation}
  [C_1,B_1]=2\ii A _1,
\end{equation}
which can be satisfied for any $A_1$.

Applying the XYZ decomposition to the exponents of the three-qubit Pauli words in \fig{hybrid}, and adding the results to the circuit cycle (in addition to the exponents of the single- and two-qubit Pauli words), while ensuring tileability and while optimizing the order as to cancel as many gates as possible (using $\ee^{\ii \al A}\ee^{\ii \al A'}=\ee^{\ii \frac{\pi}{4} B}\ee^{\ii \al C}\ee^{-\ii \frac{\pi}{4} B}\ee^{\ii \frac{\pi}{4} B'}\ee^{\ii \al' C'}\ee^{-\ii \frac{\pi}{4} B'}=\ee^{\ii \frac{\pi}{4} B}\ee^{\ii \al C}\ee^{\ii \al' C'}\ee^{-\ii \frac{\pi}{4} B'}$ when $B=B'$), leads to the circuit in \fig{fermi-hubbard}.

\begin{figure}
  \centering
\newcommand{\blue}[1]{{\color{blue} #1}}
\newcommand{\orange}[1]{{\color{orange} #1}}
\newcommand{\red}[1]{{\color{red} #1}}
\begin{tikzpicture}[scale=1.5]
    \draw[green] (2,.5) -- (-.5,3);
    \draw[green] (-2,.5) -- (.5,3);
    \draw[green] (-2,.5) -- (.5,3);
    \draw[green] (-2,1.5) -- (.5,-1);
    \draw[green] (-.5,-1) -- (2,1.5);

    \node[outer sep=4pt] (0) at (0,0) {0};
    \node[outer sep=4pt] (1) at (1,1) {1};
    \node[outer sep=4pt] (2) at (0,2) {2};
    \node[outer sep=4pt] (3) at (-1,1) {3};
    \node[outer sep=4pt] (4) at (0,1) {4};

    \node[outer sep=4pt] (010) at (0,3) {0};
    \node[outer sep=4pt] (103) at (2,1) {3};

    \path (0) -- (1) coordinate[midway] (01);
    \path (1) -- (2) coordinate[midway] (12);
    \path (2) -- (3) coordinate[midway] (23);
    \path (3) -- (0) coordinate[midway] (30);

    \draw[->,gray] (0) to node[pos=-.15] {{\tiny \hspace{2pt} $\blue{X}\orange{\ii Y}$}} node[pos=1.15] {{\tiny $\blue{\ii X}\orange{Y}$}} (1);
    \draw[gray] (01) -- (4) node[pos=1.3] {{\tiny $\blue{X}\orange{X}$}};

    \draw[->,gray] (2) to node[pos=-.15] {{\hspace{13pt}\tiny $\blue{\mathsmaller -X}\orange{\mathsmaller -\ii Y}$}} node[pos=1.15] {{\tiny $\blue{\ii X}\orange{Y}$}} (1);
    \draw[gray] (12) -- (4) node[pos=1.3] {{\tiny $\blue{Y}\orange{Y}$}};

    \draw[->,gray] (2) to node[pos=-.15] {{\tiny $\blue{X}\orange{\ii Y}$}} node[pos=1.15] {{\tiny $\blue{\ii X}\orange{Y}$}} (3);
    \draw[gray] (23) -- (4) node[pos=1.3] {{\tiny $\blue{X}\orange{X}$}};

    \draw[->,gray] (0) to node[pos=-.15] {{\tiny \hspace{-7pt} $\blue{X}\!\orange{\ii Y}$}} node[pos=1.15] {{\tiny $\blue{\ii X}\orange{Y}$}} (3);
    \draw[gray] (30) -- (4) node[pos=1.3] {{\tiny $\blue{Y}\orange{Y}$}};

    \draw[->,gray] (1) to node[pos=-.15] {{\tiny $\red{Z}$}} node[pos=1.15] {{\tiny $\red{Z}$}} (103);
    \draw[->,gray] (010) to node[pos=-.15] {{\tiny $\red{Z}$}} node[pos=1.15] {{\tiny $\red{Z}$}} (2);

    \draw[dashed, lightgray] (1.5,-1) -- (1.5,3);
    \draw[dashed, lightgray] (-1.5,-1) -- (-1.5,3);
    \draw[dashed, lightgray] (-2,-.5) -- (2,-.5);
    \draw[dashed, lightgray] (-2,2.5) -- (2,2.5);

\end{tikzpicture}

    \caption{\label{fig:hybrid} Hybrid encoding of the 2D Fermi-Hubbard model. The green lines depict the 2D square lattice on which the original Hamiltonian [\eq{FH}] is defined. The orbitals [labelled $i$ in \eq{FH}] reside on the  intersections of the green lines (one orbital per intersection). Each orbital consists of two spin-orbitals [labelled $\si$ in \eq{FH}]. The spin-up orbitals are associated with qubits placed at nodes 1--4 inside the central unit cell. The spin-down orbitals arise in altered copies of this central cell, as explained later. A Pauli operator close to a qubit indicates it acts on that qubit and operators of the same color connected by a (multi)edge act simultaneously. Operators of the type $E_{{i\si},{j\si}}V_{j\si}$ from \eq{FHev} are mapped to the blue Pauli operators, while those of the type $V_{i\si}E_{{i\si},{j\si}}$ are mapped to the orange operators, and operators of the type $V_{i\uparrow}V_{i\downarrow}$ are mapped to the red operators. Operators of the type $V_{i\si}$ (not shown) map to $Z_{i\si}$, a single-qubit Pauli-$Z$ operator at $i\si$. These operators hit only those qubits already having a $Z$. The full encoding of a patch of $n$ by $m$ orbitals is obtained by repeating the shown structure in $n$ by $m$ unit cells, where the arrows (and the Pauli operators attached to them) are reversed in spatially alternating unit cells (which we call the even and odd cells). The even cells encode the spin-up orbitals, whereas the odd cells encode the spin-down orbitals. Additionally, a minus sign is added to both the reversed bottom left and top reversed top right multi-edge. Finally, the boundaries need to be patched with spin-orbitals so that only complete orbitals arise.
    }
\end{figure}
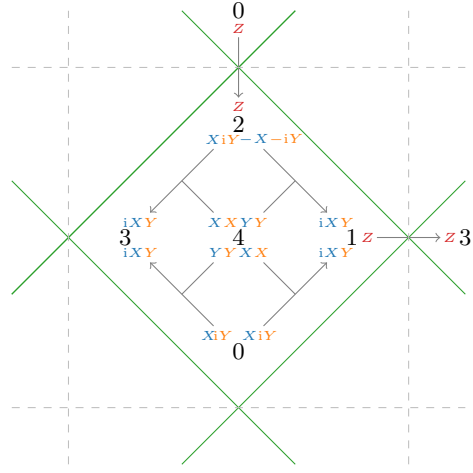

\begin{figure}
  \centering
    \input{figures/fermi-hubbard.tikz}

  \vspace{2em}

  \raisebox{-1.32em}{\begin{tikzpicture}[scale=1.000000,x=1pt,y=1pt]
\filldraw[color=white] (0.000000, -7.500000) rectangle (24.000000, 22.500000);
\draw[color=black] (0.000000,15.000000) -- (24.000000,15.000000);
\draw[color=black] (0.000000,0.000000) -- (24.000000,0.000000);
\draw (12.000000,15.000000) -- (12.000000,0.000000);
\begin{scope}
\draw[fill=white] (12.000000, 15.000000) +(-45.000000:8.485281pt and 8.485281pt) -- +(45.000000:8.485281pt and 8.485281pt) -- +(135.000000:8.485281pt and 8.485281pt) -- +(225.000000:8.485281pt and 8.485281pt) -- cycle;
\draw[very thick,solid] (12.000000, 15.000000) +(135.000000:8.485281pt and 8.485281pt) -- +(225.000000:8.485281pt and 8.485281pt);
\clip (12.000000, 15.000000) +(-45.000000:8.485281pt and 8.485281pt) -- +(45.000000:8.485281pt and 8.485281pt) -- +(135.000000:8.485281pt and 8.485281pt) -- +(225.000000:8.485281pt and 8.485281pt) -- cycle;
\draw (12.000000, 15.000000) node {$A$};
\end{scope}
\begin{scope}
\draw[fill=white] (12.000000, -0.000000) +(-45.000000:8.485281pt and 8.485281pt) -- +(45.000000:8.485281pt and 8.485281pt) -- +(135.000000:8.485281pt and 8.485281pt) -- +(225.000000:8.485281pt and 8.485281pt) -- cycle;
\draw[very thick,solid] (12.000000, -0.000000) +(135.000000:8.485281pt and 8.485281pt) -- +(225.000000:8.485281pt and 8.485281pt);
\clip (12.000000, -0.000000) +(-45.000000:8.485281pt and 8.485281pt) -- +(45.000000:8.485281pt and 8.485281pt) -- +(135.000000:8.485281pt and 8.485281pt) -- +(225.000000:8.485281pt and 8.485281pt) -- cycle;
\draw (12.000000, -0.000000) node {$B$};
\end{scope}
\end{tikzpicture}}$= \mathrm e^{-\mathrm i \frac{\pi}{4} A\otimes B}$ \hspace{8em} \raisebox{-1.32em}{\begin{tikzpicture}[scale=1.000000,x=1pt,y=1pt]
\filldraw[color=white] (0.000000, -7.500000) rectangle (24.000000, 22.500000);
\draw[color=black] (0.000000,15.000000) -- (24.000000,15.000000);
\draw[color=black] (0.000000,0.000000) -- (24.000000,0.000000);
\draw (12.000000,15.000000) -- (12.000000,0.000000);
\begin{scope}
\draw[fill=white] (12.000000, 15.000000) +(-45.000000:8.485281pt and 8.485281pt) -- +(45.000000:8.485281pt and 8.485281pt) -- +(135.000000:8.485281pt and 8.485281pt) -- +(225.000000:8.485281pt and 8.485281pt) -- cycle;
\draw[very thick,solid] (12.000000, 15.000000) +(-45.000000:8.485281pt and 8.485281pt) -- +(45.000000:8.485281pt and 8.485281pt);
\clip (12.000000, 15.000000) +(-45.000000:8.485281pt and 8.485281pt) -- +(45.000000:8.485281pt and 8.485281pt) -- +(135.000000:8.485281pt and 8.485281pt) -- +(225.000000:8.485281pt and 8.485281pt) -- cycle;
\draw (12.000000, 15.000000) node {$A$};
\end{scope}
\begin{scope}
\draw[fill=white] (12.000000, -0.000000) +(-45.000000:8.485281pt and 8.485281pt) -- +(45.000000:8.485281pt and 8.485281pt) -- +(135.000000:8.485281pt and 8.485281pt) -- +(225.000000:8.485281pt and 8.485281pt) -- cycle;
\draw[very thick,solid] (12.000000, -0.000000) +(-45.000000:8.485281pt and 8.485281pt) -- +(45.000000:8.485281pt and 8.485281pt);
\clip (12.000000, -0.000000) +(-45.000000:8.485281pt and 8.485281pt) -- +(45.000000:8.485281pt and 8.485281pt) -- +(135.000000:8.485281pt and 8.485281pt) -- +(225.000000:8.485281pt and 8.485281pt) -- cycle;
\draw (12.000000, -0.000000) node {$B$};
\end{scope}
\end{tikzpicture}}$= \mathrm e^{ \mathrm i \frac{\pi}{4} A\otimes B}$ \hspace{8em} \raisebox{-1.32em}{\begin{tikzpicture}[scale=1.000000,x=1pt,y=1pt]
\filldraw[color=white] (0.000000, -7.500000) rectangle (32.000000, 22.500000);
\draw[color=black] (0.000000,15.000000) -- (32.000000,15.000000);
\draw[color=black] (0.000000,0.000000) -- (32.000000,0.000000);
\draw (16.000000,15.000000) -- (16.000000,0.000000);
\begin{scope}
\draw[fill=white] (16.000000, 15.000000) +(-45.000000:14.142136pt and 8.485281pt) -- +(45.000000:14.142136pt and 8.485281pt) -- +(135.000000:14.142136pt and 8.485281pt) -- +(225.000000:14.142136pt and 8.485281pt) -- cycle;
\clip (16.000000, 15.000000) +(-45.000000:14.142136pt and 8.485281pt) -- +(45.000000:14.142136pt and 8.485281pt) -- +(135.000000:14.142136pt and 8.485281pt) -- +(225.000000:14.142136pt and 8.485281pt) -- cycle;
\draw (16.000000, 15.000000) node {$A$};
\end{scope}
\begin{scope}
\draw[fill=white] (16.000000, -0.000000) +(-45.000000:14.142136pt and 8.485281pt) -- +(45.000000:14.142136pt and 8.485281pt) -- +(135.000000:14.142136pt and 8.485281pt) -- +(225.000000:14.142136pt and 8.485281pt) -- cycle;
\clip (16.000000, -0.000000) +(-45.000000:14.142136pt and 8.485281pt) -- +(45.000000:14.142136pt and 8.485281pt) -- +(135.000000:14.142136pt and 8.485281pt) -- +(225.000000:14.142136pt and 8.485281pt) -- cycle;
\draw (16.000000, -0.000000) node {$B$};
\end{scope}
\end{tikzpicture}}$= \mathrm e^{\mathrm i A\otimes B}$

    \caption{\label{fig:fermi-hubbard} Basis circuit for the implementation of the time evolution along the Fermi-Hubbard Hamiltonian [\eq{FH}], up to Trotter error, with $\al_T=\frac{T}{2}\frac{t}{r}$ and $\al_U=-\frac{U}{4}\frac{t}{r}$. The circuit comprises 16 general 2-qubit gates and has a depth of 15. The shown circuit is for the even cells, the circuit for the odd cells is obtained by sending $\alpha_T\rightarrow -\alpha_T$ for gates 1,4,8,11. The circuit is tileable.
    }
\end{figure}


\subsection{Rokhsar-Kivelson}\label{sec:RK}

The Rokhsar-Kivelson Hamiltonian is given by
\begin{equation}\label{eq:RK}
  H=-J \sum_{\square}[(\ketbra{\vdimer}{\hdimer}+\ketbra{\hdimer}{\vdimer})-\la(\ketbra{\vdimer}{\vdimer}+\ketbra{\hdimer}{\hdimer})]
\end{equation}
Here, a dimer, $\singledimer$, stands for a singlet state of two spin-1/2 particles located at the vertices and the sum runs over all plaquettes of the square grid.

With each edge between nearest neighbors of the square lattice we associate a two-dimensional Hilbert space spanned by $\ket 0$, indicating there is no dimer along that edge, and $\ket 1$, indicating there is a dimer along that edge. With this identification, the Hamiltonian is equivalent to
\begin{align}\label{eq:HRK}
  H=\sum_{\square}H_\square, &&
  H_\square=-J[P_\square^\nodagger+P_\square^\dagger-\la (P_\square^\nodagger+P_\square^\dagger)^2].
\end{align}
Here,
\begin{align}\label{eq:usquare}
P_\square=\si_b^-\si_r^+\si_t^-\si_l^+, && \si^{+}=(X-\ii Y)/2=\ketbra{1}{0},&& \si^{-}=(X+\ii Y)/2=\ketbra{0}{1},
\end{align}
where we indexed the edges of a plaquette by $b$ (bottom), $r$ (right), $t$ (top), and $l$ (left). This Hamiltonian defines a $U(1)$ lattice gauge theory.


By Trotterization, $
  \ee^{-\ii t H}\approx \left( \prod_\square \ee^{-i \frac{t}{r} H_\square}\right)^r$.
We now show two methods for decomposing $\ee^{-i \frac{t}{r} H_\square}$ exactly using one- and two-qubit gates. The resulting spatial circuit structure, which is the same for both decomposition techniques, is displayed in \fig{checkerboard}.

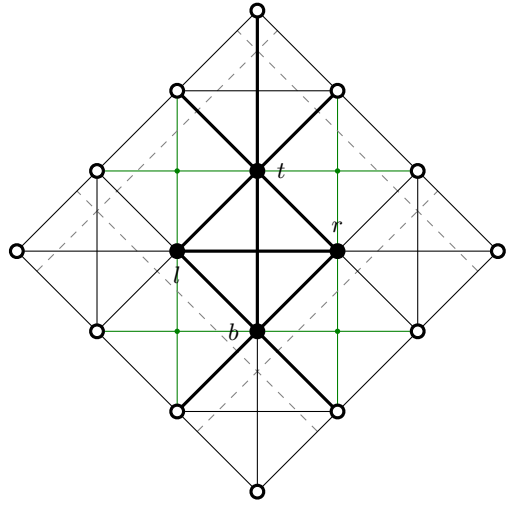
\begin{figure}
  \centering
\begin{tikzpicture}[scale=1.5, rotate=45]
\def\rows{3}
\def\cols{3}

\definecolor{lightgreen}{rgb}{0.0, 0.5, 0.0}
\draw[lightgreen] (2,0) -- (0,2);
\draw[lightgreen] (3,1) -- (1,3);

\draw[lightgreen] (0,1) -- (2,3);
\draw[lightgreen] (1,0) -- (3,2);

\node[fill=lightgreen, circle, minimum size=2pt, inner sep=0pt] at (0.5,1.5) {};
\node[fill=lightgreen, circle, minimum size=2pt, inner sep=0pt] at (1.5,0.5) {};
\node[fill=lightgreen, circle, minimum size=2pt, inner sep=0pt] at (1.5,2.5) {};
\node[fill=lightgreen, circle, minimum size=2pt, inner sep=0pt] at (2.5,1.5) {};

\foreach \i in {0,...,\rows}
    \foreach \j in {0,...,\cols} {
        \fill (\i,\j) circle (2pt); 

        \ifnum\i=1
            \ifnum\j=1
                \node at (\i - 0.15,\j + 0.15) {$b$};
            \fi
            \ifnum\j=2
                \node at (\i - 0.15,\j - 0.15) {$l$};
            \fi
        \fi
        \ifnum\i=2
            \ifnum\j=1
                \node at (\i + 0.15,\j + 0.15) {$r$};
            \fi
            \ifnum\j=2
                \node at (\i + 0.15,\j - 0.15) {$t$};
            \fi
d        \fi
    }

\draw[dashed,gray] (0.75,0) -- (0.75,3);
\draw[dashed,gray] (2.75,0) -- (2.75,3);
\draw[dashed,gray] (0,0.75) -- (3,0.75);
\draw[dashed,gray] (0,2.75) -- (3,2.75);

\foreach \i in {0,...,\rows}
\draw (\i,0) -- (\i,\cols); 

\foreach \j in {0,...,\cols}
\draw (0,\j) -- (\rows,\j); 

\draw (0,0) -- (1,1);
\draw (1,0) -- (0,1);

\draw (0,2) -- (1,3);
\draw (1,2) -- (0,3);

\draw (2,0) -- (3,1);
\draw (3,0) -- (2,1);

\draw (2,2) -- (3,3);
\draw (3,2) -- (2,3);

\draw[very thick] (2,1) -- (2,3);
\draw[very thick] (1,2) -- (3,2);
\draw[very thick] (1,1) -- (3,3);
\draw[very thick] (2,1) -- (1,2);
\draw[very thick] (1,2) -- (1,0);
\draw[very thick] (0,1) -- (2,1);

\node[draw, circle, minimum size=4pt, inner sep=0pt,fill=white] at (0,0) {};
\node[draw, circle, minimum size=4pt, inner sep=0pt,fill=white] at (1,0) {};
\node[draw, circle, minimum size=4pt, inner sep=0pt,fill=white] at (2,0) {};
\node[draw, circle, minimum size=4pt, inner sep=0pt,fill=white] at (3,0) {};

\node[draw, circle, minimum size=4pt, inner sep=0pt,fill=white] at (0,1) {};
\node[draw, circle, minimum size=4pt, inner sep=0pt,fill=white] at (0,2) {};
\node[draw, circle, minimum size=4pt, inner sep=0pt,fill=white] at (0,3) {};

\node[draw, circle, minimum size=4pt, inner sep=0pt,fill=white] at (1,3) {};
\node[draw, circle, minimum size=4pt, inner sep=0pt,fill=white] at (2,3) {};
\node[draw, circle, minimum size=4pt, inner sep=0pt,fill=white] at (3,3) {};

\node[draw, circle, minimum size=4pt, inner sep=0pt,fill=white] at (3,1) {};
\node[draw, circle, minimum size=4pt, inner sep=0pt,fill=white] at (3,2) {};
\node[draw, circle, minimum size=4pt, inner sep=0pt,fill=white] at (3,3) {};

\end{tikzpicture}
  \caption{\label{fig:checkerboard}
Spatial structure of the tileable circuit for simulating the Rokhsar-Kivelson model, implemented using either the XYZ decomposition or the multi-controlled unitaries technique. Bold edges indicate the basis circuit, dashed lines mark unit cells, and green lines represent the original square lattice on which the Rokhsar-Kivelson model is defined [\eq{RK}]. The time evolution for the central plaquette operator, acting on qubits $b,r,t,l$, follows the circuits in Fig.~\ref{fig:circxyz} or Fig.~\ref{fig:circmcu}, depending on whether the XYZ decomposition or the multi-controlled unitaries technique is used. The circuit for the second plaquette is distributed: the two bold edges at the bottom form part of the circuit for the bottom plaquette, while the three bold edges on the top form part of the circuit for the top plaquette. It is only when the basis circuit is copied and translated up by one unit cell that the entire circuit for the top plaquette is formed. This happens automatically when, e.g., a patch of 3 by 3 basis circuits is formed. We break up the circuit of one plaquette in the basis circuit in this way because otherwise there would be gates in the basis circuit fully outside the central cell, which we consider unnatural and which we assumed, without loss of generality, not to be the case in our implementation \cite{kattemolle2024quantile}. The explicit basis circuit for both decomposition techniques is available under \texttt{circuits/rokhsar-kivelson} in the implementation.}
\end{figure}

\subsubsection{XYZ decomposition}
Using \eq{usquare}, we can express $H$ as a linear combination of Pauli words. Since in $P_\square^\nodagger+P_\square^\dagger$ all $Y$ tensor factors carry a factor of $\ii$, in $P_\square^\nodagger+P_\square^\dagger$ only terms containing an even number of $Y$s remain, giving
\begin{align}\label{eq:uud}
  P_\square^\nodagger+P_\square^\dagger&=\frac{1}{16}(X_b+\ii Y_b)(X_r-\ii Y_r)(X_t+\ii Y_t)(X_l-\ii Y_l)+h.c.\nn\\  &=\frac{1}{8}(X_bX_rX_tX_l+X_bY_rY_tX_l-Y_bX_rY_tX_l+Y_bY_rX_tX_l+Y_bX_rX_tY_l   +Y_bY_rY_tY_l+X_bX_rY_tY_l-X_bY_rX_tY_l).
\end{align}
For the term in $H$ proportional to $\la$, note $\si^\pm\si^\pm=0$ and $\si^\mp\si^\pm=(\id \pm Z)$. Similarly to before, we find that only terms with an even number of $Z$s remain,
\begin{align}\label{eq:uuds}   (P_\square^\nodagger+P_\square^\dagger)^2=\frac{1}{8}(\id_b\id_r\id_t\id_l+ Z_bZ_rZ_tZ_l-Z_bZ_r\id_t\id_l- \id_b \id_rZ_t  Z_l+Z_b\id_rZ_t\id_l+ \id_bZ_r\id_tZ_l-Z_b\id_r\id_tZ_l- \id_b Z_r Z_t  \id_l).
\end{align}
All 18 terms in Eqs. \eqref{eq:uud} and \eqref{eq:uuds} commute. Thus,
\begin{equation}
\ee^{-\ii \frac{t}{r} H_\square}=\prod_{i} \ee^{\ii \al_i P^{(i)}},
\end{equation}
where $\{P^{(i)}\}_i$ are the 4-qubit Pauli words from Eqs. \eqref{eq:uud} and \eqref{eq:uuds}, and where $\alpha_i$ is either $0$, $\pm \frac{J t}{8 r }$ or $\pm \frac{J \la t}{8 r}$, depending on $i$, as determined by Eqs.~\eqref{eq:HRK}, \eqref{eq:uud} and \eqref{eq:uuds}.

Using \eq{XYZ} recursively for $\ee^{\ii \alpha C}$, we have
\begin{equation}
 \ee^{\ii \al A}=\ee^{\ii \frac{\pi}{4} B}\ee^{\ii \frac{\pi}{4} D}\ee^{\ii \al E}\ee^{-\ii \frac{\pi}{4} D}\ee^{-\ii \frac{\pi}{4} B},
\end{equation}
where $[E, D]=2\ii C$. Now, let $A=A_0A_1A_2A_3$ with Pauli operators $A_i$, $B=A_0B_1$ for some Pauli operator $B_1$ (possibly carrying an extra factor of $\pm$), and similarly $D=D_2A_3$ and $E=E_1E_2$. Then, the condition $[C,B]=2\ii A$ becomes
\begin{equation}\label{eq:aleph}
  [E_1,B_1]= 2 \ii A_1, \quad [E_2,D_2]=2\ii A_2.
\end{equation}
That is, we can decompose $\ee^{\ii \alpha A}$ into 5 two-qubit gates, giving a circuit of depth 3, by choosing Pauli operators $\pm B_1$, $\pm E_1$, $\pm D_2$, $\pm E_2$ such that \eq{aleph} holds, which is possible for any 4-qubit Pauli word $A$. This gives the exact circuit in \fig{circxyz} for the time evolution along a single plaquette, up to a global phase of $\ee^{-\ii J t \la / (8r)}$, with $\alpha_J=\frac{J t}{8r}$, and $\al_{J\la}=-\frac{J \la t}{8 r}$. The subcircuits for the factors $\ee^{\ii \al_i P^{(i)}}$ appear in the order of the Pauli words of Eqs. \eqref{eq:uud} and \eqref{eq:uuds}. We have omitted the time evolution along $\id_b\id_r\id_t\id_l$.

\begin{figure}
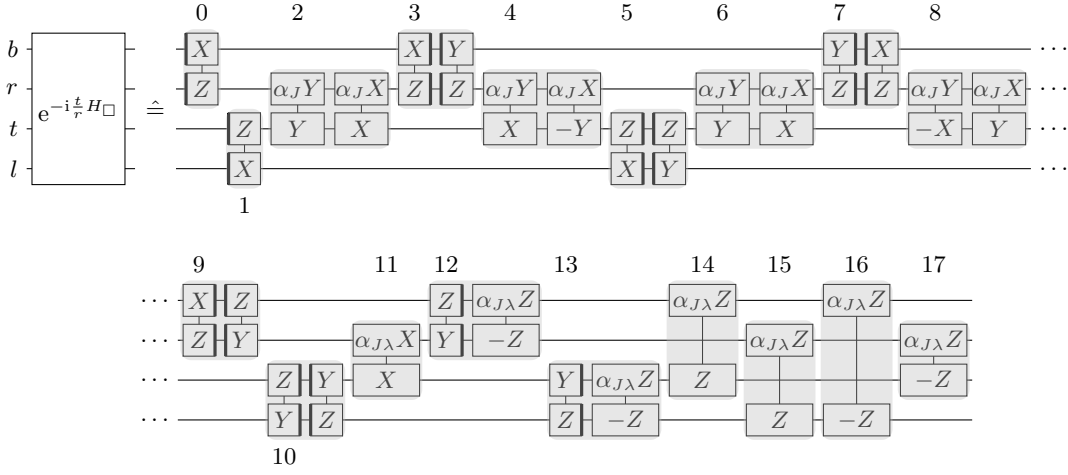

  \input{figures/circ.tikz}

  \vspace{1em}

  \input{figures/circ0-1.tikz}

  \caption{\label{fig:circxyz} Exact circuit for the time evolution along one plaquette operator $H_\square$ of the Rokhsar-Kivelson model, using the XYZ-decomposition, with $\al_J=\frac{J}{8}\frac{t}{r}$ and $\al_{J\la}=-\frac{J\la}{8}\frac{t}{r}$. (See \fig{fermi-hubbard} for the definition of the gate notation.) The circuit comprises 18 two-qubit gates, labeled 0--17. The circuits for two disjoint plaquettes can always be applied simultaneously. For two distinct plaquettes sharing a qubit, gates 0--13 can always be applied simultaneously, but not gates 14--17. The \emph{basis circuit} for one Trotter step of the time evolution along the Rokhsar-Kivelson Hamiltonian consists of the above time evolution operator applied to two plaquettes sharing a qubit, as depicted in \fig{checkerboard}. }
\end{figure}

Gates 0--13 of the first plaquette of the basis circuit (see \fig{checkerboard}) can be applied in parallel with gates 0--13 of the second plaquette. In the implementation, we write these gates as \lin{G_plaquette_0_unitary_x} and \lin{G_plaquette_1_unitary_x} respectively, with $x\in\{0,\ldots,13\}$. Thereafter, the remaining gates 14--17 are applied first for the first plaquette  (\lin{G_plaquette_0_unitary_x}, $x\in\{14,\ldots,17\}$), and then for the second plaquette (\lin{G_plaquette_1_unitary_x}, $x\in\{14,\ldots,17\}$). The resulting circuit, containing the exact circuit for $\ee^{-\ii\frac{t}{r} H_\square}$ for both plaquettes, has depth~18.

\subsubsection{Multi-controlled unitaries}
By \eq{usquare},
\begin{equation}
  H_\square=-J[\mathcal X-\la \mc I].
\end{equation}
Here, in the qubit ordering $b,r,t,l$,
\begin{equation}
  \mc X= \ketbra{0101}{1010}+\ketbra{1010}{0101}
\end{equation}
acts on the subspace spanned by $\ket{0101}$ and $\ket{1010}$ as the Pauli $X$-operator, and
\begin{equation}
  \mc I= \ketbra{0101}{0101}+\ketbra{1010}{1010}
\end{equation}
acts as the identity on that same subspace. Because $\mc X$ and $\mc I$  commute,
\begin{equation}
  \ee^{-\ii \frac{t}{r} H_\square}=\ee^{- \ii \theta \mc X/2}\ee^{\ii \al \mc I},
\end{equation}
with $\theta=-2 J \frac{t}{r}$ and $\alpha=-J\la \frac{t}{r}$.

The unitary $\ee^{- \ii \theta \mc X/2}$ is essentially a Pauli-$X$ rotation within the subspace spanned by $\ket{0101}$ and $\ket{1010}$. A similar gate was implemented in Ref.~\cite{yordanov2020efficient}. Along the lines of Ref.~\cite{yordanov2020efficient}, note the subspace spanned by $\ket{0101}$ and $\ket{1010}$ is the only subspace in the Hilbert space of four qubits where qubits $b,r$ and $t,l$ have uneven parity, respectively, and where, furthermore, qubits $r,t$ have even parity. Thus, $\ee^{- \ii \theta \mc X/2}$ can be implemented by a circuit that computes and stores these parities in qubits $r,t,l$, an RX gate on qubit $b$ controlled by $r,t,l$, and a circuit that undoes the circuit that computed and stored the parities. This gives the circuit in Fig. (\ref{fig:circuitids}a). One implementation of the triple-controlled RX is given by Fig. (\ref{fig:circuitids}b), which is a variation of the circuit for this gate in Ref.~\cite{yordanov2020efficient}.

The unitary $\ee^{\ii \al \mc I}$ exclusively multiplies the space spanned by $\ket{0101}$ and $\ket{1010}$ with a phase $\ee^{\ii \al}$. It can thus be implemented by the circuit on the right-hand side of Fig. (\ref{fig:circuitids}a) after $\mrm{RX}_\theta$ is replaced by $\ee^{\ii \al}\id$. However, we cannot similarly replace $\mrm{RX}_\theta$ by $\ee^{\ii \al}\id$ in Fig. (\ref{fig:circuitids}b) to implement the triply-controlled $\ee^{\ii\al}\id$ gate because the single-qubit $\ee^{\ii\al}\id$ gates occurring in this circuit result only in a global phase. For the implementation of the triply-controlled  $\ee^{\ii \al}\id$ gate we therefore resort to a standard and general method for implementing doubly-controlled unitaries $U$ [Fig. (\ref{fig:circuitids}c)], and use this method with $U$ a controlled $\ee^{\ii \al}\id$ gate, which leads to the circuit in Fig. (\ref{fig:circuitids}d). Putting this all together leads to the exact circuit for $\ee^{-\ii \frac{t}{r} H_\square}$ in \fig{circmcu}.

\newcommand{\circfigline}[2]{$\text{#1)} \hfill \begin{array}{c}\input{figures/circ#2.tikz}\end{array} \hfill$}

\begin{figure}
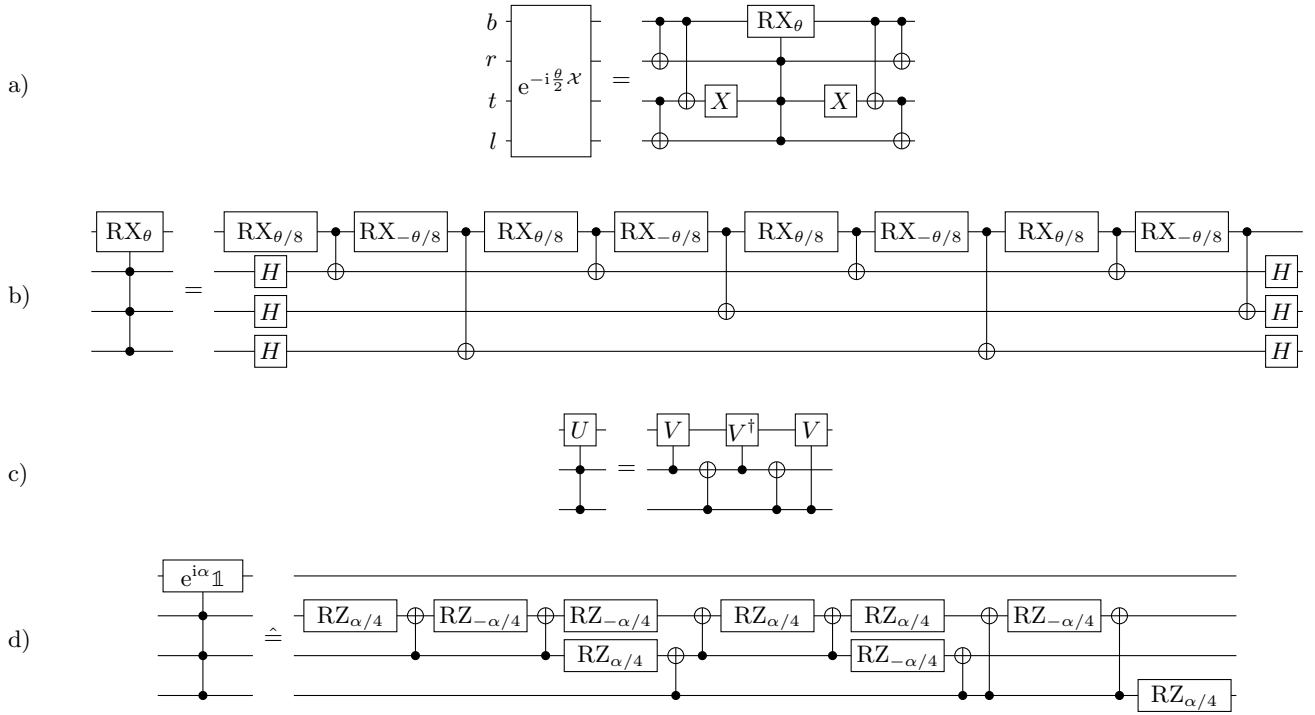

 $\text{a)} \hfill \begin{array}{c}\begin{tikzpicture}[scale=1.000000,x=1pt,y=1pt]
\filldraw[color=white] (0.000000, -7.500000) rectangle (154.000000, 52.500000);
\draw[color=black] (0.000000,45.000000) -- (154.000000,45.000000);
\draw[color=black] (0.000000,45.000000) node[left] {$b$};
\draw[color=black] (0.000000,30.000000) -- (154.000000,30.000000);
\draw[color=black] (0.000000,30.000000) node[left] {$r$};
\draw[color=black] (0.000000,15.000000) -- (154.000000,15.000000);
\draw[color=black] (0.000000,15.000000) node[left] {$t$};
\draw[color=black] (0.000000,0.000000) -- (154.000000,0.000000);
\draw[color=black] (0.000000,0.000000) node[left] {$l$};
\draw (17.000000,45.000000) -- (17.000000,0.000000);
\begin{scope}
\draw[fill=white] (17.000000, 22.500000) +(-45.000000:21.213203pt and 40.305087pt) -- +(45.000000:21.213203pt and 40.305087pt) -- +(135.000000:21.213203pt and 40.305087pt) -- +(225.000000:21.213203pt and 40.305087pt) -- cycle;
\clip (17.000000, 22.500000) +(-45.000000:21.213203pt and 40.305087pt) -- +(45.000000:21.213203pt and 40.305087pt) -- +(135.000000:21.213203pt and 40.305087pt) -- +(225.000000:21.213203pt and 40.305087pt) -- cycle;
\draw (17.000000, 22.500000) node {$\mathrm e^{-\mathrm i \frac{\theta}{2} \mathcal X}$};
\end{scope}
\draw[fill=white,color=white] (36.000000, -6.000000) rectangle (51.000000, 51.000000);
\draw (43.500000, 22.500000) node {$=$};
\draw (58.000000,45.000000) -- (58.000000,30.000000);
\filldraw (58.000000, 45.000000) circle(1.500000pt);
\begin{scope}
\draw[fill=white] (58.000000, 30.000000) circle(3.000000pt);
\clip (58.000000, 30.000000) circle(3.000000pt);
\draw (55.000000, 30.000000) -- (61.000000, 30.000000);
\draw (58.000000, 27.000000) -- (58.000000, 33.000000);
\end{scope}
\draw (58.000000,15.000000) -- (58.000000,0.000000);
\filldraw (58.000000, 15.000000) circle(1.500000pt);
\begin{scope}
\draw[fill=white] (58.000000, 0.000000) circle(3.000000pt);
\clip (58.000000, 0.000000) circle(3.000000pt);
\draw (55.000000, 0.000000) -- (61.000000, 0.000000);
\draw (58.000000, -3.000000) -- (58.000000, 3.000000);
\end{scope}
\draw (68.000000,45.000000) -- (68.000000,15.000000);
\filldraw (68.000000, 45.000000) circle(1.500000pt);
\begin{scope}
\draw[fill=white] (68.000000, 15.000000) circle(3.000000pt);
\clip (68.000000, 15.000000) circle(3.000000pt);
\draw (65.000000, 15.000000) -- (71.000000, 15.000000);
\draw (68.000000, 12.000000) -- (68.000000, 18.000000);
\end{scope}
\begin{scope}
\draw[fill=white] (81.000000, 15.000000) +(-45.000000:8.485281pt and 8.485281pt) -- +(45.000000:8.485281pt and 8.485281pt) -- +(135.000000:8.485281pt and 8.485281pt) -- +(225.000000:8.485281pt and 8.485281pt) -- cycle;
\clip (81.000000, 15.000000) +(-45.000000:8.485281pt and 8.485281pt) -- +(45.000000:8.485281pt and 8.485281pt) -- +(135.000000:8.485281pt and 8.485281pt) -- +(225.000000:8.485281pt and 8.485281pt) -- cycle;
\draw (81.000000, 15.000000) node {$X$};
\end{scope}
\draw (103.500000,45.000000) -- (103.500000,0.000000);
\begin{scope}
\draw[fill=white] (103.500000, 45.000000) +(-45.000000:17.677670pt and 8.485281pt) -- +(45.000000:17.677670pt and 8.485281pt) -- +(135.000000:17.677670pt and 8.485281pt) -- +(225.000000:17.677670pt and 8.485281pt) -- cycle;
\clip (103.500000, 45.000000) +(-45.000000:17.677670pt and 8.485281pt) -- +(45.000000:17.677670pt and 8.485281pt) -- +(135.000000:17.677670pt and 8.485281pt) -- +(225.000000:17.677670pt and 8.485281pt) -- cycle;
\draw (103.500000, 45.000000) node {$\mathrm{RX}_{\theta}$};
\end{scope}
\filldraw (103.500000, 30.000000) circle(1.500000pt);
\filldraw (103.500000, 15.000000) circle(1.500000pt);
\filldraw (103.500000, 0.000000) circle(1.500000pt);
\begin{scope}
\draw[fill=white] (126.000000, 15.000000) +(-45.000000:8.485281pt and 8.485281pt) -- +(45.000000:8.485281pt and 8.485281pt) -- +(135.000000:8.485281pt and 8.485281pt) -- +(225.000000:8.485281pt and 8.485281pt) -- cycle;
\clip (126.000000, 15.000000) +(-45.000000:8.485281pt and 8.485281pt) -- +(45.000000:8.485281pt and 8.485281pt) -- +(135.000000:8.485281pt and 8.485281pt) -- +(225.000000:8.485281pt and 8.485281pt) -- cycle;
\draw (126.000000, 15.000000) node {$X$};
\end{scope}
\draw (139.000000,45.000000) -- (139.000000,15.000000);
\filldraw (139.000000, 45.000000) circle(1.500000pt);
\begin{scope}
\draw[fill=white] (139.000000, 15.000000) circle(3.000000pt);
\clip (139.000000, 15.000000) circle(3.000000pt);
\draw (136.000000, 15.000000) -- (142.000000, 15.000000);
\draw (139.000000, 12.000000) -- (139.000000, 18.000000);
\end{scope}
\draw (149.000000,45.000000) -- (149.000000,30.000000);
\filldraw (149.000000, 45.000000) circle(1.500000pt);
\begin{scope}
\draw[fill=white] (149.000000, 30.000000) circle(3.000000pt);
\clip (149.000000, 30.000000) circle(3.000000pt);
\draw (146.000000, 30.000000) -- (152.000000, 30.000000);
\draw (149.000000, 27.000000) -- (149.000000, 33.000000);
\end{scope}
\draw (149.000000,15.000000) -- (149.000000,0.000000);
\filldraw (149.000000, 15.000000) circle(1.500000pt);
\begin{scope}
\draw[fill=white] (149.000000, 0.000000) circle(3.000000pt);
\clip (149.000000, 0.000000) circle(3.000000pt);
\draw (146.000000, 0.000000) -- (152.000000, 0.000000);
\draw (149.000000, -3.000000) -- (149.000000, 3.000000);
\end{scope}
\end{tikzpicture}\end{array} \hfill$

  \vspace{1.5em}

 $\text{b)} \hfill \begin{array}{c}\input{figures/circ6.tikz}\end{array} \hfill$

  \vspace{1em}

$\text{c)} \hfill \begin{array}{c}\begin{tikzpicture}[scale=1.000000,x=1pt,y=1pt]
\filldraw[color=white] (0.000000, -7.500000) rectangle (103.000000, 37.500000);
\draw[color=black] (0.000000,30.000000) -- (103.000000,30.000000);
\draw[color=black] (0.000000,15.000000) -- (103.000000,15.000000);
\draw[color=black] (0.000000,0.000000) -- (103.000000,0.000000);
\draw (8.000000,30.000000) -- (8.000000,0.000000);
\begin{scope}
\draw[fill=white] (8.000000, 30.000000) +(-45.000000:8.485281pt and 8.485281pt) -- +(45.000000:8.485281pt and 8.485281pt) -- +(135.000000:8.485281pt and 8.485281pt) -- +(225.000000:8.485281pt and 8.485281pt) -- cycle;
\clip (8.000000, 30.000000) +(-45.000000:8.485281pt and 8.485281pt) -- +(45.000000:8.485281pt and 8.485281pt) -- +(135.000000:8.485281pt and 8.485281pt) -- +(225.000000:8.485281pt and 8.485281pt) -- cycle;
\draw (8.000000, 30.000000) node {$U$};
\end{scope}
\filldraw (8.000000, 15.000000) circle(1.500000pt);
\filldraw (8.000000, 0.000000) circle(1.500000pt);
\draw[fill=white,color=white] (18.000000, -6.000000) rectangle (33.000000, 36.000000);
\draw (25.500000, 15.000000) node {$=$};
\draw (43.000000,30.000000) -- (43.000000,15.000000);
\begin{scope}
\draw[fill=white] (43.000000, 30.000000) +(-45.000000:8.485281pt and 8.485281pt) -- +(45.000000:8.485281pt and 8.485281pt) -- +(135.000000:8.485281pt and 8.485281pt) -- +(225.000000:8.485281pt and 8.485281pt) -- cycle;
\clip (43.000000, 30.000000) +(-45.000000:8.485281pt and 8.485281pt) -- +(45.000000:8.485281pt and 8.485281pt) -- +(135.000000:8.485281pt and 8.485281pt) -- +(225.000000:8.485281pt and 8.485281pt) -- cycle;
\draw (43.000000, 30.000000) node {$V$};
\end{scope}
\filldraw (43.000000, 15.000000) circle(1.500000pt);
\draw (56.000000,15.000000) -- (56.000000,0.000000);
\filldraw (56.000000, 0.000000) circle(1.500000pt);
\begin{scope}
\draw[fill=white] (56.000000, 15.000000) circle(3.000000pt);
\clip (56.000000, 15.000000) circle(3.000000pt);
\draw (53.000000, 15.000000) -- (59.000000, 15.000000);
\draw (56.000000, 12.000000) -- (56.000000, 18.000000);
\end{scope}
\draw (69.000000,30.000000) -- (69.000000,15.000000);
\begin{scope}
\draw[fill=white] (69.000000, 30.000000) +(-45.000000:8.485281pt and 8.485281pt) -- +(45.000000:8.485281pt and 8.485281pt) -- +(135.000000:8.485281pt and 8.485281pt) -- +(225.000000:8.485281pt and 8.485281pt) -- cycle;
\clip (69.000000, 30.000000) +(-45.000000:8.485281pt and 8.485281pt) -- +(45.000000:8.485281pt and 8.485281pt) -- +(135.000000:8.485281pt and 8.485281pt) -- +(225.000000:8.485281pt and 8.485281pt) -- cycle;
\draw (69.000000, 30.000000) node {$V^\dagger$};
\end{scope}
\filldraw (69.000000, 15.000000) circle(1.500000pt);
\draw (82.000000,15.000000) -- (82.000000,0.000000);
\filldraw (82.000000, 0.000000) circle(1.500000pt);
\begin{scope}
\draw[fill=white] (82.000000, 15.000000) circle(3.000000pt);
\clip (82.000000, 15.000000) circle(3.000000pt);
\draw (79.000000, 15.000000) -- (85.000000, 15.000000);
\draw (82.000000, 12.000000) -- (82.000000, 18.000000);
\end{scope}
\draw (95.000000,30.000000) -- (95.000000,0.000000);
\begin{scope}
\draw[fill=white] (95.000000, 30.000000) +(-45.000000:8.485281pt and 8.485281pt) -- +(45.000000:8.485281pt and 8.485281pt) -- +(135.000000:8.485281pt and 8.485281pt) -- +(225.000000:8.485281pt and 8.485281pt) -- cycle;
\clip (95.000000, 30.000000) +(-45.000000:8.485281pt and 8.485281pt) -- +(45.000000:8.485281pt and 8.485281pt) -- +(135.000000:8.485281pt and 8.485281pt) -- +(225.000000:8.485281pt and 8.485281pt) -- cycle;
\draw (95.000000, 30.000000) node {$V$};
\end{scope}
\filldraw (95.000000, 0.000000) circle(1.500000pt);
\end{tikzpicture}\end{array} \hfill$

  \vspace{.5em}

$\text{d)} \hfill \begin{array}{c}\input{figures/circ8.tikz}\end{array} \hfill$

  \vspace{1em}
  \caption{\label{fig:circuitids} Quantum circuit identities for decomposing the time evolution operator $\ee^{-\ii \frac{t}{r} H_\square}$ of a single plaquette in the Rokhsar-Kivelson using the multiple-controlled unitaries technique.}

  \end{figure}
  \begin{figure}
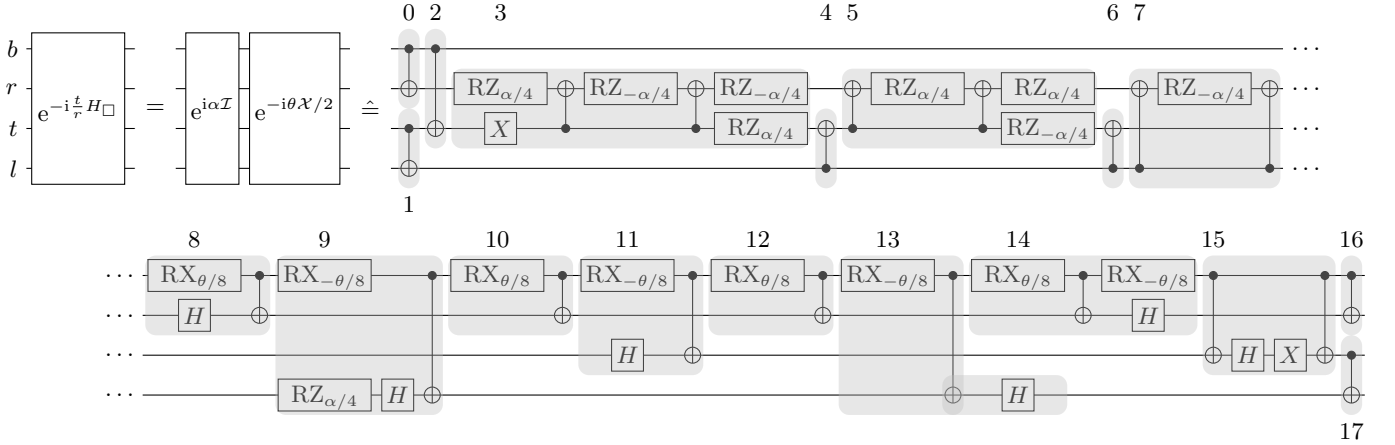

  \input{figures/circ10.tikz}

  \input{figures/circ11.tikz}
  \caption{\label{fig:circmcu}
    Exact circuit for the time evolution operator $\ee^{-\ii \frac{t}{r}H_\square}$ of a single plaquette operator in the Rokhskar-Kivelson model using the multiple-controlled unitaries technique. Here, $\alpha=-J\la \frac{t}{r}$ and $\theta=-2 J \frac{t}{r}$. The circuit comprises 18 two-qubit gates and has a depth of 16.
  }
\end{figure}

\subsubsection{Routing}
The entire, tileable circuit for the quantum simulation of $H$ leads to a circuit on the checkerboard lattice. Its basis circuit is depicted in \fig{checkerboard}.

Given the basis circuit in \fig{checkerboard}, creating a patch of $n$ by $m$ basis circuits will lead to incomplete circuits for plaquette operators at the boundaries. These can be removed in a postprocessing step. (Remove any gate in the patch that contains a qubit for which the $x$-coordinate is not within $0\leq x< n$ or for which the $y$-coordinate is not within $0\leq y< m$.)

\subsection{2D Quantum electrodynamics}
\label{sec:QED}
Here, we consider the Kogut-Susskind formulation of Wilson's gauge theories~\cite{kogut1975hamiltonian} in the case of two-dimensional quantum electrodynamics (2D QED)~\cite{haase2021resource,paulison2021simulating}. Fermions, created (annihilated) by $a_i^\dagger$ ($a_i^\nodagger$), are placed on the vertices of a 2D square grid with alternating chemical potential. A local Hilbert space, the `link space', is defined on the edges between sites $i$ and $j$. The electric field between sites $i$ and $j$ is described by the Hermitian electric-field operator $E_{(i,j)}$ (not to be confused with the edge operators $E_{i\si,i'\si'}$ appearing in the hybrid encoding for the quantum simulation of the Fermi-Hubbard model) that acts on the link space. The eigenstates $\ket{e_{ij}}$ of $E_{(i,j)}$ with eigenvalue $e_{ij}\in \mathbb Z$ span the link space. As in \sec{RK}, we will sometimes label the bottom, right, top, and left edges of a plaquette by $b,r,t,l$, respectively. Plaquette operators $P_\square$ are defined that act on the link spaces around a plaquette $\square$,
\begin{equation}\label{eq:qutritplaquette}
  P_\square=E_{b}^- E_{r}^- E_{t}^+ E_{l}^+.
\end{equation}
The operator $E_{(i,j)}^{+(-)}$ is the raising (lowering) operator associated with  $E_{(i,j)}$; $[E_{(i,j)},E^{\pm}_{(i,j)}]=\pm E^{\pm}_{(i,j)}$ and thus $E_{(i,j)}^{\pm}\ket{e_{(i,j)}}=\ket{e_{(i,j)}\pm 1}$.

With these definitions, the Kogut-Susskind model of 2D QED is a sum of an electric, magnetic, mass, and kinetic term~\cite{kogut1975hamiltonian,haase2021resource,paulison2021simulating},
\begin{equation}
H=H_E+H_B+H_{m}+H_\mrm{kin}.
\end{equation}
Here, the electric term is given by
\begin{equation}\label{eq:HE}
  H_E=\frac{g^2}{2} \sum_{(i,j)} E_{(i,j)}^2,
\end{equation}
with $g$ the bare coupling constant. The magnetic term reads
\begin{equation}\label{eq:magnetic} H_B=-\frac{1}{2g^2d^2}\sum_{\square}(P_\square^\nodagger+P_\square^\dagger),
\end{equation}
with $d$ the lattice spacing. The mass term is
\begin{equation}\label{eq:mass}
  H_m=m\sum_i(-1)^{i}a_i^\dagger a_{i}^\nodagger,
\end{equation}
with $m$ the effective electron mass. The sum is over all vertices $i$ of the 2D square grid. The index $i$ is even or odd according to a bipartition of the 2D square grid. Finally, we have the kinetic term,
\begin{equation}\label{eq:kinetic}
  H_{\mrm{kin}}=\Omega\sum_{\substack{i,\\ (i, j)\in \{\uparrow,\rightarrow\}}} a_i^\dagger E^{-}_{(i,j)}a_j^\nodagger + h.c.,
\end{equation}
with $\Om$ the kinetic strength and where $h.c.$ stands  for the Hermitian conjugate. Here, $(i, j)\in \{\uparrow,\rightarrow\}$ indicates a sum over all edges containing $i$ in the indicated directions. In the literature, the kinetic term appears with various phase differences between different types of edges~\cite{meth2023simulating}. Here, we have taken the form of Ref.~\cite{haase2021resource} for concreteness, where every hopping term has the same phase. For the routing problem of the Trotterized time evolution of $H$, these differences are irrelevant since they do not influence the spatial structure of the circuit that is to be routed.

To map $H$ to a tileable Hamiltonian defined on qudits, we use the compact encoding~\cite{derby2021compact} for the fermions (which maps the fermionic operators to qu\emph{bit} operators). As in ~\cite{haase2021resource,paulison2021simulating}, to map $E_{(i,j)}$ and $E_{(i,j)}^{\pm}$ to finite-dimensional operators on qudits, we truncate the link space. In the following, we consider the minimal nontrivial example where the link space is three dimensional~\cite{paulison2021simulating}, that is, that of a qutrit. Defining the Hermitian operator
\begin{equation}
  S^z:=\left(
    \begin{array}{ccc}
      1 & 0 & 0 \\
      0 & 0 & 0 \\
      0 & 0 & -1
    \end{array}
  \right),
\end{equation}
truncation to the qutrit space entails $E_{(i,j)}\mapsto S^z_{(i,j)}$. This makes mapping the electric part of the Hamiltonian straightforward;
\begin{equation}\label{eq:electric}
  H_E\mapsto \frac{g^2}{2}\sum_{(i,j)}[S^z_{(i,j)}]^2.
\end{equation}

\subsubsection{Magnetic part}
For the plaquette operators in the magnetic term $H_B$, we use a procedure similar to the procedure used to map the plaquette operators to qubit operators in the Rokhsar-Kivelson model [cf. Eqs. \eqref{eq:usquare} and \eqref{eq:uud}].

Truncating $E^{\pm}$ to the qutrit space entails
\begin{align}
 E^+\mapsto S^{+}:=\left(
    \begin{array}{ccc}
      0 & 1 & 0 \\
      0 & 0 & 1 \\
      0 & 0 & 0
    \end{array}
 \right), &&
E^-\mapsto S^{-}:=\left(
    \begin{array}{ccc}
      0 & 0 & 0 \\
      1 & 0 & 0 \\
      0 & 1 & 0
    \end{array}
  \right).
\end{align}
We now go beyond Ref.~\cite{paulison2021simulating} by defining the Hermitian operators
\begin{align}
  S^x:=\left(
    \begin{array}{ccc}
      0 & 1 & 0 \\
      1 & 0 & 1 \\
      0 & 1 & 0
    \end{array}
  \right),
  &&
  S^y:=\left(
    \begin{array}{ccc}
      0 & -\ii & 0 \\
      \ii & 0 & -\ii \\
      0 & \ii & 0
    \end{array}
  \right)
\end{align}
So that
\begin{equation}\label{eq:LinkToQutrit}
E^\pm \mapsto S^{\pm} = (S^x\pm \ii S^y)/2.
\end{equation}
With \eq{qutritplaquette} this gives [cf. \eq{uud}]
\begin{align}
  P_\square^\nodagger+P_\square^\dagger&\mapsto \frac{1}{16}(S^{x}_{b}-\ii S^{y}_b)(S^{x}_r-\ii S^{y}_r)(S^{x}_t+\ii S^{y}_t)(S^{x}_l+\ii S^{y}_l)+h.c.\\
  &=\frac{1}{8}(S_b^xS_r^xS_t^xS_l^x-S_b^yS_r^yS_t^xS_l^x-S_b^xS_r^xS_t^yS_l^y+S_b^yS_r^yS_t^yS_l^y\nn\\
  &\phantom{=\frac{1}{8}(}+S_b^xS_r^yS_t^xS_l^y+S_b^yS_r^xS_t^xS_l^y+S_b^xS_r^yS_t^yS_l^x+S_b^yS_r^xS_t^yS_l^x).\label{eq:QutritPauliPlaqs}
\end{align}
With the correct prefactors, this gives the mapping for the magnetic term $H_B$.

We will now show a quantum circuit for $U^{(zzzz)}(\al):=\ee^{-\ii \al S_b^zS_r^zS_t^zS_l^z/2}$ using only two-qutrit gates and no auxiliary qutrits. The circuits for the exponentiation of the terms in \eq{QutritPauliPlaqs} (and thereby for time evolution along the magnetic part $H_B$) are found afterward by inserting single-qutrit gates for basis transformations.  Denote the eigenvectors of $S^z$ by $\ket{1}, \ket{0}$ and $\ket{-1}$ according to their eigenvalue. We call this the (qutrit) computational basis. We will consider the action of $U^{(zzzz)}(\al)$ on computational basis states only; action on other states follows by linearity. Note that $U^{(zzzz)}(\al)$ acts like $\ee^{-\ii \al Z_bZ_rZ_tZ_l/2}$ on states in $\{\ket 1, \ket{-1}\}^{\otimes{4}}$ and as the identity on the complement (i.e. any computational basis state containing a tensor factor $\ket{0}$). First, assume $U^{(zzzz)}(\al)$ acts on a state in $\{\ket 1, \ket{-1}\}^{\otimes{4}}$. Define the Hermitian and unitary operator
\begin{equation}\label{eq:qutritX}
  X':=\left(
    \begin{array}{ccc}
      0 & 0 & 1 \\
      0 & 1 & 0 \\
      1 & 0 & 0
    \end{array}
  \right).
\end{equation}
It acts like a Pauli-$X$ gate on the subspace $\mrm{span}\{\ket 1, \ket{-1}\}$ (and like the identity on $\ket 0$). Define $\mrm{C}^{(-1)}_i\mrm{NOT}_{j}$ as the qutrit gate that performs an $X'$ gate on qutrit $j$ conditioned on qutrit $i$ being in the state $\ket{-1}$. Then, the circuit $C^\dagger\mrm{RS}^z_\al C$, with $C$ a cascade of $\mrm{C}^{(-1)}_i\mrm{NOT}_{j}$ gates, $C=\mrm{C}^{(-1)}_1\mrm{NOT}_{0}\times\mrm{C}^{(-1)}_2\mrm{NOT}_{1}\times \mrm{C}^{(-1)}_3\mrm{NOT}_{2}$, and $\mrm{RS}^z_\al=\ee^{-\ii \al S^z/2}$, performs $U^{(zzzz)}(\al)$ in the subspace $\mrm{span}\{\ket 1, \ket{-1}\}^{\otimes{4}}$.

If $U^{(zzzz)}(\al)$ acts on a product state not in $\{\ket 1, \ket{-1}\}^{\otimes 4}$ it acts as the identity. Thus, the full $U^{(zzzz)}(\al)$ can be implemented with $U^{(zzzz)}(\al)=C^\dagger D C$, where $D$ performs $\mrm{RS}^z_\al$ conditioned on qutrits 0, 1 and 2 all being in $\{\ket{1},\ket{-1}\}$. For $D$, we can use the technique of Ref.~\cite{yordanov2020efficient} (also see \fig{circuitids}b), replacing the $\CNOT$s by $\mrm{C}^{(1,-1)}_i\mrm{NOT}_{j}$, where the latter gate performs a $X'$ gate on qutrit $j$ conditioned on qutrit $i$ being in the state $\ket{1}$ or $\ket{-1}$. The full circuit is given in \fig{qutritplaquette}.

The circuit for $U^{(zzzz)}(\al)$ gives the circuit for the exponentiation of the terms in \eq{QutritPauliPlaqs} upon an appropriate local change of basis. Define the unitaries
\begin{align}
  A^x:=\frac{1}{2}\left(
  \begin{array}{ccc}
      1 & -\sqrt{2} & 1 \\
      \sqrt{2} & 0  & -\sqrt{2} \\
      1 & \sqrt{2} & 1
    \end{array}
  \right),
  &&
  A^{y}:=\frac{1}{2}\left(
  \begin{array}{ccr}
      -1 & \sqrt{2} & -1 \\
       -\ii \sqrt{2} & 0  & \ii\sqrt{2} \\
      1 & \sqrt{2} & 1
    \end{array}
  \right).
\end{align}
Then
\begin{align}\label{eq:QutritBasisTransfo}
  A^{x} S^z(A^{x})^\dagger=S^x/\sqrt{2}, &&  A^{y} S^z(A^{y})^\dagger =S^y/\sqrt{2},
\end{align}
which can be used straightforwardly to generalize $U^{(zzzz)}$ to $U^{(ijk\ell)}(\al):=\ee^{- \ii \al \nu^{(ijk\ell)}  S_b^iS_r^jS_t^kS_l^\ell/2}$. We remind that $b,r,t,l$ denote the link of the plaquette the operators act on (bottom, right, top, left), and $i,j,k,\ell\in\{x,y,z\}$ determine the types of the operators, as in \eq{QutritPauliPlaqs}. For the general form $U^{(ijk\ell)}(\al)$, the factor $\nu^{(ijk\ell)}$ is determined by the factors  $1/\sqrt{2}$ in [\eq{QutritBasisTransfo}]. That is, $\nu^{(ijk\ell)}=2^{m/2}$, with $m$ the number of indices $i,j,k,\ell$ unequal to $z$. The terms in \eq{QutritPauliPlaqs} contain no $S^z$ factors. For these terms, we note that $U^{(ijk\ell)}(\al):=\ee^{- \ii \al  S_b^iS_r^jS_t^kS_l^\ell/8}$ with $i,j,k,\ell\in\{x,y\}$.

\begin{figure}
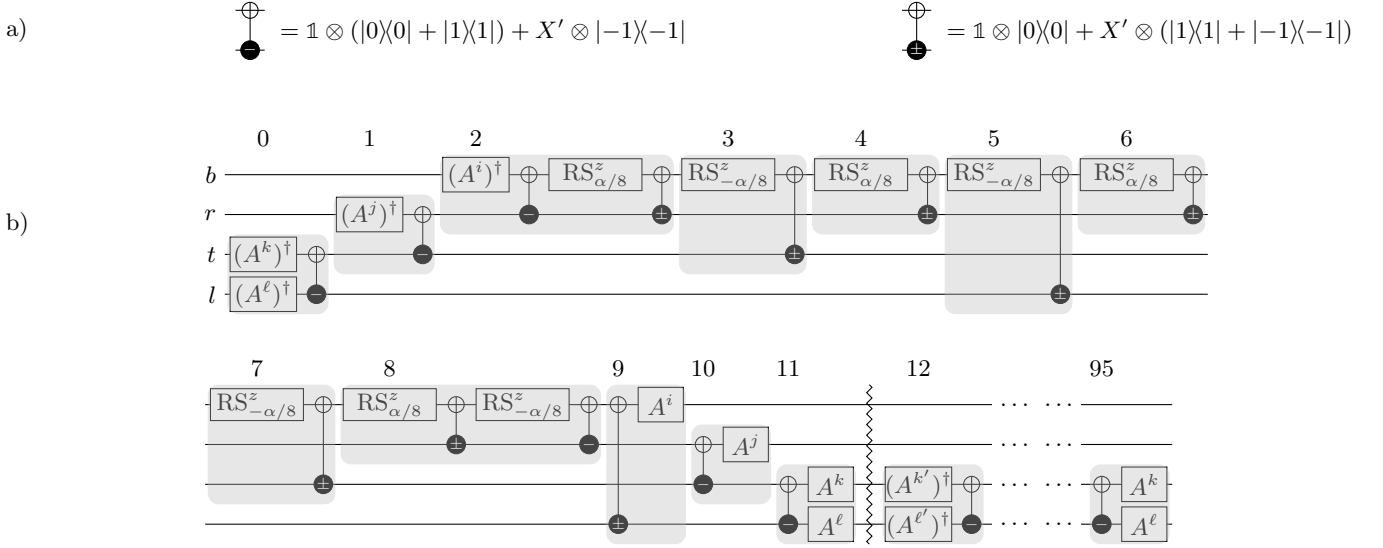

$\text{a)} \hfill
  \begin{array}{c} \raisebox{-1.32em}{\begin{tikzpicture}[scale=1.000000,x=1pt,y=1pt]
\filldraw[color=white] (0.000000, -7.500000) rectangle (11.000000, 22.500000);
\draw[color=black] (0.000000,15.000000) -- (11.000000,15.000000);
\draw[color=black] (0.000000,0.000000) -- (11.000000,0.000000);
\draw (5.500000,15.000000) -- (5.500000,0.000000);
\begin{scope}
\draw[fill=black] (5.500000, 0.000000) circle(3.500000pt);
\clip (5.500000, 0.000000) circle(3.500000pt);
\draw (5.500000, 0.000000) node {{\tiny {\color{white} $-$}}};
\end{scope}
\begin{scope}
\draw[fill=white] (5.500000, 15.000000) circle(3.000000pt);
\clip (5.500000, 15.000000) circle(3.000000pt);
\draw (2.500000, 15.000000) -- (8.500000, 15.000000);
\draw (5.500000, 12.000000) -- (5.500000, 18.000000);
\end{scope}
\end{tikzpicture}}= \id \otimes (\ketbra{0}{0}+\ketbra{1}{1})+  X' \otimes \ketbra{-1}{-1}
  \end{array}
  \hfill
  \begin{array}{c}
  \raisebox{-1.32em}{\begin{tikzpicture}[scale=1.000000,x=1pt,y=1pt]
\filldraw[color=white] (0.000000, -7.500000) rectangle (11.000000, 22.500000);
\draw[color=black] (0.000000,15.000000) -- (11.000000,15.000000);
\draw[color=black] (0.000000,0.000000) -- (11.000000,0.000000);
\draw (5.500000,15.000000) -- (5.500000,0.000000);
\begin{scope}
\draw[fill=black] (5.500000, 0.000000) circle(3.500000pt);
\clip (5.500000, 0.000000) circle(3.500000pt);
\draw (5.500000, 0.000000) node {{\tiny {\color{white} $\pm$}}};
\end{scope}
\begin{scope}
\draw[fill=white] (5.500000, 15.000000) circle(3.000000pt);
\clip (5.500000, 15.000000) circle(3.000000pt);
\draw (2.500000, 15.000000) -- (8.500000, 15.000000);
\draw (5.500000, 12.000000) -- (5.500000, 18.000000);
\end{scope}
\end{tikzpicture}}= \id \otimes\ketbra{0}{0} +  X' \otimes (\ketbra{1}{1} +\ketbra{-1}{-1} )
  \end{array}
$

  \vspace{2em}

  $\text b) \hfill \begin{array}{c}\input{figures/qutritplaquette.tikz}\end{array}\hfill$

  \vspace{1em}

  \input{figures/qutritplaquette2.tikz}

  \caption{\label{fig:qutritplaquette} a) Definition of the CNOTs, $\mrm{C}^{(-1)}\mrm{NOT}$ and $\mrm{C}^{(1,-1)}\mrm{NOT}$, both acting on two qutrits, used in the circuit for the quantum simulation of the Kogut-Susskind model of 2D QED. The first tensor factor corresponds to the upper qutrit wire.  b) Qutrit circuit for approximate time evolution of one plaquette in the magnetic part $H_B$ of the Kogut-Susskind lattice model of 2D QED. The circuit before the zigzag line implements $U^{(ijk\ell)}(\al):=\ee^{- \ii \al \nu^{(ijk\ell)}  S_b^iS_r^jS_t^kS_l^\ell/2}$. The full circuit for approximate time evolution along the magnetic Hamiltonian $H_B$ consists of eight repetitions of the circuit before the zigzag line, each time with different basis transformation matrices, as determined by \eq{magnetic}. The angle $\al$ is determined by Eqs. \eqref{eq:lie-trotter} or \eqref{eq:suzuki}, \eqref{eq:magnetic}, and  \eqref{eq:QutritPauliPlaqs}. The integers and gray areas give an initial definition of the gates $G^{H_B}_i$ used in the implementation~\cite{kattemolle2024quantile}. Implementing the exponents according to the ordering in \eq{QutritPauliPlaqs} will lead to  $G^{H_B}_{12}G^{H_B}_{11}=G^{H_B}_{36}G^{H_B}_{35}=G^{H_B}_{60}G^{H_B}_{59}G^{H_B}_{84}G^{H_B}_{83}=\id$ and hence these can be omitted. Subsequently, the two-qutrit gates around the removed gates can be merged into a single two-qutrit unitary, absorbing $G^{H_B}_{13}$ into $G^{H_B}_{10}$, $G^{H_B}_{37}$ into $G^{H_B}_{34}$, $G^{H_B}_{61}$ into $G^{H_B}_{58}$, and $G^{H_B}_{85}$ into $G^{H_B}_{82}$. Finally, also the gates around the zigzag lines can be merged where there is no cancellation, absorbing $G^{H_B}_{24}$ into $G^{H_B}_{23}$, $G^{H_B}_{48}$ into $G^{H_B}_{47}$ and $G^{H_B}_{72}$ into $G^{H_B}_{71}$. We do not relabel gates when removing or merging gates. The final circuit for the time evolution along one plaquette operator has 81 general two-qutrit gates. The circuit for the second, lower plaquette operator in the basis circuit (see \fig{QEDCircuitLayoutSM}) contains gates $G^{H_B}_{96}$ to $G^{H_B}_{191}$, from which gates are removed and merged as before.}
\end{figure}

\subsubsection{Mass part}
The mass term [\eq{mass}] and the kinetic term [\eq{kinetic}] involve fermionic degrees of freedom, which we encode using the same encoding as for the Fermi-Hubbard Hamiltonian, but now with one spin orbital per orbital site, thus recovering the compact encoding of Ref.~\cite{derby2021compact}. The compact encoding defines Majorana operators $\gamma_i=a_i^\nodagger+a^\dagger_i$ and $\bar \gamma_i=(a_i^\nodagger-a^\dagger_i)/\ii$, where $i$ labels the sites of a 2D square grid. These operators define edge operators $E_{i,i'}=-\ii \gamma_i\gamma_{i'}$ and vertex operators $V_i=-\ii \gamma_i\bar \gamma_i$ which anticommute if they share an index and commute otherwise.

The fermionic number operator is related to these operators by $a_i^\dagger a_i=(\id-V_i)/2$. In turn, in the compact encoding, $V_i$ is mapped to $Z_i$. Thus, the mass term maps to
\begin{equation}
 H_m\mapsto -\frac{m}{2}\sum_i(-1)^{i+1}Z_i,
\end{equation}
where we have left out terms proportional to the identity.

\subsubsection{Kinetic part}
The kinetic term [\eq{kinetic}] involves both fermionic and link operators. Mapping the link operators to qutrit operators [\eq{LinkToQutrit}], and the fermionic operators to edge and vertex operators, we have
\begin{equation}\label{eq:QEDhopping}
  a_i^\dagger E^-_{(i,j)}a_j^\nodagger + h.c. = -\frac{1}{4}[\ii S^x_{(i,j)}(V_i-V_j)+S^y_{(i,j)}(\id-V_iV_j)]E_{ij}.
\end{equation}
The compact encoding further maps the edge operators $E_{ij}$ to three-qubit Pauli words acting on qubits $i$, $j$, and an auxiliary qubit which we denote by $\langle i,j \rangle$ associated with $i$ and $j$. The explicit three-qubit Pauli words can be found in Ref.~\cite{derby2021compact}. After defining alternating even and odd plaquettes, all the qubits around the odd plaquettes share a single auxiliary qubit (I.e., $\langle i,j\rangle=\langle j,k\rangle=\langle k,l\rangle=\langle l,i\rangle$ for the qubits $i,j,k,l$ around an odd plaquette) and there is no auxiliary qubit associated with the even plaquettes. The vertex operators $V_i$ are mapped to the single-qubit Pauli operators $Z_i$. After this final mapping of edge and vertex operators to Pauli words, \eq{QEDhopping} will be a sum of four terms of the form $S^{z}_{(i,j)}Z_iZ_{\langle i,j\rangle}Z_j$ up to local basis transformations and scalar prefactors. For $S^z_{(i,j)}$, these basis transformations are given in \eq{QutritBasisTransfo}. For the qubits, the transformations are given by
\begin{align}
  A^xZ(A^x)^\dagger=X && (A^x=H)\\
  A^yZ(A^y)^\dagger=Y && (A^y=SH),
\end{align}
with $H$ the Hadamard gate and $S=\diag(1,\ii)$ the phase gate. It is clear from the context when $A^{x}$ and $A^y$ refer to the qubit basis transformations instead of the qutrit basis transformations [\eq{QutritBasisTransfo}]. The above gives the circuit for $a_i^\dagger E^-_{(i,j)}a_j + h.c.$, and thereby the kinetic term of the Hamiltonian, $H_\mrm{kin}$, given in \fig{QEDkinetic}.

\begin{figure}
\input{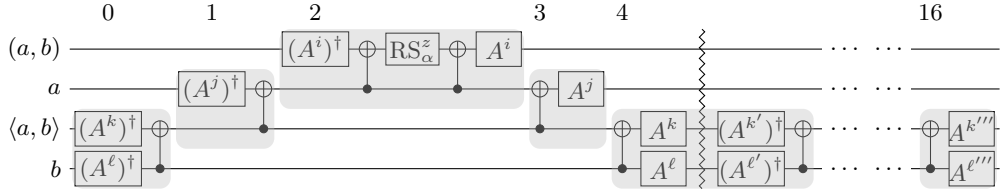}
\caption{\label{fig:QEDkinetic} Quantum circuit for time evolution along a single edge $(a,b)$ in the kinetic term $H_\mrm{kin}$ in the Kogut-Susskind lattice model of 2D QED. Here, the qubit line $(a,b)$ represents the link qubit between $a$ and $b$, and $\langle a,b\rangle$ the auxiliary qubit for the compact encoding in the odd plaquette to which the edge $(a,b)$ is adjacent. The circuit before the zigzag line implements $\exp(-\ii \al \nu^{(ijk\ell)} S^i_{(a,b)}\si^j_a\si^k_{\langle a,b\rangle}\si^l_b/2)$, where $\si^x_a=X_a$, etc., and where $\nu^{(ijk\ell)}$ arises from [\eq{QutritBasisTransfo}] as in \fig{qutritplaquette}. The total circuit consists of four repetitions of the circuit before the zigzag line, each time with different basis transformation matrices $A^\nu$, as determined by \eq{kinetic} and the local encoding. The angle $\alpha$ is determined by Trotterization [\eq{lie-trotter} or \eq{suzuki}], \eq{kinetic}, and \eq{QEDhopping}, and the compact encoding. We define a regular CNOT symbol with its target on a qutrit as $\ketbra{0}{0}\otimes \id+\ketbra{1}{1}\otimes X'$, where the first tensor factor is the control qubit and the second the target qutrit, and with $X'$ as in \eq{qutritX}. In the implementation~\cite{kattemolle2024quantile}, the gates $G^{H_\mrm{kin}}_i$ are defined as follows. Gates $G^{H_\mrm{kin}}_0$ to $G^{H_\mrm{kin}}_{16}$ implement the above circuit for the kinetic term between fermionic sites $a=(0,0,0)$ and $b=(0,0,2)$ (see \fig{QEDCircuitLayoutSM}), gates $G^{H_\mrm{kin}}_{17}$ to $G^{H_\mrm{kin}}_{33}$ for  $a=(0,1,2)$ and $b=(1,0,0)$, gates $G^{H_\mrm{kin}}_{34}$ to $G^{H_\mrm{kin}}_{50}$ for  $a=(0,0,0)$ and $b=(0,1,2)$, and gates $G^{H_\mrm{kin}}_{51}$ to $G^{H_\mrm{kin}}_{67}$ for  $a=(0,0,2)$ and $b=(1,0,0)$.}
\end{figure}

\subsubsection{Total Hamiltonian}
The structure of the total tileable circuit for the simulation of the Kogut-Susskind model of 2D QED is given in \fig{QEDCircuitLayoutSM}. The exponentiation of the electric and mass parts of the Hamiltonian $H_E$ [\eq{electric}] leads to single-qutrit gates only, which can all be absorbed into the subsequent two-qudit gates. The total basis circuit has 230 gates, much more than the other circuits we consider in this work. Out of these gates, 162 are on account of the magnetic term. Although we provide the full circuit, in running the circuit on hardware, as a proof of principle, one may first simulate the model in the nonmagnetic limit $H_B\rightarrow 0$.

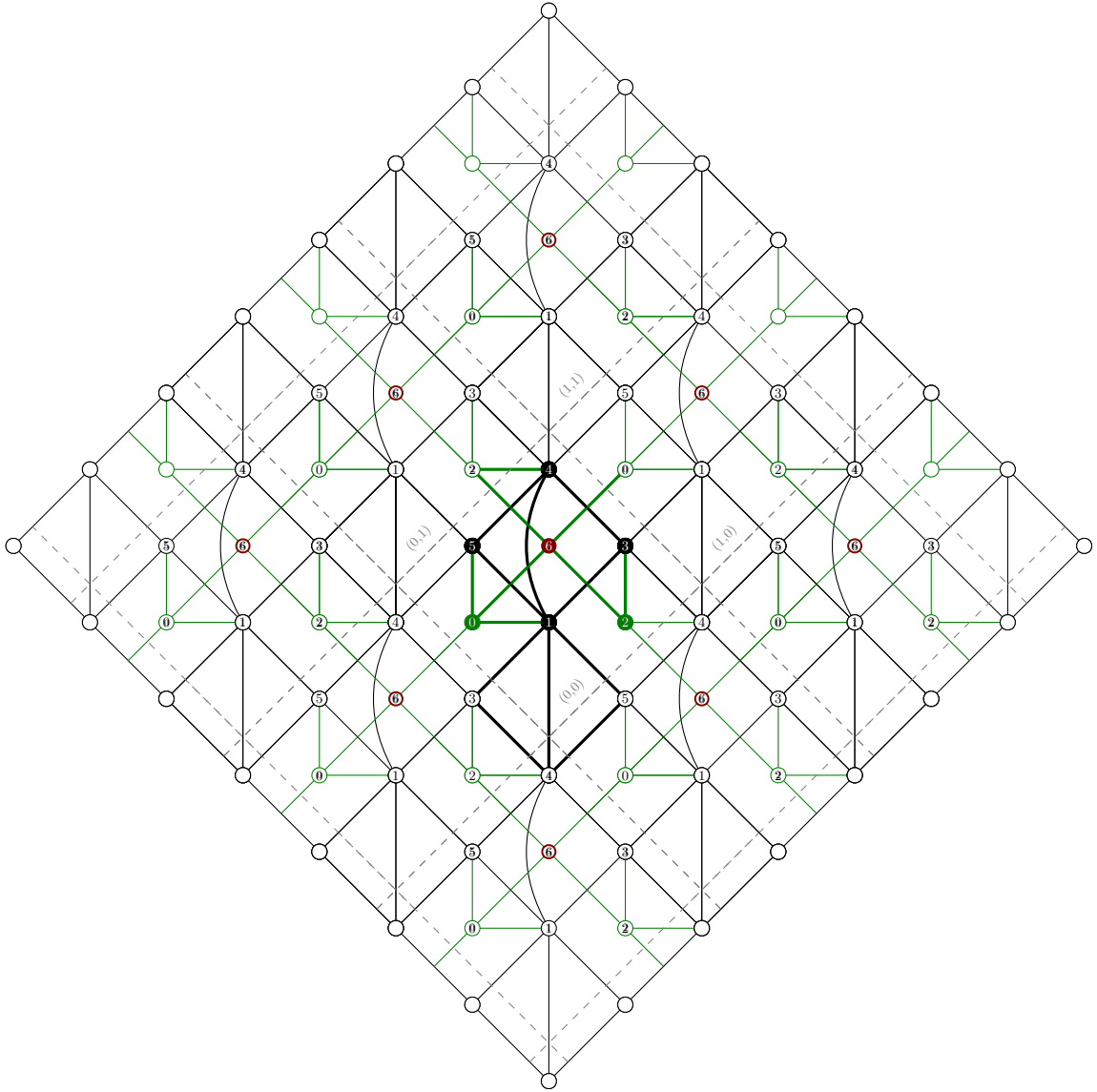
\begin{figure}
  \centering

\begin{tikzpicture}[scale=1.5, rotate=45]
    \definecolor{lightgreen}{rgb}{0.0, 0.5, 0.0}
    \definecolor{lightred}{rgb}{0.5, 0.0, 0.0}
  \newcommand{\drawCircuit}{
    \foreach \i in {0,...,3}
        \draw (\i,0) -- (\i,3);
    \foreach \j in {0,...,3}
        \draw (0,\j) -- (3,\j);

    \draw (1,1) to[bend left=30] (2,2);
    \draw (2,2) -- (3,3);
    \draw (0,0) -- (1,1);
    \draw (2,0) -- (3,1);
    \draw (0,2) -- (1,3);

    \draw[lightgreen] (1.5,0) -- (1.5,3);
    \draw[lightgreen] (0,1.5) -- (3,1.5);
    \draw[lightgreen] (0.5,1.5) -- (1,2);
    \draw[lightgreen] (0.5,1.5) -- (1,1);
    \draw[lightgreen] (1.5,0.5) -- (2,1);
    \draw[lightgreen] (1.5,2.5) -- (2,2);
    \draw[lightgreen] (2.5,1.5) -- (3,2);
    \draw[lightgreen] (2.5,1.5) -- (3,1);
    \draw[lightgreen] (1.5,2.5) -- (2,3);
    \draw[lightgreen] (1.5,0.5) -- (2,0);

    \filldraw[fill=white, draw=lightgreen] (0.5,1.5) circle (2pt);
    \filldraw[fill=white, draw=lightgreen] (1.5,0.5) circle (2pt);
    \filldraw[fill=white, draw=lightgreen] (1.5,2.5) circle (2pt);
    \filldraw[fill=white, draw=lightgreen] (2.5,1.5) circle (2pt);
    \fill[lightred] (1.5,1.5) circle (2pt);
    \fill[white] (1.5,1.5) circle (1.5pt);

    \foreach \i in {0,...,3}
        \foreach \j in {0,...,3} {
            \filldraw[fill=white, draw=black] (\i,\j) circle (2pt);
        }

    \draw[dashed,gray] (0.25,0) -- (0.25,3);
    \draw[dashed,gray] (2.25,0) -- (2.25,3);
    \draw[dashed,gray] (0,0.25) -- (3,0.25);
    \draw[dashed,gray] (0,2.25) -- (3,2.25);

    \foreach \i/\j in {0/0, 1/0, 2/0, 3/0, 0/1, 0/2, 0/3, 1/3, 2/3, 3/1, 3/2, 3/3} {
        \filldraw[draw=black, fill=white] (\i,\j) circle (2pt);
      }

    \node[anchor=center, scale=0.4] at (0.5,1.5) {0};

    \node[anchor=center, scale=0.4] at (1,1) {1};
    \node[anchor=center, scale=0.4] at (1.5,0.5) {2};
    \node[anchor=center, scale=0.4] at (2,1) {3};
    \node[anchor=center, scale=0.4] at (2,2) {4};
    \node[anchor=center, scale=0.4] at (1,2) {5};
    \node[anchor=center, scale=0.4] at (1.5,1.5) {6};

    }

    \draw[very thick] (1,1) to[bend left=30] (2,2);
    \draw[style=very thick] (1,1) -- (1,2);
    \draw[style=very thick] (1,2) -- (2,2);
    \draw[style=very thick] (2,2) -- (2,1);
    \draw[style=very thick] (2,1) -- (1,1);
    \draw[style=very thick] (1,1) -- (0,1);
    \draw[style=very thick] (0,1) -- (0,0);
    \draw[style=very thick] (0,0) -- (1,0);
    \draw[style=very thick] (1,0) -- (1,1);
    \draw[style=very thick] (1,1) -- (0,0);

    \draw[lightgreen, style=very thick] (0.5,1.5) -- (1,2);
    \draw[lightgreen, style=very thick] (0.5,1.5) -- (1,1);
    \draw[lightgreen, style=very thick] (0.5,1.5) -- (1.5,1.5);
    \draw[lightgreen, style=very thick] (1.5,1.5) -- (1.5,2.5);
    \draw[lightgreen, style=very thick] (1.5,1.5) -- (1.5,0.5);
    \draw[lightgreen, style=very thick] (1.5,1.5) -- (2.5,1.5);
    \draw[lightgreen, style=very thick] (1.5,0.5) -- (2,1);
    \draw[lightgreen, style=very thick] (1.5,2.5) -- (2,2);

    \foreach \i in {0,-2,2}{
    \foreach \j in {0,-2,2} {
    \begin{scope}[shift={(\i,\j)}]
      \drawCircuit
    \end{scope}
        }}

    \fill[lightgreen] (0.5,1.5) circle (2pt);
    \fill (1,2) circle (2pt);
    \fill (1,1) circle (2pt);
    \fill (2,2) circle (2pt);
    \fill (2,1) circle (2pt);
    \fill[lightgreen] (1.5,0.5) circle (2pt);
    \fill[lightred] (1.5,1.5) circle (2pt);

  \newcommand{\drawBlackLables}{
    \node[anchor=center, scale=0.6] at (0.5,1.5) {0};

    \node[anchor=center, scale=0.6] at (1,1) {1};
    \node[anchor=center, scale=0.6] at (1.5,0.5) {2};
    \node[anchor=center, scale=0.6] at (2,1) {3};
    \node[anchor=center, scale=0.6] at (2,2) {4};
    \node[anchor=center, scale=0.6] at (1,2) {5};
    \node[anchor=center, scale=0.6] at (1.5,1.5) {6};

    }

    \foreach \i in {0,-2,2}{
    \foreach \j in {0,-2,2} {
    \begin{scope}[shift={(\i,\j)}]
      \drawBlackLables
    \end{scope}
        }}

\node[anchor=center, scale=0.6, color=white] at (0.5,1.5) {0};

\node[anchor=center, scale=0.6, color=white] at (1,1) {1};
\node[anchor=center, scale=0.6, color=white] at (1.5,0.5) {2};
\node[anchor=center, scale=0.6, color=white] at (2,1) {3};
\node[anchor=center, scale=0.6, color=white] at (2,2) {4};
\node[anchor=center, scale=0.6, color=white] at (1,2) {5};
\node[anchor=center, scale=0.6, color=white] at (1.5,1.5) {6};

\node[anchor=center, scale=0.6, color=gray, rotate=45] at (0.7,0.4) {(0,0)};
\node[anchor=center, scale=0.6, color=gray, rotate=45] at (0.7,2.4) {(0,1)};
\node[anchor=center, scale=0.6, color=gray, rotate=45] at (2.7,0.4) {(1,0)};
\node[anchor=center, scale=0.6, color=gray, rotate=45] at (2.7,2.4) {(1,1)};

\end{tikzpicture}
  \caption{\label{fig:QEDCircuitLayoutSM} Spatial structure of the tileable circuit for the simulation of the Kogut-Susskind lattice model of 2D QED. The basis circuit is depicted with bold lines. The green vertices (filled and open) depict qubits associated with the staggered fermionic orbitals on a 2D square lattice. The red vertices (filled and open) depict the auxiliary qubits used in the compact fermion-to-qubit encoding. They define the even plaquettes. The black vertices (filled and open) depict qutrits encoding the electric field degrees of freedom. An edge between two vertices indicates that a two-qudit gate acts between the corresponding qudits during the circuit. Black edges correspond to gates for the magnetic part of the Hamiltonian and green edges for the kinetic part of the Hamiltonian. The other parts of the Hamiltonian, i.e., the mass and electric part, only involve single-qudit gates on the green and black vertices, respectively. Dashed lines delineate the unit cells. Integers inside the vertices show the seed numbers and the tuples show the unit cell coordinates.}
\end{figure}

\section{Merging SWAP gates}\label{sec:merge_swaps}
In this work, we do not assume a specific type of quantum hardware. As in Ref.~\cite{clinton2024towards}, we therefore adopt a costing model in which each general two-qudit gate has unit cost. In particular, this implies that any SWAP gate immediately preceding (following) a two-qudit unitary $U_{ij}$ acting on the same two qudits can be merged into that unitary, yielding a new unitary
$U'_{ij} = U_{ij} \, \SWAP$  ($ U'_{ij} = \SWAP \, U_{ij}$) that has unit cost. We emphasize that in the formulation of the routing problem as an SMT formula, it is entirely possible to go beyond this costing model by incorporating additional details about the hardware, if necessary. For instance, Ref.~\cite{tan2020optimal} considers a scenario where a SWAP gate requires three units of time.

Nevertheless, in certain cases, our current costing model is already faithful. We discuss three such cases here. First, consider a qu\emph{bit}-based quantum computer with native CNOT gates and comparatively fast, high-fidelity single-qubit gates. It is well known that any two-qubit unitary can be decomposed using three CNOTs and four layers of interspersed single-qubit gates~\cite{vidal2004universal}. Assuming the two-qubit unitaries $U_{ij}$ in a given circuit are sufficiently general, they typically require three CNOTs for decomposition. Since a $\SWAP$ gate also decomposes into three CNOTs, both $U_{ij}$ and any $\SWAP_{i'j'}$ require approximately the same execution time. Furthermore, appending or prepending a $\SWAP_{ij}$ gate to $U_{ij}$ typically does not increase the depth of its decomposition, as $\SWAP_{ij}U_{ij}$ and $U_{ij}\SWAP_{ij}$ are just as general as $U_{ij}$. Consequently, SWAP gates can be merged with two-qubit unitaries, treating the entire operation as a single two-qubit unitary with unit cost.

An example of a set of two-qubit unitaries  $U_{ij}$ that require three CNOT gates are those introduced in \eq{general_gate}, with $H_i, H_j$ arbitrary and $H_{ij}=\la_{ij}^x X_iX_j + \la_{ij}^y Y_iY_j + \la_{ij}^z Z_{i}Z_j$, where $\pi/4\geq \frac{t}{r}\la^x_{ij}\geq\frac{t}{r}\la_{ij}^y\geq \lvert \frac{t}{r}\la_{ij}^z\rvert$ and $\frac{t}{r}\la_{ij}^z\neq 0$~\cite{vidal2004universal}. The two-qubit unitaries arising in the quantum simulation of the XXZ model [\eq{XXZ}], form examples of such unitaries.

As a second example in which our costing model is already faithful, we consider spin qubits, which naturally interact via the Heisenberg exchange interaction,
\begin{equation}
    H_{ij} = J \bS_i \cdot \bS_j = \frac{J}{4} (X_i X_j + Y_i Y_j + Z_i Z_j),
\end{equation}
where $J$ is the tunable coupling strength \cite{loss1998quantum, burkard2023semiconductor}. In such systems, the natural two-qubit gate is
\begin{equation}
    V_{ij}^{(\mrm{XXX})}(Jt) := \ee^{-\ii J t (X_i X_j + Y_i Y_j + Z_i Z_j)/4}.
  \end{equation}
  It thus natively implements the simulation gate for the XXZ model at the isotropic point. (For comparison with Eqs.~\eqref{eq:general_gate} and~\eqref{eq:XXZ}, note  $V_{ij}^{(\mrm{XXX})}(Jt)=U_{ij}(t)$ with $\Delta_{ij}=1$ and $J/4=\la_{ij}$, where $\Delta_{ij}$ is the anisotropy and $\la_{ij}$ the interaction strength in the XXZ model.)

  Noting that
\begin{equation}
    \SWAP_{ij} \,\hat =\, \ee^{-\ii \pi (X_i X_j + Y_i Y_j + Z_i Z_j)/4},
\end{equation}
we can naturally merge a SWAP gate into a native gate by
\begin{equation}
    \SWAP\, V_{ij}^{(\mrm{XXX})}(Jt) = V_{ij}^{(\mrm{XXX})}(Jt)\, \SWAP\, \hat =\, V_{ij}^{(\mrm{XXX})}(Jt + \pi).
\end{equation}
Thus, a SWAP gate can be incorporated into the native two-qubit gate by appropriately rescaling the interaction strength or the gate duration (or both) such that $Jt \rightarrow Jt + \pi$. Under the assumption that $J$ can be made
sufficiently large without impacting the fidelity of $V_{ij}^{\mrm{XXX}}$, this does not lead to an increase in gate time.

As a third example, consider a quantum computer with a native continuous two-qubit gate~\cite{arute2019quantum,foxen2020demonstrating}
\begin{equation}\label{eq:fsim}
  \fSim(\theta,\phi)\, \hat=\, \ee^{-\ii \left[\frac{\theta}{2} (X_iX_j+Y_iY_j)+\frac{\phi}{4} Z_iZ_j\right]}\ee^{\ii \frac{\phi}{4}Z_i}\ee^{\ii \frac{\phi}{4}Z_j},
\end{equation}
with arbitrary iSWAP angle $\theta$ and arbitrary conditional phase angle $\phi$, and where the hat denotes equality up to a global phase. This gate was realized natively in Ref.~\cite{foxen2020demonstrating} for arbitrary $\theta$ and $\phi$ (also see Appendix B. of~\cite{kattemolle2022variational}). It is periodic with period $2\pi$ in both $\theta$ and $\phi$.

From \eq{fsim}, we see the fSim gate natively implements the quantum simulation gate $U_{ij}(t)$ [\eq{general_gate}] for the XXZ model [\eq{XXZ}] with interaction strength $\la_{ij}$ and anisotropy $\Delta_{ij}$, up to single-qubit $Z$-rotations;
\begin{equation}
 U_{ij}^{(\mrm{XXZ})}\left(t\right)=\fSim\!\left(2 t \la_{ij},4t\la_{ij}\Delta_{ij}\right)\ee^{-\ii t(\la_{ij}\Delta_{ij}+\la_i)Z_i}\ee^{-\ii t(\la_{ij}\Delta_{ij}+\la_j)Z_j}.
\end{equation}
Noting that $\SWAP\,\hat{=}\,\fSim(\pi/2,\pi)\ee^{-\ii \frac{\pi}{4} Z_i}\ee^{-\ii \frac{\pi}{4} Z_j}$ and that $\ee^{-\ii \frac{\phi}{4}Z_i}\ee^{-\ii \frac{\phi}{4}Z_j}$ commutes with the fSim gate for all $\theta,\phi$, we have
\begin{equation}
  \SWAP\,U_{ij}^{(\mrm{XXZ})}\left(t\right)\hat=\,\fSim\!\left(2 t \la_{ij}+\pi/2,4t\la_{ij}\Delta_{ij}+\pi\right)\ee^{-\ii \left[t(\la_{ij}\Delta_{ij}+\la_i)+\frac{\pi}{4}\right]Z_i}\ee^{-\ii \left[t(\la_{ij}\Delta_{ij}+\la_j)+\frac{\pi}{4}\right]Z_j}.
\end{equation}
For small arguments $2t\la_{ij}$ and $4t\la_{ij}\Delta_{ij}$ of the $\fSim$ gate ($2t\la_{ij},4t\la_{ij}\Delta_{ij}\sim 0.3\,\mrm{rad}$), the shift of these arguments by $\pi/2$ and $\pi$ does not significantly change the fidelity of the $\fSim$ gate (see Ref.~\cite{foxen2020demonstrating}, Fig. 4). The same holds when $\theta=0.4\pi, \phi=0.8\pi$ (Ref.~\cite{foxen2020demonstrating}, Fig. 4). These angles are used in Ref.~\cite{rosenberg2024dynamics} for the Floquet quantum simulation of an XXZ model. So, on a quantum computer with the parameterized $\fSim$ as a native gate, and when simulating the XXZ model, SWAP gates can be merged with $\fSim$ gates at no increase in infidelity of that gate.

\section{Transpiler}\label{sec:transpiler}
Here, we explain our transpiler in detail, found at \lin{quantile/transpiler.py} in the implementation~\cite{kattemolle2024quantile}. As input, the transpiler receives a basis circuit, describing the circuit to be routed. This can be a reseeded basis circuit (\sec{reseeding}). We call this circuit the \emph{logical circuit} and it consists of \emph{logical gates} acting on \emph{logical qudits}. It also receives a basis graph (which likewise can be reseeded), describing the connectivity of the hardware that the circuit is to be routed to, and some transpilation options.

Without loss of generality, we assume that all logical qudits from the input basis circuit are in a patch of 3 by 3 unit cells centered around the logical cell $(0,0)$. This is without loss of generality because a basis circuit can always be reseeded so that the assumption becomes true.

We initialize an integer depth $D$ to the lower bound given by the longest critical path in the directed acyclic graph (DAG) representation of the circuit (\sec{gate_dependencies}). Using this $D$, an SMT formula encoding the routing problem is constructed and passed to the solver Z3~\cite{moura2008z3}. The integer $D$ is increased until the formula becomes satisfiable. The first $D$ for which the SMT formula is satisfiable is the minimum depth in which the circuit can be routed.

The output of the transpiler is a minimal-depth routed basis circuit that respects the connectivity of the target hardware, as defined by the basis graph of the hardware which  was received as input. We call this routed circuit the \emph{physical circuit}, consisting of \emph{physical gates} acting on \emph{physical qudits}.

In the physical basis circuit, all logical qudits need to be assigned to a physical qudit before the first gates act. We say the logical qudits reside at their assigned physical qudits. This assignment gives the initial logical-to-physical qudit map, or initial \emph{qudit map} for short. If a logical qudit initially resides on physical qudit $i$ and a SWAP gate acts on physical qudits $i$ and $j$, it is convenient to describe the result as the logical qudit moving from $i$ to $j$, even though, strictly speaking, the SWAP gate exchanges states. After each SWAP gate, the qudit map is updated accordingly, resulting in the final qudit map at the end of the circuit. Also the initial and final qudit maps are provided as output of the transpiler.

For very deep circuits, we offer the option to slice the input logical circuit into subcircuits (slices) of a given depth. In that case, we run a transpiler for each of those subcircuits in sequence and stitch together the solutions in postprocessing.

\subsubsection{Assumptions}

In constructing the SMT formula, we make the following assumptions.
\begin{itemize}
  \item[-] \emph{Cell-to-same-cell and seed-to-seed.} We assume the initial map is cell-to-same-cell. This means that each logical qudit $(i,j,s)$ resides in cell $(I,J)$ of the hardware basis graph, with $I=i,J=j$. Furthermore, we assume the initial map is seed-to-seed, which means that if a logical qudit $(i,j,s)$ is mapped to $(I,J,S)$, then all logical qudits of the form $(i',j',s)$ (with $s$ as before) are mapped to $(I',J',S)$ (with $S$ as before).

\end{itemize}

Additionally, for simplicity of formulation of the SMT formula, we make the following assumptions that are without loss of generality.
\begin{itemize}

\item[-] \emph{Reach.} The input basis circuit has a reach of one basis circuit cell. That is, any logical qudit that the basis circuit acts on must be in the unit cell of the basis circuit or in any of the adjacent cells (including diagonally adjacent cells).

\item[-] \emph{Attachment.} Any edges of the basis graph that are fully in one cell are fully in the central cell $(0,0)$. The same holds for basis circuits.

\item[-] \emph{Congruent seeds.} The input basis circuit and the input basis graph have congruent seeds. That is, the set of seed numbers of the qudits that the input basis circuit acts on is $\{0,\ldots,N_\mrm{seeds}\}$, and similarly for the set of seed numbers of the vertices of the basis graph (with a different $N_\mrm{seeds}$).

\item[-] \emph{Minimality.} The hardware basis graph does not have redundant edges. An edge in a basis graph is redundant if it can be removed without altering the induced lattice graph.

\item[-] \emph{Mobility zone.} After the initial time, the logical qudits can move in a \emph{mobility zone}, consisting of the physical qudits in $(2\delta+1§)\times(2\delta+1)$ basis graph cells centered around the cell $(0,0)$. That is, during the physical circuit, logical qudits can move to physical qudits in other unit cells than their initial `home' cell, but they may not leave the mobility zone.
\end{itemize}

\subsubsection{Transpiler options}

The various ways in which our transpiler can be used are best described by the available transpiler options, which are as follows. These options are fixed before the construction of the SMT formula.

\begin{itemize}
\item[-] \lin{cyclic} (boolean) : If set to true, the initial qudit map is imposed to be equal to the final qudit map.

\item[-] \lin{gate_dependencies} (boolean) : If set to true, it is imposed that if a gate $A$ is performed before a gate $B$ in the logical circuit, then also $A$ is performed before $B$ in the physical circuit.

\item[-] \lin{merge_swaps} (boolean) : If set to true, consider $\SWAP\,G$, with the $\SWAP$ gate and the two-qudit gate $G$ acting on the same physical qudits, as a single gate.

\item[-] \lin{slice_depth} (integer or None) : Reschedule the input circuit and slice it into subcircuits of depth \lin{slice_depth}. Route each subcircuit optimally. If no integer is given, the input circuit is not rescheduled and sliced.

\item[-] \lin{minimize_swaps} (boolean) : After minimizing the depth $D$, minimize the number of SWAPs. If \lin{merge_swaps == True}, after minimizing $D$, at that $D$, first minimize the number of naked SWAPs, and at that number of naked SWAPs, minimize the total number of SWAPs. A naked SWAP is a SWAP that is not merged with another two-qudit gate.

\item[-] \lin{fixed_naked_swaps} (boolean or int) : If set to an integer and \lin{merge_swaps == True}, the SMT formula will fix the number of SWAPs that are not merged with a two-qudit gate to the given value.

\end{itemize}
\subsection{Notation}

To specify the SMT variables and constraints, we introduce the following notation. Throughout this section, $A.\tilde b$ denotes the \emph{property} $\tilde b$ of $A$. We stress that $b$ is a property by the tilde. On the other hand, $A_b$ and $A^{(b)}$ are values of $A$ indexed by $b$. The number of logical gates $N_g$, the number of logical qudits $N_q$, and the number of SWAPs $N_S$ are exceptions to these rules. It is stressed by a bar that a symbol refers to an SMT variable. These are only fixed after a satisfying assignment of the variables is found. We generally use uppercase letters to refer to physical objects (such as physical qudits) and lowercase letters to refer to logical objects (such as logical qudits). To summarize notation by an example, consider $\overline{A_b.\tilde C}$. It is an SMT variable corresponding to the property $\tilde C$ of the physical object $A$ indexed by the logical object $b$.

In this notation, a logical basis circuit $c$ consists of a sequence of quantum gates $c=\left(g^{(0)},\ldots,g^{(N_g-1)}\right)$, with $N_g$ the number of logical gates. Each gate $g$ (dropping the gate index for brevity) acts on a logical qudit $g.\tilde q$ if $g$ is a single-qudit gate and on logical qudits $g.\tilde q,g.\tilde q'$ if it is a two-qudit gate. Each qudit in each quantum gate has a cell $x$-coordinate $g.\tilde q.\tilde x$, a cell $y$ coordinate $g.\tilde q.\tilde y$, and a seed number $g.\tilde q.\tilde s$.

This circuit is to be routed, by inserting SWAP gates, to hardware whose connectivity graph is generated by a basis graph.  A basis graph is a finite simple graph with vertices of the form $v=(v_x,v_y,v_s)$. For (basis) graphs and their edges, we do not use the convention that lowercase letters refer to logical objects and uppercase letters refer to physical objects. As for qudits, the integers $v_x,v_y$ denote the cell a vertex is in and the seed number $v_s$ indicates the vertex in that cell. To generate a patch $P$ (with a size of $n$ by $m$ basis graphs) of a hardware connectivity graph, proceed as follows: for all $i\in\{0,\ldots,n-1\}, j\in\{0,\ldots,m-1\}$, copy the hardware's basis graph, translate all coordinates $(v_x,v_y)$ of all vertices $v$ by $(i,j)$ and add the vertices and edges of the translated graph to $P$. We will often consider patches $P$ of size $(n,m)$ centered around cell $(0,0)$. These are patches with $n,m$ odd and where $i,j$ range over $i\in\{-\lfloor n/2\rfloor,\ldots ,\lfloor n/2 \rfloor\}, j\in\{-\lfloor m/2\rfloor,\ldots ,\lfloor m/2 \rfloor\}$.

The routed circuit $C$ consists of a set of physical gates $C=\{G_{g^{(0)}},\ldots,G_{g^{(N_g-1)}}\} \cup \{\SWAP^{(0)},\ldots,\SWAP^{(N_S-1)}\}$, with  $G_g$ the physical gate corresponding to the logical gate $g$, $\SWAP^{(i)}$ the $i$th SWAP gate, and $N_S$ the number of SWAPs. Each physical gate $G_g$ acts at a time $G_{g}.\tilde t$ and acts on physical qudit $G_{g}.{\tilde Q}$ if it is a single-qudit gate, and on physical qudits $G_{g}.{\tilde Q},G_{g}.{\tilde Q'}$ if it is a two-qudit gate. Each physical qudit has a cell $x$-coordinate $(G_g.\tilde Q).\tilde x=G_g.\tilde Q.\tilde x$, cell $y$-coordinate $G_g.\tilde Q.\tilde y$, and a seed number $G_g.\tilde Q.\tilde s$. Each SWAP gate $\SWAP^{(i)}$ acts at time $\SWAP^{(i)}.\tilde t$ along an edge $\SWAP^{(i)}.\tilde e=\{\SWAP^{(i)}.\tilde e_0,\SWAP^{(i)}.\tilde e_1\}$ of the basis graph.

\subsection{Variables}\label{sec:variables}
Here, we list the variables that appear in the SMT formula. In this, we define a depth $\mc D$ that depends on the transpiling options. If the final qudit map is arbitrary, there is no need for SWAP gates in the last layer of the routed circuit. This situation occurs if the routed circuit need not be cyclic (\lin{cyclic == False}). In the definition of the variables below, we therefore find it practical to define the integer $\mc D$ as follows: $\mc D=D-1$ if the final map is arbitrary, and $\mc D=D$ otherwise. The integer $\mc D$ can be thought of as the depth of the circuit in which the SMT formula should account for possible SWAP gates. Note that $D$ and hence $\mc D$ are not SMT variables.

Since, if the final map is arbitrary, there are no SWAPs in the last circuit layer $D-1$, there is also no need to define the qudit map at $D$. We will therefore also use $\mc D$ in the definition of the variables for the qudit map.

\subsubsection{Physical-qudit coordinates of the logical-to-physical-qudit mapping}
For all integer times $0\leq T < \mc D+1$, and all logical qudits $q$, the SMT formula needs to specify at which physical qudit $Q$ the logical qudit resides. For this purpose, we have the $(\mc D+1)\times N_{q} \times 3 $ integer SMT variables
\begin{equation}
\{\overline{Q_{T, q}.\tilde I} \mid  0\leq T < \mc D+1,  q \in \text{logical qudits},\tilde I\in\{\tilde X,\tilde Y,\tilde S\}\}.
\end{equation}
The set of logical qudits is given by $\{g.\tilde q \mid g\in c\} \cup \{g.\tilde q' \mid g\in c \wedge \text{$g$ is a two-qudit gate}\}$. Note in the definition of the variables above, $0\leq T < \mc D+1$ because we have to represent the location of the qudits after the last SWAP gates, which can occur at $T=\mc D-1$ the latest.

At $T=0$, we assume the initial qudit mapping is cell-to-same-cell and seed-to-seed. That is, if a logical qudit $(i,j,s)$ is initially at physical qudit $(I,J,S)$, we enforce $(i',j',s)$ initially resides at $(I',J',S)$, with $I'=i', J'=j'$ for all logical seed numbers $s$ and all possible $i',j'$. The cell-to-same-cell and seed-to-seed property of the initial qudit map can be enforced with constraints on the variables. Instead, in the implementation we achieve this by letting all $\{\overline{Q_{0,q}.\tilde S} \mid q.\tilde s=s_0 \}$ point to the same SMT variable instance for each $s_0$ separately. This is equivalent and leads to fewer variables and constraints.

\subsubsection{Physical-qudit coordinates of physical gates}
With each logical gate $g$, we associate a physical gate $G_g$, acting on qudit $G_g.Q$ if $g$ is a single-qudit gate and on qudits $G_g.Q,G_g.Q'$ if $g$ is a two-qudit gate. So, we have the integer SMT variables
\begin{equation}
  \begin{aligned}
  &\{\overline{G_g.Q.\tilde I} \mid g \in \text{single-qudit logical gates}, \tilde I\in\{\tilde X,\tilde Y,\tilde S\}\}\\
  &\cup\ \{\overline{G_g.A.\tilde I} \mid g \in \text{two-qudit logical gates}, A\in\{Q,Q'\},\tilde I\in \{\tilde X,\tilde Y,\tilde S\}\}.
\end{aligned}
\end{equation}

\subsubsection{Gate times of physical gates}
Each physical gate acts at a specific time, leading to the variables
\begin{equation}
 \{\overline{G_g.\tilde T}\mid G_g\in \text{ physical gates}\}.
\end{equation}

\subsubsection{Activity of SWAP gates}\label{sec:SWAP_variables}

For all times $0\leq T < \mc D$, we consider SWAP gates acting along each edge $\tilde e=\{\tilde e_0,\tilde e_1\}$ of a patch $P_{2\delta+3,2\delta+3}$ of the target hardware, centered around the target hardware unit cell $(0,0)$. Given the mobility zone of $2\delta+1$ by $2\delta+1$ hardware cells in which the logical qudits can move, considering a patch size of $2\delta+3$ by $2\delta+3$ basis graphs for the physical SWAP gates is necessary (and sufficient) since, e.g., in the case $\delta=1$, the hardware basis graph may have an edge $((0,0,0),(-1,0,1))$. If a SWAP gate is placed along this edge in the routed circuit, then in a large enough patch of the routed circuit, a SWAP gate will also be placed along $((2,0,0),(1,0,1))$ by translation of the physical basis circuit by two unit cells in the positive $x$-direction, thus acting on the physical qudit $(1,0,1)$ on which a logical qudit possibly resides.

However, considering a patch $P_{2\delta+3,2\delta+3}$ for the SWAP variables also creates needlessly many SMT variables. Again in the example where $\delta=1$, in $P_{5,5}$ there can also be edges like $((2,0,0), (2,0,1))$, which will never act on a physical qudit of the physical basis circuit that houses a logical qudit. This is because a SWAP on the edge $((2,0,0), (2,0,1))$ fully acts outside the mobility zone of the logical qudits. We therefore exclude any SWAP variables that correspond to SWAP gates acting fully outside the mobility zone. That is, we consider the patch $P_{\mrm{SWAPs}}$ defined by
\begin{equation}\label{eq:SWAP_edges}
  E(P_{\mrm{SWAPs}})=\{e\in E(P_{2\delta+3,2\delta+3}) \mid  (-\delta\leq e_{0_x} \leq \delta \wedge -\delta\leq e_{0_y} \leq \delta ) \vee (-\delta\leq e_{1_x} \leq \delta \wedge -\delta\leq e_{1_y} \leq \delta )\}.
\end{equation}
Here, $e_{i_j}$ is the $j$-coordinate of the $i$th vertex of $e$.

At each time $T$, each SWAP gate can be ``on'' (i.e., it is placed in the routed circuit along edge $e$) or ``off'' (it is not placed in the routed circuit). So, we have the boolean SMT variables
\begin{equation}\label{eq:SWAPvars}
\{\overline{\SWAP_{T,e}} \mid 0 \leq T < \mc D, e \in E(P_{\mrm{SWAPs}})\}.
\end{equation}
The gates $\SWAP^{(i)}$ in the routed circuit are obtained from the above SWAP variables in postprocessing after a satisfying assignment is found.

Presume that in the routed and postprocessed physical basis circuit $C$, at time $T$, a SWAP gate is placed along the edge $e$. When a circuit patch $P$ of size $(n,m)$ is created from this physical basis circuit, then at time $T$, it contains a SWAP gate along all edges $e'=((e_{0_x}+i,e_{0_y}+j,e_{0_s}),(e_{1_x}+i,e_{1_y}+j,e_{1_s}))$ for all $i\in\{0,\ldots,n-1\},j\in\{0,\ldots,m-1\}$. This creates SWAP gates that are not in $C$, but that act on physical qudits from $C$ holding logical qudits nevertheless. For this reason, even if $C$ faithfully implements the basis circuit $c$, the patch $P$ may not faithfully implement the patch $p$ of the same size created by copying and translating the basis circuit $c$.

This effect has to be accounted for. We could do so by adding a constraint to the SMT formula that says that if at time $T$ a SWAP gate acts along an edge $e$ of $P$ (i.e., $\overline{\SWAP_{T,e}}$ is ``on''), then a SWAP gate must act along all edges $e'$ in $P_{\mrm{SWAPs}}$ related to $e$ by translation. The same can be achieved by, for each $e$, letting all variables $\overline{\SWAP_{T,e'}}$ in \eq{SWAPvars} so that $e'$ is related to $e$ by a translation, point to the same SMT boolean variable instance. This is equivalent to adding the constraints but leads to fewer variables and constraints.

Thus, if for some $T$, $\SWAP_{T,e}$ with, e.g., $e=\{(0,0,0),(0,0,1)\}$  is ``on'', then also $\overline{\SWAP_{T,e'}}$ with $e'=\{(1,0,0),(1,0,1)\}$ is ``on''. In a naive postprocessing step, a $\SWAP$ gate will be added to the physical basis circuit along both edges, leading to a SWAP-SWAP collision. We discuss the solution to this problem in section \ref{sec:postprocessing}.

\subsection{Constraints}\label{sec:constraints}

Constraints are put on the SMT variables and are combined into a single SMT formula. In the implementation, the constraints are found in \lin{quantile/constraints.py}.

\subsubsection{Mapping}

At the initial time $T=0$, each logical qudit resides on a physical qudit with the same cell coordinates (i.e., the initial mapping is cell-to-same-cell). Also, naturally, logical qudits have to be mapped to existing physical qudits. Because the basis graph has congruent seeds, with seed numbers ranging from $0$ to $N_\mrm{seeds}$ (exclusive), we can ensure logical qudits are mapped to existing physical qudits by constraining the seed number of the physical qudits that the logical qudits map to. That is, for all logical qudits $q$,
\begin{equation}
  \overline{Q_{0,q}.\tilde X}=q.\tilde x\ \wedge\ \overline{Q_{0,q}.\tilde Y}=q.\tilde y\ \wedge\ 0\leq \overline{Q_{0,q}.\tilde S}< N_{\mrm{seeds}}.
\end{equation}

At later times $T>0$, we constrain the qudits to stay in the mobility zone measuring $(2\delta+1)\times(2\delta+1)$ basis graph cells centered around the central basis graph cell. That is, for all $1 \leq T < \mc D+1$ and all logical qudits $q$,
\begin{equation}
    -\delta\leq \overline{Q_{T,q}.\tilde X} \leq \delta \ \wedge\ -\delta\leq \overline{Q_{T,q}.\tilde Y}\leq \delta\ \wedge\ 0\leq \overline{Q_{T,q}.\tilde S}< N_{\mrm{seeds}}.
\end{equation}
In our software implementation, we assume $\delta=1$ for simplicity. One can always increase the size of the logical circuit and the basis graph (by creating patches and reseeding) to effectively enlarge the mobility zone.

At any time, the qudit map has to be injective. That is, for all $0\leq T< \mc D+1$ and all pairs $\{q,q'\}$ $(q\neq q')$ of logical qudits,
\begin{equation}
  \overline{Q_{T,q}}\neq\overline{Q_{T,q'}},
\end{equation}
which is shorthand notation for $\bigwedge_{\tilde I\in\{\tilde X,\tilde Y,\tilde S\}} \overline{Q_{T,q}.\tilde I}\neq  \overline{Q_{T,q'}.\tilde I}$.

If a slice depth is given (\lin{slice_depth != None}), the transpiler transpiles logical subcircuits of depth \lin{slice_depth} in series. When transpiling one such subcircuit, other than the first, we have to enforce a specific initial qudit map $\mc Q_{0,q}$. That is, for all logical qudits $q$,
\begin{equation}
  \overline{Q_{0,q}}=\mc Q_{0,q}.
\end{equation}

\subsubsection{Time}

All physical gates must be performed in the allotted circuit depth. That is, for all physical gates $G_g$,
\begin{equation}
  0\leq \overline{G_g.\tilde T} < D.
\end{equation}
Note that here, we use $D$ instead of $\mc D$ because there can be gates in the last layer $D-1$, irrespective of whether the final qudit map is arbitrary or not.

\subsubsection{Consistency}

Physical gates have to act on the physical qudit(s) holding the correct logical qudit(s). That is, for all single-qudit gates $G_g$ and $0\leq T < D$,
\begin{equation}
  \overline{G_g.\tilde T}=T\Rightarrow \overline{Q_{T,g.q}}=\overline{G_g.\tilde Q}.
\end{equation}
Here, $\overline{Q_{T,g.q}}=\overline{G_g.Q}$ is shorthand for $\bigwedge_{\tilde I\in\{\tilde X,\tilde Y,\tilde S\}} \overline{Q_{T,g.q}.\tilde I}=\overline{G_g.Q.\tilde I}$. Additionally, for all two-qudit gates $G_g$ and $0\leq T < D$,
\begin{equation}
  \overline{G_g.\tilde T}=T\Rightarrow \overline{Q_{T,g.q}}=\overline{G_g.\tilde Q}\wedge \overline{Q_{T,g.q'}}=\overline{G_g.\tilde Q'}
\end{equation}
in the same notation.

\subsubsection{Connectivity}

Physical gates need to respect the connectivity constraints of the target hardware. That is, physical gates may only be performed along edges of the graph $\mc G$ induced by tiling the hardware basis graph. Since logical qudits may only move within the mobility zone of $(2\delta+1)\times(2\delta+1)$ cells, we can equivalently demand that physical gates may only be performed along those edges in $\mc G$ that are fully within the mobility zone. That is, the set of allowed edges is
\begin{equation}
 E_\mrm{allowed}=\{e\in E(P_{2\delta+1,2\delta+1}) \mid  (-\delta\leq e_{0_x} \leq \delta \wedge -\delta\leq e_{0_y} \leq \delta) \wedge (-\delta\leq e_{1_x} \leq \delta \wedge -\delta\leq e_{1_y} \leq \delta )\},
\end{equation}
with $e_{i_j}$ the $j$-coordinate of the $i$th vertex of $e$, and $P_{2\delta+1,2\delta+1}$ a patch of $(2\delta+1)\times(2\delta+1)$ basis graphs centered around $(0,0)$.

For all two-qudit logical gates $g$,
\begin{equation}
  \bigvee_{e\in E_{\mrm{allowed}}}  (\overline{G_g.\tilde Q}=e_0 \wedge \overline{G_g.\tilde Q'}=e_1) \vee (\overline{G_g.\tilde Q}=e_1 \wedge \overline{G_g.\tilde Q'}=e_0).
\end{equation}
Here, $\overline{G_g.\tilde Q}=e_i$ is shorthand for $\overline{G_g.\tilde Q.\tilde X}=e_{i_x} \wedge \overline{G_g.\tilde Q.\tilde Y}=e_{i_y} \wedge \overline{G_g.\tilde Q.\tilde S}=e_{i_s}$.

\subsubsection{SWAP effect}

If no SWAP gates act on a physical qudit holding a logical qudit, then at the next time step, the physical qudit holds the same logical qudit. Let $I_v$ be the edges in $P_{\mrm{SWAPs}}$ [\eq{SWAP_edges}] incident on a vertex $v$. Then, for all $0\leq T<\mc D$, logical qudits $q$ and vertices $v$ in the mobility zone, we have
\begin{equation}
 \overline{Q_{T,q}}=v \wedge \left(\bigwedge_{e\in I_v} \neg \overline{\SWAP_{T,e}}\right)\Rightarrow \overline{Q_{T+1,q}}=v.
\end{equation}
Similar to earlier, $\overline{Q_{T,q}}=v$ is shorthand for $\overline{Q_{T,q}.\tilde X}=v_x\ \wedge\  \overline{Q_{T,q}.\tilde Y}=v_y\ \wedge \ \overline{Q_{T,q}.\tilde S}=v_s$.

If a SWAP acts along an edge, any logical qudit on a vertex of that edge changes position to the other vertex of the edge. That is, for all $0\leq T<\mc D$, edges $e\in P_{\mrm{SWAPs}}$ [\eq{SWAP_edges}], $i\in\{0,1\}$ and logical qudits $q$,
\begin{equation}
  \left(\overline{Q_{T,q}}=e_{i}\ \wedge\  \overline{\SWAP_{T,e}} \right) \Rightarrow \overline{Q_{T+1,q}}=e_{i\oplus 1},
\end{equation}
with `$\oplus$' addition modulo 2.
\subsubsection{Gate dependencies}

If gate dependencies are to be respected (\lin{gate_dependencies == True}), the gate dependency list is extracted from the basis circuit, as described in \sec{gate_dependencies}. If gate $g$ has to be performed before gate $g'$ in the logical basis circuit, then gate $G_{g}$ has to be performed before gate $G_{g'}$ in the physical basis circuit. That is, for all dependencies $(g,g')$ in the gate dependency list,
\begin{equation}
  G_{g}.\tilde T < G_{g'}.\tilde T.
\end{equation}

\subsubsection{No SWAP-SWAP collisions}

Two active SWAP gates cannot act at the same time on the same qudit, even when the basis circuit is used to generate an arbitrarily large routed circuit patch. That is, for all $0\leq T <\mc D$, all pairs of edges $\{e,e'\}$ ($e\neq e'$) in the hardware basis graph, and $(i,j)\in\{0,1\}^2$,
\begin{equation}
  e_{i_s}=e'_{j_s} \Rightarrow \neg(\overline{\SWAP_{T,e}} \wedge \overline{\SWAP_{T,e'}})
\end{equation}

A condition preventing SWAP-SWAP collisions is omitted in~\cite{tan2020optimal} although we consider it necessary even when routing standard, nontileable circuits. It is omitted presumably because if at $T$, one SWAP gate acts on qudits $Q,Q'$ (here we denote the physical qudits instead of the vertices they are associated with), and, simultaneously, another SWAP gate on $Q',Q''$, and, if furthermore, $Q$ and $Q''$ hold logical qudits, then at $T+1$ the qudit map is not injective. This is already ruled out by the mapping constraints. This situation always arises when there are as many physical qudits as logical qudits.

However, when there are at least two more physical qudits than there are logical qudits, the situation may arise where, e.g., only $Q$ holds a logical qudit. In that case, SWAPs acting along $Q,Q'$ and $Q',Q''$ simultaneously do not cause the qudit map to be two-to-one. Here, the former SWAP gate does not act on any logical qudits, but such SWAPs may be inserted by the SMT solver even in depth-optimal circuits when the number of SWAP gates is not minimized directly. In the mobility zone, there are typically more physical qudits than logical qudits, so we always enforce that there are no SWAP-SWAP collisions.

\subsubsection{No gate-gate collisions}

If gate dependencies are to be respected (\lin{gate_dependencies == True}), gate-gate collisions are automatically excluded. This is because if there exists a dependency between $g^{(i)}$ and $g^{(j)}$ (this dependency can be indirect; by transitivity of gate dependency, also indirect dependencies are respected), then $\overline{G_{g^{(i)}}.\tilde T}\neq\overline{ G_{g^{(j)}}.\tilde T}$. If there is no dependency, then by construction of the dependency list,  $g^{(i)}$ and $g^{(j)}$ do not act on the same seed number (i.e., all $g^{(a)}.\tilde b.\tilde s$ are unequal for all $a\in\{i,j\}$, and for all $\tilde b\in\{\tilde q,\tilde q'\}$ if $g^{(a)}$ is a two-qudit gate or $\tilde b=\tilde q$ if $g^{(a)}$ is a single-qudit gate). It follows $G_{g^{(i)}}$ and $G_{g^{(j)}}$ do not act on the same physical seed number, because the initial map is seed-to-seed and because this property is conserved throughout the routed circuit by construction of the SWAP variables.

If gate dependencies are not enforced (\lin{gate_dependencies == False}), gate-gate collisions have to be excluded. That is, for all pairs of logical gates $\{g,g'\}$ ($g\neq g'$), and all $\tilde A\in\{\tilde Q,\tilde Q'\}$ if $g$ is a two-qudit gate ($\tilde A=\tilde Q$ if $g$ is a single-qudit gate), and all $\tilde B\in\{\tilde Q,\tilde Q'\}$ if $g'$ is a two-qudit gate ($\tilde B=\tilde Q$ if $g'$ is a single-qudit gate),
\begin{equation}
  \overline{G_{g}.\tilde T}=\overline{G_{g'}.\tilde T}\Rightarrow \overline{G_{g}.\tilde A.\tilde S}\neq\overline{G_{g'}.\tilde B.\tilde S}.
\end{equation}

Irrespective of whether gate dependencies are enforced, a collision between a physical gate and a translated version of itself, arising from another copy of the physical basis circuit, are already excluded. This is because if a physical qudit acts on two physical qudits with the same seed, $G_{g}.\tilde Q.\tilde S=G_{g}.\tilde Q'.\tilde S$, then the corresponding logical gate must act on two logical qudits with the same seed, which is impossible if the input basis circuit is a valid basis circuit.  Note that the initial qudit map is seed-to-seed and that this property is conserved throughout the physical circuit by construction of the SWAP variables.

\subsubsection{No gate-SWAP collisions}

If SWAP gates cannot be merged with gates (\lin{merge_swaps == False}), a gate and a SWAP may not act on the same qudit at the same time. Considering tileability, the gate and the SWAP may together not act on two qudits with the same seed number at the same time. That is, for all $0\leq T<\mc D$, logical gates $g$, edges $e$ of the basis graph, qudits $\tilde A\in\{Q,Q'\}$ ($\tilde A=\tilde Q$ if $g$ is a single-qudit gate), and vertices $v$ of $e$,
\begin{equation}\label{eq:no_gate_SWAP}
  \neg (\overline{G_{g}.\tilde T} = T\,\wedge\,(\overline{G_{g}.\tilde A.\tilde S} = v_s)\,\wedge\,\overline{\SWAP_{T,e}}    ).
\end{equation}

If SWAP gates can be merged with two-qudit gates (\lin{merge_swaps == True}), a SWAP and a gate may together still not act on two qudits with the same seed number at the same time \emph{unless} the SWAP and the gate overlap fully. Full spatial overlap of a SWAP along edge $e$ and a physical two-qudit gate $G_g$ is expressed by
\begin{equation}
 h(g,e)=\bigvee_{i\in\{0,1\}}(\overline{G_g.\tilde Q}=e_{0\oplus i} \wedge \overline{G_g.\tilde Q'}=e_{1\oplus i}).
\end{equation}
Let $h$ return `False' if $g$ is a single-qudit gate. The expression $h$ is not yet a condition added to the SMT formula. We do add the condition to the formula that for all $0\leq T < \mc D$ and logical gates $g$,
\begin{equation}
                       \bigwedge_{\substack{\tilde A\\ e\in E(P)\\v\in e}} \left[f(g,T,\tilde A,v)\right] \vee\ \bigvee_{e\in E(P_{\mrm{SWAPs}})}\left[\overline{G_g.\tilde T}=T \wedge h(g,e) \wedge \overline{\SWAP_{T,e}}\right],
\end{equation}
with $B$ the basis graph and $f(g,T,\tilde A,v)$ the condition in \eq{no_gate_SWAP}. Here, as before, $\tilde A=Q$ if $g$ is a single-qudit gate and $\tilde A\in\{Q,Q'\}$ if $g$ is a two-qudit gate.

\subsubsection{Cyclicity}
If the final qudit map is demanded to be equal to the initial qudit map (\lin{cyclic == True}), we have that for all logical qudits $q$,
\begin{equation}
 \overline{Q_{\mc D,q}}=\overline{Q_{0,q}}.
\end{equation}

However, when a slice depth is set (\lin{slice_depth = d}) and the transpiler is transpiling the last subcircuit, instead, the final qudit map should equal the initial map $\mc Q_{0,q}$ of the first circuit slice. That is, for all logical qudits $q$,
\begin{equation}
 \overline{Q_{\mc D,q}}=\mc Q_{0,q}.
\end{equation}

\subsubsection{Minimize SWAPs}

If the number of SWAPs is to be minimized (\lin{minimize_swaps == True}), add
\begin{equation}
    \sum_{T,e}  \overline{\SWAP_{T,e}},
\end{equation}
as a minimization objective, where $T$ runs over all $0\leq T< \mc D$ and $e$ over all edges of the hardware basis graph. Here, $\overline{\SWAP_{T,e}}$ is interpreted as 0 (1) if  $\overline{\SWAP_{T,e}}$ is ``off'' (``on'').

If we allow SWAPs to be merged with overlapping two-qudit gates (\lin{merge_swaps == True}), we favor SWAP gates that have been merged with a two-qudit gate over `naked' SWAP gates, i.e., SWAP gates that have not been merged with a two-qudit gate. In that case, we first minimize the number of naked SWAPs, and at that number of naked SWAPs, we minimize the total number of SWAPs. Formally, let $n(T,e)$ describe whether a SWAP at time $T$ and edge $e$ is naked, even after creating an arbitrarily large patch of the routed circuit,
\begin{equation}
n(T,e)=\overline{\SWAP_{T,e}} \wedge \neg \left\{  \bigvee_{g} [\overline{G_g.\tilde T}=T \wedge h(g,e)]\right\},
\end{equation}
where $g$ ranges over all logical gates in a patch of 3 by 3 logical basis circuits centered around the central cell $(0,0)$. This is not a condition added to the SMT formula.

If the number of SWAPs is to be minimized (\lin{minimize_swaps == True}) and we allow SWAPs to be merged with overlapping two-qudit gates (\lin{merge_swaps == True}), we do first add
\begin{equation} \label{eq:objective1}
 \sum_{T,e}n(T,e)
\end{equation}
as a minimization objective. [If \lin{minimize_swaps == True}, \lin{merge_swaps == True} and the number of naked SWAPs is fixed manually (by setting \lin{fixed_naked_swaps} to an integer $\eta$), then instead we add $\sum_{T,e}n(T,e)=\eta$ as a constraint.] Then, (again if \lin{minimize_swaps == True}, \lin{merge_swaps == True}), we add
\begin{equation} \label{eq:objective2}
    \sum_{T,e}  \overline{\SWAP_{T,e}}
  \end{equation}
as a minimization objective. In both objectives, $T$ runs over all $0\leq T< \mc D$ and $e$ over all edges of the hardware basis graph. In the implementation, the order of adding the objectives, first \eq{objective1} and then \eq{objective2}, implies that first the number of naked SWAPs is minimized, and at that number of naked SWAPs, the total number of SWAPs.

\subsection{Boundary effects}\label{sec:postprocessing}
After a solution of the SMT formula is found, under the (optional) minimization objectives, some postprocessing is necessary. If a solution is found, the satisfying assignments of the variables have to be translated to a standard circuit representation. From this, some final conceptual steps have to be performed to obtain a valid basis circuit and valid routed circuit patches.

First, in routing a single logical (reseeded) basis circuit, the effects of SWAP gates arising from routing other copies of the circuit have to be taken into account. We did so in \sec{SWAP_variables} by ensuring that, whenever a SWAP gate is placed along an edge $e$ of the (reseeded) hardware basis graph, a SWAP gate is also placed along all edges $e'$ so that ($i$) both $e'$ is related to $e$ by translation of the $x$ and $y$ coordinates of its vertices, and ($ii$) $e'$ does not act fully outside the qudit mobility zone. Naively including all those SWAP gates in the physical basis circuit leads to an invalid physical basis circuit since SWAP-SWAP collisions will arise once one repeats the physical basis circuit to generate a physical circuit patch.

One method for solving this issue is to simply accept physical basis circuits that contain SWAP gates that lead to full SWAP-SWAP collisions. After such a physical basis circuit is used to generate a physical circuit patch, one can simply remove any duplicate SWAP gates. There is a complexity of $O(N_G)$ for copying and translating all the basis circuits to form the patch, and also $O(N_G)$ for removing duplicate gates (e.g., by using the bit-vector implementation of sets~\cite{aho1983data}). Thus, generating a circuit patch from the output of the transpiler has complexity $O(N_G)$. This method is implemented as the method \lin{get_patch_fast}. The method has the downside that the basis graph thus output by the transpiler is formally not a valid basis circuit.

We implemented another method, where we insist that also physical basis circuits are valid basis circuits as-is. From all SWAP gates related to one another by translation (at a given time), we only keep a single representative in the physical basis circuit,  namely the one that acts along an edge $e$ in the hardware basis graph. (This edge is unique, for otherwise there would be redundant edges in the hardware basis graph.) This assures the resulting physical basis circuit is a valid basis circuit and it can be treated on an equal footing with any other basis circuit. In generating a physical circuit patch from this physical basis circuit, however, we now need to ensure the correct SWAP gates are included at the boundaries of that patch. After a physical circuit patch $P_{n,m,l}$ is generated of $n$ by $m$ physical basis circuits (of which only single SWAP-representatives were kept), we add back all SWAP gates from a patch of $n+2\delta$ by $m+2\delta$ ($\delta=1$ in the current implementation) physical basis circuits (of which only single SWAP-representatives were kept), centered around $P_{n,m,l}$. Additionally, in this method, we remove all boundary SWAP gates that do not act on a physical qudit holding a logical qudit. In the implementation, this method is slower because it can be optimized further. However, its execution time may still not be significant in practice. The method is implemented as \lin{get_patch}.

If a slice depth was given, the transpiler is run as described above for the subcircuits representing the slices separately. Thereafter, the routed physical solutions are put one after another, thus creating a single physical circuit, which is then rescheduled. In creating patches, this physical basis circuit can be treated as any other physical basis circuit.

\section{Results}\label{sec:results}
We ran the transpiler for the circuits in \sec{examples}. The hardware lattice  graphs these circuits were routed to are the chain (`line'), the ladder, and the square lattice. (Only 1D basis circuits were routed to the chain and the ladder, which are 1D hardware basis graphs). We make tables sorted by the depth overhead that was achieved (in ascending order). The main point of these tables is to identify circuits that can be run at no or little overhead at the given hardware connectivity available, which here we assumed to be a chain, a ladder, or a square lattice. The fact that our routing method is locally optimal indicates that the order of the basis circuits in the tables is not an artifact of the routing method and therefore truly indicates which basis circuits are easiest to perform on the given hardware.

Out of all runs, which can be inspected at the repository~\cite{kattemolle2024quantile}, here, we report the ones we deem the most interesting. For example, if for a given logical basis circuit, basis graph, and the transpiler options \lin{merge_swaps} and \lin{cyclic}, a routing solution with depth $d$ was found, here, we do not report the lines with the same logical basis circuit, basis graph, and the transpiler options \lin{merge_swaps} and \lin{cyclic} that also attain depth $d$, but with a larger basis graph size.

Concretely, out of all runs, for each fixed combination of logical basis circuit, basis graph, and the transpiler options \lin{merge_swaps} and \lin{cyclic}, we select the run that is listed in the table as follows.

\begin{itemize}
    \item[-] \emph{Minimize depth overhead.}
    We retain only the runs with the least depth overhead. The minimum possible depth overhead, which the transpiler is guaranteed to find, may depend on the sizes of the logical basis circuit and basis graph, as these are not optimized by the transpiler.

    \item[-] \emph{Minimize naked SWAPs.}
    From the remaining runs, we keep the ones with the fewest naked SWAPs. The number of naked SWAPs can vary depending on \lin{minimize_swaps}, \lin{fixed_nkd_swaps}, and whether certain runs timed out.

    \item[-] \emph{Minimize basis graph size.}
    Among the remaining runs, we select the ones with the smallest basis graph, measured by the number of its seeds.

    \item[-] \emph{Prioritize \lin{minimize_swaps}.}
    If runs with \lin{minimize_swaps == True} exists, we select them; otherwise, we take the runs with \lin{minimize_swaps == False}.

    \item[-] \emph{Maximize squareness.}
    We retain the runs with the `squarest' logical basis circuit and basis graph, where squareness is quantified as
    $\mathfrak{s} = \lvert n - m \rvert + \lvert n' - m' \rvert$, with $(n,m)$ representing the logical basis circuit size (in terms of unreseeded basis circuits) and $(n',m')$ the basis graph size. A lower $\mathfrak{s}$ indicates greater squareness.

    \item[-] \emph{Minimize wall clock time.}
    Finally, we select the run with the lowest wall clock time. Before this stage, typically only one run remains, so it has an insignificant effect on the reported wall-clock times.
\end{itemize}

The tables report the following in sequence.
\begin{itemize}
\item[-] \emph{Basis circuit.} The name of the logical basis circuit. In the case of the circuits for the simulation of two-local Hamiltonians, the circuit is named after the basis graph on which the circuit is defined. These basis graphs include all Archimedean and Laves lattices. All basis graphs are depicted in \sec{basis_graphs_database}.

\item[-] \emph{Basis circuit size.} Size of the reseeded logical basis circuit, in terms of unreseeded basis circuits. The reseeded basis circuit is input to the transpiler.

\item[-] \emph{Basis graph.} The target hardware basis graph inducing the target hardware connectivity. The chain (`line'|), ladder, and square lattice are considered. For the square lattice we consider three different unreseeded basis graphs (all inducing the square lattice), which are also depicted in \sec{basis_graphs_database}.

\item[-] \emph{Basis graph size.} Size of the reseeded basis graph to which the circuit was routed, in terms of unreseeded basis graphs. The reseeded patch is input to the transpiler. Note that in the physical (i.e., routed) basis circuit, qudits and gates may be placed outside this reseeded basis graph, but not outside the mobility zone of $2\delta+1$ by $2\delta+1$ (with $\delta=1$ in the software implementation) cells of the reseeded basis graph centered around $(0,0)$.

\item[-] \emph{Depth overhead.} The depth of the physical (i.e., routed) basis circuit minus the minimal depth of the input reseeded logical basis circuit. The latter is set by the longest critical path of the DAG representation of the wrapped reseeded circuit (\sec{gate_dependencies}). This is equal to the depth of the reseeded circuit after wrapping and rescheduling. For the simulation of two-local Hamiltonians, the order of the gates is not fixed. This raises the question of what order achieves the minimal depth.  This is answered in Ref.~\cite{kattemolle2024edge}; for all logical basis circuits listed in the two-local table, the minimal depth is the maximum degree of the infinite lattice graph that is generated by the basis graph, which, as it is shown, is equal to the maximum degree of a 3 by 3 patch (in terms of reseeded basis circuits) after it is wrapped. The depth overhead relative to the minimal depth is shown in parentheses.

\item[-] \emph{Merge SWAPs.} Whether or not it is assumed that SWAP gates can be merged into directly preceding two-qudit gates acting on the same two qudits. (Transpiler option \lin{merge_swaps}.) This effectively also means a SWAP gate can be merged into a subsequent two-qudit gate by interchanging the `control' and `target' physical qubits of the gate.

\item[-] \emph{Cyclic.} Whether or not, at the end of the circuit, the logical qudits must return to their initial position. This is needed, for example, for routing circuits arising from first-order Trotter formulas, but not, for example, when routing second-order Trotter formulas (\sec{trotterization}). (Transpiler option \lin{cyclic}.)

\item[-] \emph{Qudit overh.} The number of qudits in one cell of the reseeded input logical basis circuit minus the number of vertices in one cell of the reseeded hardware basis graph.

\item[-] \emph{Naked SWAPs.} The number of SWAP gates (per physical reseeded basis circuit) that were not merged into preceding two-qudit gates. (If SWAP gates cannot be merged, \lin{merge_swaps == False}, all SWAP gates inserted by routing are naked.) An asterisk denotes that for the relevant data line, the number of naked SWAPs was not minimized (i.e., that the transpiler option \lin{minimize_swaps} was set to \lin{False}) because it was considered computationally too expensive with the current inefficient implementation of SWAP minimization. That is, an asterisk denotes that the shown number of naked SWAPs may not be minimal.

\item[-] \emph{Slice depth.} (Only for Kogut-Susskind.) The slice depth used by the transpiler (\sec{constraints}). The Kogut-Susskind circuits are much deeper than any other circuit considered in this work and are the only circuits that were sliced before being routed.

\item[-] \emph{Wall-clock time}. The wall-clock time of the run, reported in the format days:hours:minutes:seconds. Computations were performed on the scientific compute cluster of the University of Konstanz. All runs were executed using a single thread on an Intel Xeon Processor E5-2680 v2 (25M Cache, 2.80 GHz). An asterisk denotes the run was obtained with a fixed number of naked SWAPs (\lin{fixed_nkd_swaps == True}), which may affect the wall-clock time compared to an unfixed but minimized number of naked SWAPs.

In any case, the reported wall-clock times should be interpreted with great caution for the following reasons.

        \begin{itemize}

        \item[-] \emph{Large variations.} We observed fluctuations of up to a factor of 10 in wall-clock time for identical transpilation problems.

        \item[-] \emph{Unoptimized implementation.} The primary goal of this work is to construct optimal routing solutions that can be repeated spatially and temporally at negligible cost, with the main point being the scaling behavior of the run-time as patch sizes are increased. Consequently, we did not focus on optimizing the implementation of the routing solution.

        Significant improvements in wall-clock time have been reported using highly optimized implementations. Our implementation (leveraging periodicity) is based on the method of OLSQ  (not leveraging periodicity, introduced in Ref.~\cite{tan2020optimal}). A highly optimized version of OLSQ was presented in Ref.~\cite{lin2023scalable}, achieving a 692× speed up for depth optimization and a 6\,957× speed up for SWAP optimization.

        \item[-] \emph{Prudent constraints.} Compared to OLSQ~\cite{tan2020optimal}, we impose an additional constraint that eliminates SWAP-SWAP collisions. While such collisions are rare, we found that preventing them may slow down computation.

        \item[-] \emph{Minimized SWAPs.} Typically not only the depth is minimized, but also the number of SWAPs, which in some cases we found can greatly increase the wall-clock time. For SWAP minimization, we use the optimizer in Z3, which is known to be unnecessarily slow~\cite{lin2023scalable}, and which can be replaced with a custom, more efficient minimization strategy.

        \item[-] \emph{Gate order.} For the circuits for the simulation of two-local Hamiltonians, the order of the gates is also optimized over, as opposed to other routing methods.
        \end{itemize}
\end{itemize}

The lines are sorted in the following order of properties: depth overhead, naked SWAPs, qudit overhead, merge SWAPs, and cyclic. Except for ``merge SWAPs'' (sorted in descending order), all lines are sorted in ascending order. For sorting boolean variables, we use $\texttt{True}>\texttt{False}$.

\clearpage
\subsection{Arbitrary two-local}\label{sec:results_two-local}
\rowcolors{5}{}{gray!20}
\addtolength{\tabcolsep}{-0.4em}
 \begin{longtable}{lclccccccl}
 \hline
 $\ $                   & Basis   & $\ $      & Basis        & $\ $       & $\ $    & $\ $              & $\ $  & $\ $            & Wall-clock       \\
 Basis                  & circuit & Basis     & graph        & Depth      & Naked   & Qudit             & Merge & $\ $            & time       \\
 circuit                & Size    & graph     & size         & overhead   & SWAPs   & overh.            & SWAPs & Cyclic          & \texttt{d:hh:mm:ss}  \\
 \hline
 \endhead
 J1J2-ladder            & (2,1)  & ladder         & (2,1)  & 0 (0 \%)   & 0   & 0 & True  & False    & \texttt{0:00:00:23*} \\
 J1J2-ladder            & (2,1)  & line           & (4,1)  & 0 (0 \%)   & 0   & 0 & True  & False    & \texttt{0:00:00:20*} \\
 J1J2-ladder            & (2,1)  & square         & (2,2)  & 0 (0 \%)   & 0   & 0 & True  & False    & \texttt{0:00:00:39}  \\
 J1J2-line              & (4,1)  & ladder         & (2,1)  & 0 (0 \%)   & 0   & 0 & True  & False    & \texttt{0:00:00:09}  \\
 J1J2-square            & (2,2)  & square         & (2,2)  & 0 (0 \%)   & 0   & 0 & True  & False    & \texttt{0:00:48:48}  \\
 dice                   & (2,1)  & square         & (3,2)  & 0 (0 \%)   & 0   & 0 & True  & False    & \texttt{0:01:06:18*} \\
 ladder                 & (2,1)  & line           & (4,1)  & 0 (0 \%)   & 0   & 0 & True  & False    & \texttt{0:00:00:03}  \\
 trellis                & (2,1)  & square         & (2,2)  & 0 (0 \%)   & 0   & 0 & True  & False    & \texttt{0:00:01:01}  \\
 triangular             & (2,2)  & square         & (2,2)  & 0 (0 \%)   & 0   & 0 & True  & False    & \texttt{0:00:07:58}  \\
 tetrille               & (1,1)  & square         & (3,2)  & 0 (0 \%)   & 1   & 0 & True  & False    & \texttt{0:02:16:55}  \\
 dice                   & (1,1)  & diamond-square & (2,1)  & 0 (0 \%)   & 2   & 1 & True  & False    & \texttt{0:00:08:02}  \\
 floret-pentagonal      & (1,1)  & square         & (3,3)  & 0 (0 \%)   & 4   & 0 & False & False    & \texttt{0:09:02:31}  \\
 union-jack             & (2,2)  & diamond-square & (2,2)  & 0 (0 \%)   & 4*  & 0 & True  & True     & \texttt{2:12:12:38}  \\
 dice                   & (1,1)  & diamond-square & (2,1)  & 0 (0 \%)   & 4   & 1 & True  & True     & \texttt{0:00:03:30}  \\
 dice                   & (1,1)  & square         & (2,2)  & 0 (0 \%)   & 4   & 1 & True  & True     & \texttt{0:00:10:02}  \\
 dice                   & (1,1)  & square         & (2,2)  & 0 (0 \%)   & 5   & 1 & False & False    & \texttt{0:00:03:58}  \\
 floret-pentagonal      & (1,1)  & square         & (3,3)  & 0 (0 \%)   & 6*  & 0 & True  & False    & \texttt{0:02:29:22}  \\
 kisrhombille           & (1,1)  & square         & (3,2)  & 0 (0 \%)   & 6*  & 0 & True  & False    & \texttt{6:13:34:06}  \\
 union-jack             & (2,2)  & square         & (3,3)  & 0 (0 \%)   & 8*  & 1 & True  & True     & \texttt{5:02:36:56}  \\
 kisrhombille           & (1,1)  & square         & (3,2)  & 0 (0 \%)   & 12* & 0 & False & False    & \texttt{0:17:27:51}  \\
 J1J2-line              & (4,1)  & line           & (4,1)  & 1 (25 \%)  & 0   & 0 & True  & False    & \texttt{0:00:00:17*} \\
 bridge                 & (1,1)  & square         & (3,2)  & 1 (20 \%)  & 0   & 0 & True  & False    & \texttt{0:04:18:20*} \\
 kagome                 & (2,2)  & square         & (4,3)  & 1 (25 \%)  & 0   & 0 & True  & False    & \texttt{2:12:17:46*} \\
 prismatic-pentagonal   & (2,1)  & square         & (3,2)  & 1 (25 \%)  & 0   & 0 & True  & False    & \texttt{0:00:15:04*} \\
 shuriken               & (1,1)  & square         & (3,2)  & 1 (25 \%)  & 0   & 0 & True  & False    & \texttt{0:00:26:03*} \\
 snub-square            & (1,1)  & square         & (2,2)  & 1 (20 \%)  & 0   & 0 & True  & False    & \texttt{0:00:13:41}  \\
 star                   & (1,1)  & square         & (3,2)  & 1 (33 \%)  & 0   & 0 & True  & False    & \texttt{0:00:03:34*} \\
 J1J2-ladder            & (2,1)  & ladder         & (2,1)  & 1 (20 \%)  & 1   & 0 & False & False    & \texttt{0:00:00:45}  \\
 J1J2-ladder            & (2,1)  & square         & (2,2)  & 1 (20 \%)  & 1   & 0 & False & False    & \texttt{0:00:01:45}  \\
 J1J2-line              & (4,1)  & ladder         & (2,1)  & 1 (25 \%)  & 1   & 0 & False & False    & \texttt{0:00:00:14}  \\
 trellis                & (2,1)  & square         & (2,2)  & 1 (20 \%)  & 1   & 0 & False & False    & \texttt{0:00:02:14}  \\
 J1J2-ladder            & (2,1)  & ladder         & (2,1)  & 1 (20 \%)  & 1   & 0 & True  & True     & \texttt{0:00:01:24*} \\
 J1J2-ladder            & (2,1)  & square         & (2,2)  & 1 (20 \%)  & 1   & 0 & True  & True     & \texttt{0:00:02:49*} \\
 J1J2-line              & (4,1)  & ladder         & (2,1)  & 1 (25 \%)  & 1   & 0 & True  & True     & \texttt{0:00:00:31*} \\
 trellis                & (2,1)  & square         & (2,2)  & 1 (20 \%)  & 1   & 0 & True  & True     & \texttt{0:00:14:10}  \\
 cairo-pentagonal       & (1,1)  & square         & (3,2)  & 1 (25 \%)  & 1   & 0 & True  & False    & \texttt{0:02:43:34}  \\
 cross                  & (1,1)  & square         & (3,4)  & 1 (33 \%)  & 1   & 0 & True  & False    & \texttt{1:02:06:17}  \\
 J1J2-line              & (4,1)  & line           & (4,1)  & 1 (25 \%)  & 2   & 0 & False & False    & \texttt{0:00:00:10}  \\
 J1J2-square            & (2,2)  & square         & (2,2)  & 1 (12 \%)  & 2   & 0 & False & False    & \texttt{0:03:31:05}  \\
 ladder                 & (2,1)  & line           & (4,1)  & 1 (33 \%)  & 2   & 0 & False & False    & \texttt{0:00:00:04}  \\
 prismatic-pentagonal   & (2,1)  & square         & (3,2)  & 1 (25 \%)  & 2   & 0 & False & False    & \texttt{0:00:08:54}  \\
 triangular             & (2,2)  & square         & (2,2)  & 1 (17 \%)  & 2   & 0 & False & False    & \texttt{0:00:20:36}  \\
 J1J2-square            & (2,2)  & square         & (2,2)  & 1 (12 \%)  & 2   & 0 & True  & True     & \texttt{0:03:23:09*} \\
 ladder                 & (2,1)  & line           & (4,1)  & 1 (33 \%)  & 2   & 0 & True  & True     & \texttt{0:00:00:12}  \\
 prismatic-pentagonal   & (2,1)  & square         & (3,2)  & 1 (25 \%)  & 2   & 0 & True  & True     & \texttt{0:00:24:16}  \\
 triangular             & (2,2)  & square         & (2,2)  & 1 (17 \%)  & 2   & 0 & True  & True     & \texttt{0:02:49:18}  \\
 shuriken               & (1,1)  & square         & (3,3)  & 1 (25 \%)  & 2   & 3 & False & False    & \texttt{0:01:31:38}  \\
 star                   & (1,1)  & square         & (3,3)  & 1 (33 \%)  & 2   & 3 & False & False    & \texttt{0:00:11:20}  \\
 cross                  & (1,1)  & square         & (4,4)  & 1 (33 \%)  & 2   & 4 & False & False    & \texttt{0:04:42:49}  \\
 tetrille               & (1,1)  & square         & (3,2)  & 1 (17 \%)  & 3   & 0 & False & False    & \texttt{0:02:03:54}  \\
 ruby                   & (1,1)  & square         & (3,3)  & 1 (25 \%)  & 3   & 3 & True  & False    & \texttt{0:14:23:38}  \\
 dice                   & (1,1)  & diamond-square & (2,1)  & 1 (17 \%)  & 4   & 1 & False & False    & \texttt{0:00:06:59}  \\
 tetrille               & (1,1)  & square         & (3,2)  & 1 (17 \%)  & 6   & 0 & True  & True     & \texttt{0:03:48:46}  \\
 union-jack             & (2,2)  & square         & (4,2)  & 1 (12 \%)  & 7*  & 0 & False & False    & \texttt{6:10:33:33}  \\
 floret-pentagonal      & (1,1)  & square         & (3,3)  & 1 (17 \%)  & 7   & 0 & True  & True     & \texttt{6:04:28:39}  \\
 J1J2-ladder            & (2,1)  & ladder         & (2,1)  & 2 (40 \%)  & 2   & 0 & False & True     & \texttt{0:00:01:32}  \\
 J1J2-ladder            & (2,1)  & square         & (2,2)  & 2 (40 \%)  & 2   & 0 & False & True     & \texttt{0:00:02:02}  \\
 J1J2-line              & (4,1)  & ladder         & (2,1)  & 2 (50 \%)  & 2   & 0 & False & True     & \texttt{0:00:00:33}  \\
 trellis                & (2,1)  & square         & (2,2)  & 2 (40 \%)  & 2   & 0 & False & True     & \texttt{0:00:03:32}  \\
 snub-square            & (1,1)  & square         & (2,2)  & 2 (40 \%)  & 2   & 0 & False & False    & \texttt{0:00:20:48}  \\
 J1J2-line              & (4,1)  & line           & (4,1)  & 2 (50 \%)  & 2   & 0 & True  & True     & \texttt{0:00:00:45*} \\
 shuriken               & (1,1)  & square         & (3,2)  & 2 (50 \%)  & 2   & 0 & True  & True     & \texttt{0:06:22:39}  \\
 snub-square            & (1,1)  & square         & (2,2)  & 2 (40 \%)  & 2   & 0 & True  & True     & \texttt{0:02:30:47}  \\
 star                   & (1,1)  & square         & (3,2)  & 2 (67 \%)  & 2   & 0 & True  & True     & \texttt{0:01:18:24}  \\
 kagome                 & (1,1)  & square         & (2,2)  & 2 (50 \%)  & 2   & 1 & False & False    & \texttt{0:00:01:32}  \\
 J1J2-ladder            & (2,1)  & line           & (4,1)  & 2 (40 \%)  & 3   & 0 & False & False    & \texttt{0:00:01:42}  \\
 cairo-pentagonal       & (1,1)  & square         & (3,2)  & 2 (50 \%)  & 3   & 0 & False & False    & \texttt{0:02:27:01}  \\
 ruby                   & (1,1)  & square         & (3,2)  & 2 (50 \%)  & 3   & 0 & False & False    & \texttt{0:06:12:43}  \\
 J1J2-line              & (4,1)  & line           & (4,1)  & 2 (50 \%)  & 4   & 0 & False & True     & \texttt{0:00:00:24}  \\
 J1J2-square            & (2,2)  & square         & (2,2)  & 2 (25 \%)  & 4   & 0 & False & True     & \texttt{0:02:54:26}  \\
 ladder                 & (2,1)  & line           & (4,1)  & 2 (67 \%)  & 4   & 0 & False & True     & \texttt{0:00:00:12}  \\
 prismatic-pentagonal   & (2,1)  & square         & (3,2)  & 2 (50 \%)  & 4   & 0 & False & True     & \texttt{0:01:00:30}  \\
 triangular             & (2,2)  & square         & (2,2)  & 2 (33 \%)  & 4   & 0 & False & True     & \texttt{0:00:18:18}  \\
 cairo-pentagonal       & (1,1)  & square         & (3,2)  & 2 (50 \%)  & 4   & 0 & True  & True     & \texttt{0:02:38:14}  \\
 kagome                 & (1,1)  & square         & (2,2)  & 2 (50 \%)  & 4   & 1 & False & True     & \texttt{0:00:00:53}  \\
 kagome                 & (1,1)  & square         & (2,2)  & 2 (50 \%)  & 4   & 1 & True  & True     & \texttt{0:00:13:42}  \\
 shuriken               & (1,1)  & square         & (3,3)  & 2 (50 \%)  & 4   & 3 & False & True     & \texttt{0:04:40:50}  \\
 star                   & (1,1)  & square         & (3,3)  & 2 (67 \%)  & 4   & 3 & False & True     & \texttt{0:00:23:12}  \\
 cross                  & (1,1)  & square         & (4,4)  & 2 (67 \%)  & 4   & 4 & False & True     & \texttt{2:01:43:49}  \\
 tetrille               & (1,1)  & square         & (3,2)  & 2 (33 \%)  & 6   & 0 & False & True     & \texttt{0:02:13:01}  \\
 dice                   & (1,1)  & diamond-square & (2,1)  & 2 (33 \%)  & 6   & 1 & False & True     & \texttt{0:00:17:49}  \\
 dice                   & (1,1)  & square         & (2,2)  & 2 (33 \%)  & 6   & 1 & False & True     & \texttt{0:00:29:03}  \\
 union-jack             & (2,2)  & diamond-square & (2,2)  & 2 (25 \%)  & 8*  & 0 & False & True     & \texttt{4:21:53:00}  \\
 ruby                   & (2,1)  & square         & (4,3)  & 2 (50 \%)  & 8*  & 0 & True  & True     & \texttt{0:09:32:23}  \\
 floret-pentagonal      & (1,1)  & square         & (3,3)  & 2 (33 \%)  & 12* & 0 & False & True     & \texttt{0:05:35:54}  \\
 union-jack             & (2,2)  & square         & (4,2)  & 2 (25 \%)  & 12* & 0 & False & True     & \texttt{3:16:41:07}  \\
 cross                  & (1,1)  & square         & (4,4)  & 2 (67 \%)  & 16* & 4 & True  & True     & \texttt{0:00:56:31}  \\
 J1J2-ladder            & (2,1)  & line           & (4,1)  & 3 (60 \%)  & 2   & 0 & True  & True     & \texttt{0:00:02:37*} \\
 J1J2-ladder            & (2,1)  & line           & (4,1)  & 3 (60 \%)  & 4   & 0 & False & True     & \texttt{0:00:02:26}  \\
 snub-square            & (1,1)  & square         & (2,2)  & 3 (60 \%)  & 4   & 0 & False & True     & \texttt{0:00:34:48}  \\
 bridge                 & (1,1)  & square         & (3,2)  & 3 (60 \%)  & 4   & 0 & False & False    & \texttt{4:09:17:54}  \\
 bridge                 & (1,1)  & square         & (3,2)  & 3 (60 \%)  & 4*  & 0 & True  & True     & \texttt{1:06:13:10}  \\
 cairo-pentagonal       & (1,1)  & square         & (3,2)  & 3 (75 \%)  & 6   & 0 & False & True     & \texttt{0:07:34:26}  \\
 bridge                 & (2,1)  & square         & (4,3)  & 3 (60 \%)  & 12* & 0 & False & True     & \texttt{5:03:34:17}  \\
 ruby                   & (1,1)  & square         & (3,2)  & 4 (100 \%) & 6   & 0 & False & True     & \texttt{2:07:18:57}  \\
 \hline
 \end{longtable}

\subsection{Rule 54}

\rowcolors{5}{}{gray!20}
{\begin{longtable}{lclccccccl}
  \hline
 $\ $                   & Basis          & $\ $      & Basis        & $\ $       & $\ $    & $\ $              & $\ $  & $\ $            & Wall-clock       \\
 Basis                  & circuit        & Basis     & graph        & Depth      & Naked   & Qudit             & Merge & $\ $            & time       \\
 circuit                & Size           & graph     & size         & overhead   & SWAPs   & overh.            & SWAPs & Cyclic          & \texttt{d:hh:mm:ss}  \\
\hline
  \endhead
 rule54                 & (1,1)  & ladder        & (2,1)  & 0 (0 \%)  &  1 & 0 & True  & True     & \texttt{0:00:04:11} \\
 rule54                 & (2,1)  & square        & (4,2)  & 0 (0 \%)  &  2 & 0 & True  & True     & \texttt{0:02:30:08} \\
 rule54                 & (1,1)  & ladder        & (2,1)  & 1 (5 \%)  &  8 & 0 & False & True     & \texttt{0:00:06:10} \\
 rule54                 & (2,1)  & square        & (4,2)  & 1 (5 \%)  & 16 & 0 & False & True     & \texttt{0:04:24:32} \\
 rule54                 & (1,1)  & line          & (4,1)  & 3 (15 \%) &  4 & 0 & True  & True     & \texttt{0:00:50:40} \\
 rule54                 & (1,1)  & line          & (4,1)  & 6 (30 \%) &  8 & 0 & False & True     & \texttt{0:00:29:49} \\
\hline
\end{longtable}}

\subsection{Rokhsar-Kivelson}
\rowcolors{5}{}{gray!20}
\begin{longtable}{lclccccccl}
  \hline
 $\ $                   & Basis          & $\ $      & Basis        & $\ $       & $\ $    & $\ $              & $\ $  & $\ $            & Wall-clock       \\
 Basis                  & circuit        & Basis     & graph        & Depth      & Naked   & Qudit             & Merge & $\ $            & time       \\
 circuit                & Size           & graph     & size         & overhead   & SWAPs   & overh.            & SWAPs & Cyclic          & \texttt{d:hh:mm:ss}  \\
\hline
  \endhead
Rokhsar-Kivelson\_xyz   & (1,1)  & square        & (2,2)  & 2 (11 \%) & 2   & 0 & True  & True     & \texttt{0:00:43:02} \\
 Rokhsar-Kivelson\_mcu  & (1,1)  & square        & (2,2)  & 2 (6 \%)  & 18* & 0 & True  & True     & \texttt{0:03:01:00} \\
 Rokhsar-Kivelson\_xyz  & (1,1)  & square        & (2,2)  & 3 (17 \%) & 4   & 0 & False & True     & \texttt{0:00:23:59} \\
  Rokhsar-Kivelson\_mcu & (1,1)  & square        & (2,2)  & 3 (9 \%)  & 18  & 0 & False & True     & \texttt{4:14:21:51} \\
\hline
\end{longtable}

\subsection{Fermi-Hubbard}
\rowcolors{5}{}{gray!20}
\begin{longtable}{lclccccccl}
  \hline
 $\ $                   & Basis          & $\ $      & Basis        & $\ $       & $\ $    & $\ $              & $\ $  & $\ $            & Wall-clock       \\
 Basis                  & circuit        & Basis     & graph        & Depth      & Naked   & Qudit             & Merge & $\ $            & time       \\
 circuit                & Size           & graph     & size         & overhead   & SWAPs   & overh.            & SWAPs & Cyclic          & \texttt{d:hh:mm:ss}  \\
\hline
  \endhead
 Fermi-Hubbard          & (1,1)       & plus-square   & (1,1)       & 0 (0 \%) & 0 & 0 & False & True     & \texttt{0:00:01:18} \\
\hline
\end{longtable}

\clearpage
\subsection{Kogut-Susskind}
\rowcolors{5}{}{gray!20}
\begin{longtable}{lclcccccccl}
  \hline
 $\ $                   & Basis          & $\ $      & Basis        & $\ $       & $\ $    & $\ $              & $\ $  & $\ $      & $\ $       & Wall-clock       \\
 Basis                  & circuit        & Basis     & graph        & Depth      & Naked   & Qudit             & Merge & $\ $      & Slice      & time       \\
 circuit                & Size           & graph     & size         & overhead   & SWAPs   & overh.            & SWAPs & Cyclic    & depth      & \texttt{d:hh:mm:ss}  \\
 \hline
 \endhead
 kogut-susskind         & (1,1)  & diamond-square & (2,2)  & 7 (3 \%)   & 284* & 1 & True  & True     &            40 & \texttt{0:04:18:50} \\
 kogut-susskind         & (1,1)  & diamond-square & (2,2)  & 9 (4 \%)   & 253  & 1 & True  & True     &            20 & \texttt{0:01:41:25} \\
 kogut-susskind         & (2,1)  & square         & (4,4)  & 12 (6 \%)  & 626* & 2 & False & True     &            40 & \texttt{2:13:31:44} \\
 kogut-susskind         & (1,1)  & diamond-square & (2,2)  & 13 (6 \%)  & 262  & 1 & False & True     &            20 & \texttt{0:00:40:02} \\
 kogut-susskind         & (1,1)  & diamond-square & (2,2)  & 13 (6 \%)  & 290* & 1 & False & True     &            40 & \texttt{0:04:28:20} \\
 kogut-susskind         & (1,1)  & diamond-square & (2,2)  & 14 (7 \%)  & 201* & 1 & True  & True     &            10 & \texttt{0:01:31:30} \\
 kogut-susskind         & (1,1)  & diamond-square & (2,2)  & 21 (10 \%) & 198  & 1 & False & True     &            10 & \texttt{0:01:17:32} \\
 kogut-susskind         & (2,1)  & square         & (4,4)  & 22 (10 \%) & 634  & 2 & False & True     &            20 & \texttt{1:19:45:51} \\
 kogut-susskind         & (2,1)  & square         & (4,4)  & 28 (13 \%) & 608  & 2 & True  & True     &            10 & \texttt{0:16:13:24} \\
 kogut-susskind         & (2,1)  & square         & (4,4)  & 32 (15 \%) & 563  & 2 & False & True     &            10 & \texttt{0:04:07:11} \\
\hline
\end{longtable}

\subsection{Main-text table}
For completeness, here we repeat the table of the main text with all available data fields and in the notation of the other tables in the SM. The rows are in the same order as in the main text. All rows are taken from the preceding tables.

\setlength{\tabcolsep}{.5em}
\rowcolors{5}{}{gray!20}
\begin{longtable}{lclcccccccl}
  \hline
 $\ $                   & Basis          & $\ $      & Basis        & $\ $       & $\ $    & $\ $              & $\ $  & $\ $      & $\ $       & Wall-clock       \\
 Basis                  & circuit        & Basis     & graph        & Depth      & Naked   & Qudit             & Merge & $\ $      & Slice      & time       \\
 circuit                & Size           & graph     & size         & overhead   & SWAPs   & overh.            & SWAPs & Cyclic    & depth      & \texttt{d:hh:mm:ss}  \\
 \hline
\endhead
ladder                & (2,1) & line           & (4,1) & 0 (0 \%)  & 0   & 0 & True  & False & N.A.          & \texttt{0:00:00:03}  \\
J1J2-ladder           & (2,1) & line           & (4,1) & 0 (0 \%)  & 0   & 0 & True  & False & N.A.          & \texttt{0:00:00:20*} \\
J1J2-line             & (4,1) & line           & (4,1) & 1 (25 \%) & 0   & 0 & True  & False & N.A.          & \texttt{0:00:00:17*} \\
\hline
rule54                & (1,1) & ladder         & (2,1) & 0 (0 \%)  & 1   & 0 & True  & True  & N.A.          & \texttt{0:00:04:11} \\
J1J2-line             & (4,1) & ladder         & (2,1) & 0 (0 \%)  & 0   & 0 & True  & False & N.A.          & \texttt{0:00:00:09}  \\
\hline
J1J2-square           & (2,2) & square         & (2,2) & 0 (0 \%)  & 0   & 0 & True  & False & N.A.          & \texttt{0:00:48:48}  \\
triangular            & (2,2) & square         & (2,2) & 0 (0 \%)  & 0   & 0 & True  & False & N.A.          & \texttt{0:00:07:58}  \\
kagome                & (2,2) & square         & (4,3) & 1 (25 \%) & 0   & 0 & True  & False & N.A.          & \texttt{2:12:17:46*} \\
shuriken              & (1,1) & square         & (3,2) & 1 (25 \%) & 0   & 0 & True  & False & N.A.          & \texttt{0:00:26:03*} \\
snub-square           & (1,1) & square         & (2,2) & 1 (20 \%) & 0   & 0 & True  & False & N.A.          & \texttt{0:00:13:41}  \\
\hline
Rokhsar-Kivelson\_xyz & (1,1) & square         & (2,2) & 2 (11 \%) & 2   & 0 & True  & True  & N.A.          & \texttt{0:00:43:02} \\
Fermi-Hubbard         & (1,1) & plus-square    & (1,1) & 0 (0 \%)  & 0   & 0 & False & True  & N.A.          & \texttt{0:00:01:18} \\
kogut-susskind        & (1,1) & diamond-square & (2,2) & 9 (4 \%)  & 253 & 1 & True  & True  &            20 & \texttt{0:01:41:25} \\
\hline
\end{longtable}

\section{Benchmarking}
The main point of our method is its linear complexity in the patch size of the circuit for which a routing solution is sought, while still attaining locally optimal solutions, by leveraging the spatial periodicity of that circuit. As a result, it is expected to outperform other routing methods that do not exploit spatial periodicity, in both overhead and wall-clock time, for sufficiently large patch sizes. Beyond this scaling advantage, it is also worthwhile to assess the practical benefits of our method for today’s quantum chips -- ranging from a hundred to a thousand qubits -- compared to established approaches, in terms of both computational wall-clock time and solution quality.

\subsection{Method}

The routing problems we consider here are as follows.
\begin{itemize}
\item[-] \emph{J1J2-line.} A single Trotter step of the Trotterized time evolution of the J1J2 model on a linear chain, mapped to hardware with a connectivity graph also described by a chain (or a `line').

\item[-] \emph{J1J2-square.} A single Trotter step of the Trotterized time evolution of the J1J2 model on a square lattice, mapped to hardware with a connectivity graph described by a square lattice.
\end{itemize}

Note that here, we do not enforce logical qubits to return to their initial position (i.e., we do not enforce cyclic routing). So, here, the routed Trotter steps can be interpreted as routed circuits for a single first-order Trotter step or half a second-order Trotter step. From the latter, quantum simulation circuits with an arbitrary number of steps can be generated trivially (\sec{trotterization}).

For the two listed problems, we benchmark QuanTile against well-established methods. These methods do not support cyclic routing, nor do they allow SWAP merging (\sec{merge_swaps}). To ensure a fair comparison, we run QuanTile with these options disabled. Additionally, since the well-established methods we consider do not optimize gate order, we allow QuanTile to perform this optimization and then use the resulting gate order as a fixed input for these methods. Specifically, QuanTile is benchmarked against the following methods.

\begin{itemize}
  \item[-] \emph{Qiskit SabreSwap.} This method implements the routing approach from Refs.~\cite{li2019tackling, zou2024lightsabre}, as available in Qiskit 1.3.1. We used the highest optimization level (\lin{optimization_level = 3}). In previous work~\cite{kattemolle2023line}, we found SabreSwap to be the best-performing standard method in Qiskit for closely related routing problems.

  \item[-] \emph{Qiskit AIRouter.} This method (not among the standard methods in Qiskit mentioned before) follows the approach from Ref.~\cite{kremer2024practical}, implemented in Qiskit 1.3.3~\cite{qiskit2024quantum} (using \texttt{qiskit-ibm-transpiler} 0.10.1 and \texttt{qiskit-ibm-ai-local-transpiler 0.2.0}). We applied the highest available optimization levels (\lin{optimization_level = 3}, \lin{ai_optimization_level = 3}).

  \item[-] \emph{Cirq.} This method uses the standard \texttt{RouteCQC} routing method in Cirq 1.4.1~\cite{cirq2024cirq}. We set the lookahead radius to 32 (\lin{lookahead_radius = 32}).
\end{itemize}

The full benchmarking code is available in the repository~\cite{kattemolle2024quantile}. All benchmarks were performed on the same hardware as in \sec{results}, using a single thread for QuanTile and allowing the use of 12 threads for the other methods. The only exception is Qiskit AIRouter, which was run on a 2.6 GHz 6-Core Intel Core i7 with up to 12 threads due to compatibility issues preventing its execution in local mode on the Linux-based compute cluster at the time of benchmarking.

\clearpage

\subsection{Results}

\begin{figure}[b]
\includegraphics[width=.45\textwidth]{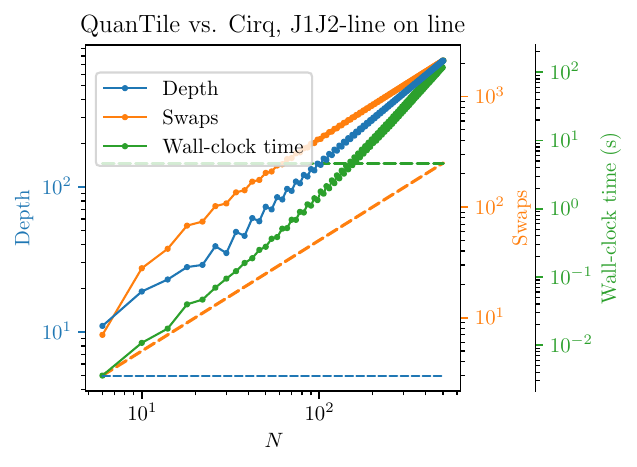}
\includegraphics[width=.45\textwidth]{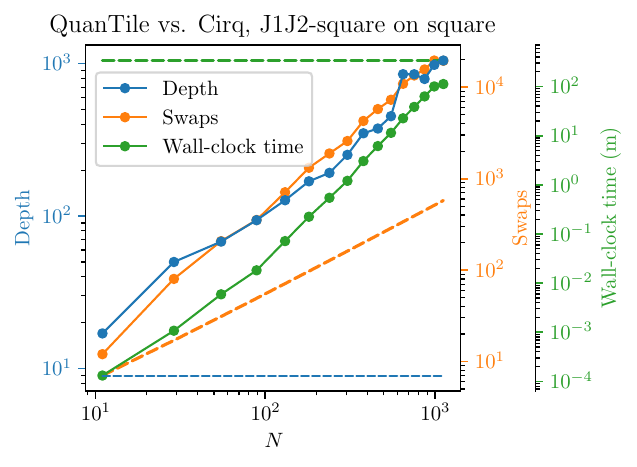}
\includegraphics[width=.47\textwidth]{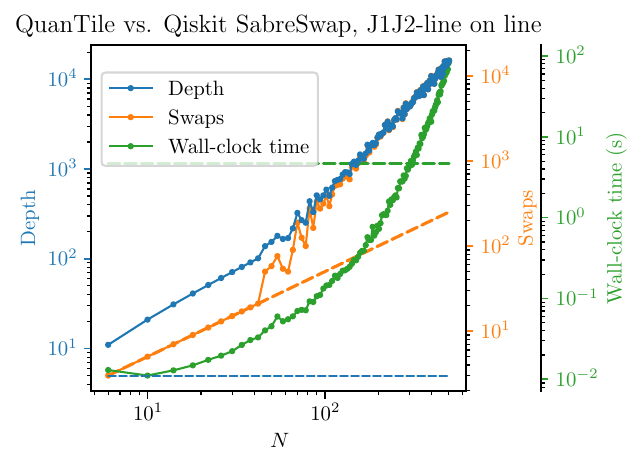}
\includegraphics[width=.45\textwidth]{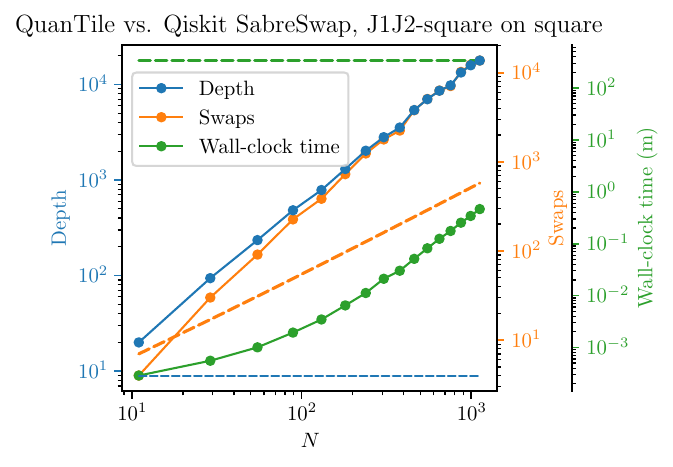}
\includegraphics[width=.45\textwidth]{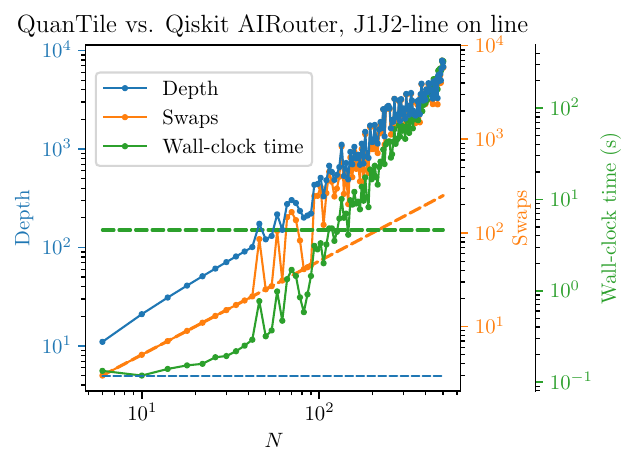}
\includegraphics[width=.45\textwidth]{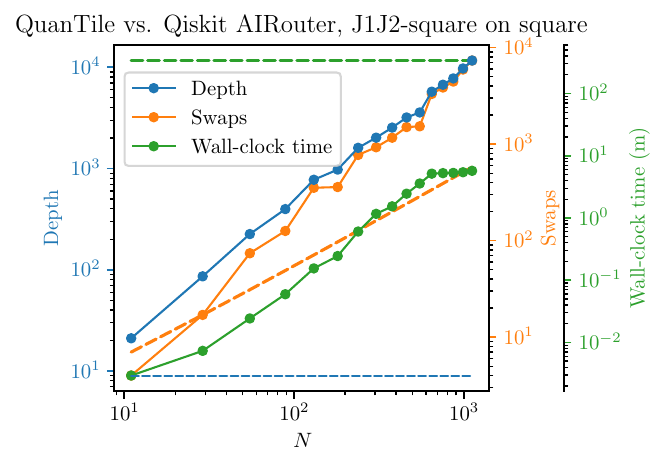}
\caption{\label{fig:benchmarking_plots_large}
Benchmarking results of QuanTile against Cirq (first row), Qiskit SabreSwap (second row), and Qiskit AIRouter (third row) for routing a single Trotter step in the quantum simulation of a two-local Hamiltonian. The benchmarks are performed on two interaction graphs: a chain with nearest and next-nearest neighbor edges (J1J2-line, first column) and a square lattice with nearest and next-nearest neighbor edges (J1J2-square, second column). The horizontal axes show the system size $N$ -- the number of qudits on which the Hamiltonian is defined.}
\end{figure}

\clearpage
\section{Basis graphs database}\label{sec:basis_graphs_database}
Here, we depict the basis graphs used in this work. They are defined in \lin{quantile/basis_graphs.json} in the implementation \cite{kattemolle2024quantile}.
Given an infinite lattice graph, the basis graph that induces it is not unique. So, for the circuits in \sec{results} to be fully unambiguous, it is necessary to give the basis graphs explicitly, which we do here. In \sec{results}, we used various basis graphs that generate the square lattice; they are all listed.

Each image depicts a patch of 5 by 5 basis graphs, with nodes represented using lattice graph notation. The central basis graph is highlighted with bold edges. Below each patch, the name of the corresponding basis graph is displayed. For the one-dimensional basis graphs, a patch of 5 by 1 basis graphs is shown.

We show the Archimedean lattices, the Laves lattices that are not also Archimedean, and miscellaneous other lattices. We gave a short introduction to these lattices in Ref.~\cite{kattemolle2024edge}.

\subsection{Archimedean}
\foreach \n in {
  square,
  plus-square,
  diamond-square,
  honeycomb,
  triangular,
  square-octagon,
  trellis,
  snub-square,
  kagome,
  bridge,
  ruby,
  star}
{
\placefig{figures/basis_graphs/\n.pdf}{\n}
}

\placefig{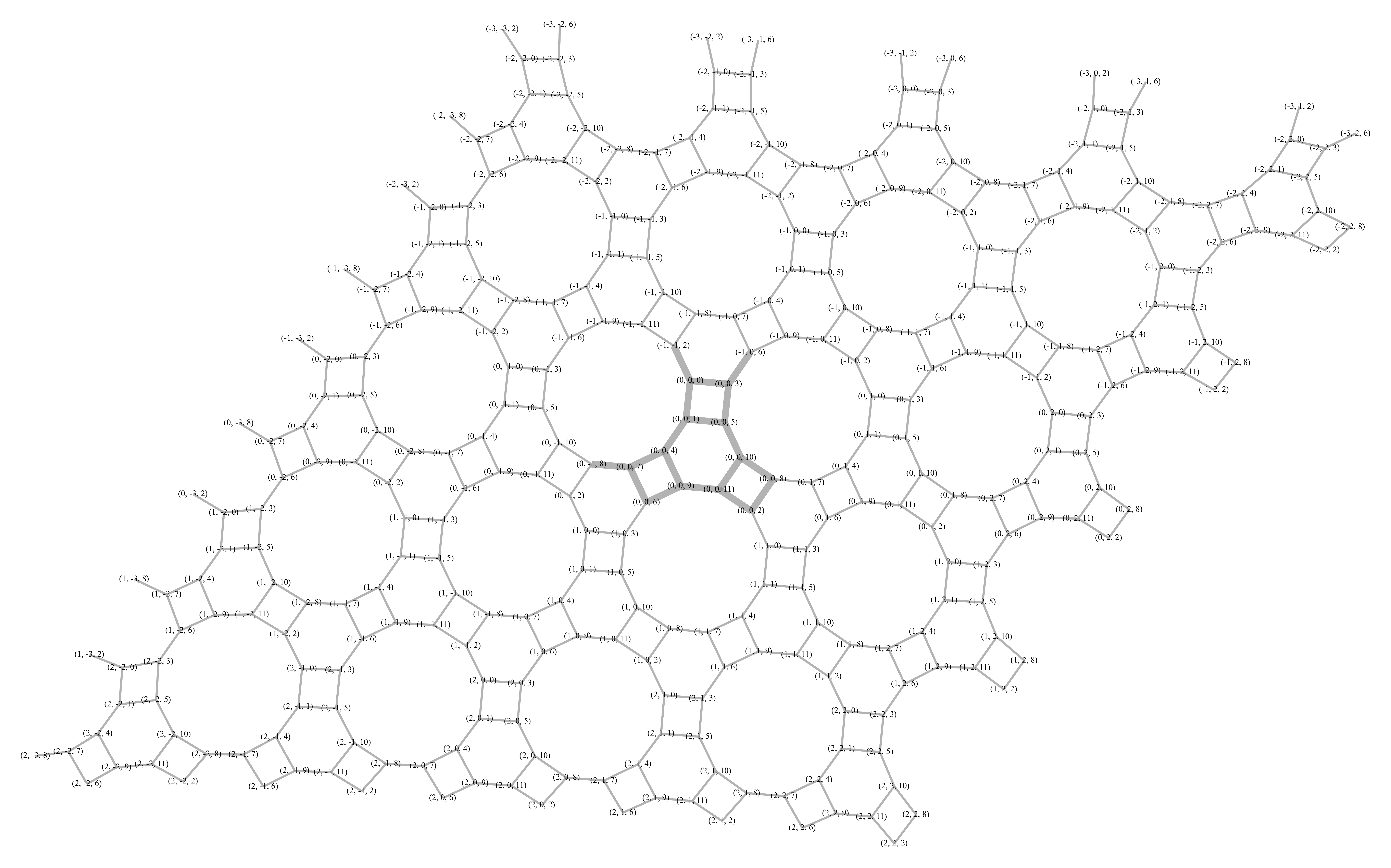}{cross}

\subsection{Laves}
\foreach \n in {
  union-jack,
  prismatic-pentagonal,
  cairo-pentagonal,
  dice,
  tetrille,
  asanoha,
  floret-pentagonal,
  kisrhombille}
{
\placefig{figures/basis_graphs/\n.pdf}{\n}
}
\subsection{Miscellaneous}
\vspace{-1em}
\foreach \n in {
ladder,
  line, shuriken,
  J1J2-ladder, J1J2-line, J1J2-square, J1J2J3-square,
  heavy-hex}
{
\placefig{figures/basis_graphs/\n.pdf}{\n}
}

\clearpage

\end{document}